\documentclass[11pt]{article}
\pdfoutput=1
\usepackage[ascii]{inputenc}
\usepackage[bookmarksnumbered,hypertexnames=false,colorlinks=false,linkcolor=blue,urlcolor=blue,citecolor=blue,anchorcolor=green,pdfusetitle]{hyperref}
\usepackage{bbm,braket,microtype,mathrsfs,amsmath,amssymb,amsthm,mathtools,fullpage,graphicx,enumitem,caption,overpic,caption,booktabs}
\usepackage[svgnames]{xcolor}
\usepackage[numbers,sort&compress]{natbib}
\usepackage[T1]{fontenc}
\usepackage[USenglish]{babel}
\usepackage{newpxtext}
\usepackage[euler-digits,euler-hat-accent]{eulervm}
\usepackage[capitalize]{cleveref}
\usepackage[noblocks]{authblk}
\numberwithin{equation}{section}
\allowdisplaybreaks[4]
\newtheorem{thm}{Theorem}[section]\crefname{thm}{Theorem}{Theorems}
\newtheorem{lem}[thm]{Lemma}\crefname{lem}{Lemma}{Lemmas}
\newtheorem{prop}[thm]{Proposition}\crefname{prop}{Proposition}{Propositions}
\crefname{cor}{Corollary}{Corollaries}
\theoremstyle{definition}
\newtheorem{dfn}[thm]{Definition}\crefname{def}{Definition}{Definitions}
\theoremstyle{remark}

\setlist[enumerate,1]{label=(\roman*)}

\DeclareMathOperator{\tr}{tr}
\DeclareMathOperator{\sgn}{sgn}
\DeclareMathOperator{\Span}{span}
\DeclareMathOperator{\per}{per}
\DeclareMathOperator{\maj}{maj}
\DeclareMathOperator{\chir}{chir}
\DeclareMathOperator{\MERA}{MERA}
\DeclareMathOperator{\IR}{IR}
\DeclareMathOperator{\UV}{UV}
\newcommand{\CC}{\mathbb C}
\newcommand{\ZZ}{\mathbb Z}
\newcommand{\RR}{\mathbb R}
\newcommand{\NN}{\mathbb N}
\newcommand{\HH}{\mathbb H}
\renewcommand{\SS}{\mathbb{S}^1}
\newcommand{\Htr}{\mathcal H}
\newcommand{\CAR}{\mathcal{A}_{\wedge}}
\newcommand{\CARsd}{\mathcal{A}_{\wedge}^{\mathrm{sd}}}
\newcommand{\fock}{\mathcal{F}_{\wedge}}
\newcommand{\ot}{\otimes}
\newcommand{\op}{\oplus}
\newcommand{\eps}{\varepsilon}
\newcommand{\id}{\mathbbm 1}
\newcommand{\one}{\mathbf 1}
\renewcommand{\d}{\ensuremath{\mathrm{d}}}
\newcommand{\cL}{\mathcal L}
\newcommand{\bigo}{\mathcal O}
\newcommand{\change}[1]{\ignorespaces}

\begin{document}

%=============================================================================
% \title{Quantum circuit approximations for the \\ Dirac field in 1+1 dimensions}
% \title{Entanglement renormalization quantum circuits \\ for the Dirac field in 1+1 dimensions}
\title{Quantum circuit approximations and entanglement renormalization for the Dirac field in 1+1 dimensions}
% \title{Entanglement renormalization for a free fermionic field} %  (using wavelets)
% \title{Entanglement renormalization for free-fermion CFT}
% \title{Entanglement renormalization for free-fermion CFT}
% \title{Entanglement renormalization quantum circuits for free-fermion CFT}
% \title{Local quantum circuits for free-fermion CFT}
% \title{Constructive entanglement renormalization \\ for the Dirac field in two dimensions}
\author[1]{Freek Witteveen}
\author[2]{Volkher Scholz}
\author[3]{Brian Swingle}
\author[1,4]{Michael Walter}
\affil[1]{Korteweg-de Vries Institute for Mathematics and QuSoft, University of Amsterdam}
\affil[2]{Department of Physics, Ghent University}
\affil[3]{Condensed Matter Theory Center, Maryland Center for Fundamental Physics, Joint Center for Quantum Information and Computer Science, and Department of Physics, University of Maryland}
\affil[4]{Institute for Theoretical Physics, Institute for Language, Logic, and Computation, University of Amsterdam}
\date{}
\maketitle
%=============================================================================
\begin{abstract}
The multiscale entanglement renormalization ansatz describes quantum many-body states by a hierarchical entanglement structure organized by length scale.
Numerically, it has been demonstrated to capture critical lattice models and the data of the corresponding conformal field theories with high accuracy.
However, a rigorous understanding of its success and precise relation to the continuum is still lacking.
To address this challenge, we provide an explicit construction of entanglement-renormalization quantum circuits that rigorously approximate correlation functions of the massless Dirac conformal field theory.
We directly target the continuum theory: discreteness is introduced by our choice of how to probe the system, not by any underlying short-distance lattice regulator.
To achieve this, we use multiresolution analysis from wavelet theory to obtain an approximation scheme and to implement entanglement renormalization in a natural way.
This could be a starting point for constructing quantum circuit approximations for more general conformal field theories.
% is a real-space renormalization scheme that
% While our most fundamental theories of nature are formulated as quantum field theories, quantum information is generally studied in terms of quantum bits and circuits.
% Our circuits implement entanglement renormalization...
\end{abstract}
%=============================================================================
\tableofcontents
%=============================================================================

%=============================================================================
\section{Introduction}
%=============================================================================

%-----------------------------------------------------------------------------
% \subsection{Motivation}
%-----------------------------------------------------------------------------
Quantum information theory is generally formulated in terms of discrete quantum bits and quantum circuits.
However, our most fundamental theories of nature are formulated as quantum field theories, and it is a physical and mathematical challenge to understand the role of quantum information in such continuum theories.
In this work, we bridge these two paradigms for the case of a free massless Dirac field in 1+1 dimensions and show how to rigorously represent its entanglement structure through a quantum circuit.
% The massless free fermion field is perhaps the most basic example of a conformal field theory, and the form of our circuits are based on the scale invariance of the model.
% We believe that our results will serve as a starting point for understanding quantum information in (and deriving quantum circuits for) more general conformal field theories.

%-----------------------------------------------------------------------------
% \subsection{Tensor networks and quantum field theory}
%-----------------------------------------------------------------------------
Quantum circuits are examples of tensor networks, which parameterize quantum many-body states with a relatively small number of parameters by restricting the allowed entanglement structure. Tensor networks have been very successful for studying discrete quantum systems~\cite{orus2014practical}.
Several approaches have been proposed to extend the notion of a quantum circuit, or more generally of a tensor network, to quantum field theories.
Roughly speaking there are two distinct routes: one is to define a variational class of continuum states, whereas the other is to consider a restricted set of observables and try to approximate correlation functions of these observables.

An example of the former is cMERA~\cite{haegeman2013entanglement}, which defines a class of states that arise from a real-space renormalization procedure.
In this case the `quantum circuit' that performs the entanglement renormalization is also continuous.
Another example is cMPS~\cite{verstraete2010continuous}, which can be interpreted as a path integral~\cite{brockt2012continuum}.
Both approaches have been successfully demonstrated numerically for free theories, and these classes of states have also been used as a basis for perturbation theory~\cite{cotler2019entanglement} and variational algorithms~\cite{haegeman2010applying} for 1+1 dimensional quantum field theories.
\change{In particular, cMPS can be used as a variational ansatz to study interacting field theories at very high precision, see for instance~\cite{verstraete2010continuous,ganahl2017continuous}.
Yet rigorous proofs have largely been elusive.}

In this paper, we follow the second route, by considering correlation functions of smeared operators.
These operators are discretized at an appropriate scale and an ordinary quantum circuit circuit is used to prepare a state with which to compute their correlation functions.
This means that the discreteness in our description arises not from the system itself, but in our choice of how to probe the system.

The circuits that we derive fit in the Multi-scale Entanglement Renormalization Ansatz (MERA)~\cite{vidal2007entanglement,vidal2008class}, a tensor network ansatz designed for systems with scale invariance that implements a kind of real-space renormalization.
A MERA tensor network prepares a quantum many-body state through a series of layers, each of which consists of isometries followed by local unitary transformations.
If we apply the circuit in reverse, the latter disentangle local degrees of freedom and the former coarse-grain the system by a factor of two.
For a scale-invariant theory, each of these layers can be taken identical, and it has been demonstrated numerically for some paradigmatic Hamiltonians that \change{for a variationally optimized scale-invariant MERA} the conformal data of the limiting theory, such as the scaling dimensions and operator product expansion (OPE) coefficients, can be extracted from the scaling superoperator corresponding to a single network layer~\cite{evenbly2013quantum}.

Tensor networks have to a large extent been developed as a method to efficiently simulate quantum systems on a classical computer.
However, evaluating correlation functions for a MERA tensor network can still be very costly, with the computational cost scaling as a high power of the number of parameters.
If one extends the MERA to a quantum circuit, it can be simulated efficiently on a quantum computer provided the complexity of each layer is not too large. It has been argued that the structure of entanglement renormalization may be relatively insensitive to small errors and that many models of physical interest have layers of low complexity, thus it may be a useful circuit model for quantum computers to simulate quantum systems at or away from criticality~\cite{kim2017robust}.
In this regard, our results provide additional evidence that tensor networks are a promising application of noisy quantum computers, as we now also have the possibility to address continuum theories.
\change{Of course, the free Dirac field in 1+1~dimensions is exactly solvable, and our analytically constructed circuits provide no computational speed-up.
However, while the exact solution for a free theory is straightforward in momentum space, the construction of entanglement renormalization circuits is more complicated in the free setting, and closely related to wavelet theory \cite{evenbly2016entanglement,haegeman2018rigorous,witteveen2020wavelet}.}

A final motivation to investigate tensor networks for conformal field theories is provided by the wish to study holography (a duality between two quantum theories, one in~$d$ dimensions and one in~$d+1$ dimensions).
The main example is provided by the AdS/CFT correspondence, a conjectural relation between quantum gravity on an AdS space with a conformal field theory on its conformal boundary~\cite{maldacena1999large}.
It has been remarked that entanglement renormalization has a structure reminiscent of this duality~\cite{swingle2012entanglement}, as the circuit reorganizes a critical one-dimensional system to a two-dimensional structure that is a discretization of AdS space, although the precise connection to holographic theories is still being developed~\cite{bao2015consistency,milsted2018geometric}.
Any MERA tensor network can be extended to a unitary quantum circuit by extending the isometries to unitaries with an auxiliary input, so that the MERA is recovered by applying the circuit to an appropriate product state.
Such extensions are not unique.
In contrast, our construction naturally yields a unitary quantum circuit that reorganizes the degrees of freedom of the Dirac theory in one higher dimension, by position and scale, cleanly separating positive and negative energy modes of the Dirac fermion.
Thus it can be seen as a circuit realization of a holographic mapping for an actual conformal field theory, complementing tensor network toy models of holographic mappings as proposed in~\cite{pastawski2015holographic,yang2016bidirectional,hayden2016holographic,nezami2016multipartite}.
\change{However, since the Dirac fermion is not an example of a conformal field theory with an AdS gravity dual, this should merely be seen as an analogy, without direct application to AdS/CFT.}

\subsection{Prior work}\label{sec:prior work}
\subsection{Summary of results}
%-----------------------------------------------------------------------------
We now describe our main results.
The model that we consider is the free massless Dirac fermion in 1+1 dimensions, with action
\begin{align*}
  S(\Psi) = \frac{1}{2} \int \Psi^\dagger \gamma^0 \gamma^{\mu} \partial_{\mu}  \Psi \,  \d x \d t
\end{align*}
for a two-component complex fermionic field~$\Psi$ on the line (or on a circle).
The usual second quantization procedure shows that the fields have correlation function
\begin{align*}
  \langle \Psi^\dagger(x) \Psi(y) \rangle = \frac{1}{x - y}.
\end{align*}
The stress energy tensor is a normal-ordered product of the fields and its derivatives. In complex coordinates~$z = x + it$ and $\overline{z} = x - it$, the stress-energy tensor has a holomorphic $T_{zz}$ component
for which one may deduce that
\begin{align*}
  \langle T_{zz}(x) T_{zz}(y) \rangle = \frac{1/2}{(x - y)^4}
\end{align*}
and hence the theory has central charge~$c =1$. For details from the conformal field theory point of view, see~\cite{francesco2012conformal}.
\change{One approach to a mathematical formulation of quantum field theories is the algebraic approach, which allows a rigorous description of certain quantum field theories.
This is in particular the case for free field theories, such as the Dirac fermion.}
We will briefly review the algebraic approach to the Dirac fermion in \cref{sec:prelim}. In this approach, in order to have well-behaved operators, one usually `smears' the fields. That is, for some function $f$ one defines
\begin{align*}
  \Psi(f) = \int f(x) \Psi(x) \d x.
\end{align*}
From a physical perspective the smearing function is justified by the fact that one can only probe the system at some finite scale.

We will now describe a procedure which approximates correlation functions of smeared operators.
Informally, the procedure is that we first discretize the operators at some scale (i.e., we impose a UV cut-off), and then, in order to obtain the free fermion vacuum, we need to `fill the Dirac sea' up to the relevant scale.
So, the circuit, starting from the Fock vacuum, has to fill all the negative energy modes over the range of scales that are relevant for the inserted operators, directly analogous to a real-space renormalization procedure.
We know the negative energy states explicitly in Fourier space, but the non-trivial problem is that we want to construct a \emph{local} circuit, while the Fourier basis for the negative energy solutions is very non-local.
In order to obtain a circuit that is compatible with scale invariance and translation invariance, but is still local, we are led to search for a \emph{wavelet} basis for the space of negative energy solutions.
It is not possible to construct a basis that is both completely local and consists of exactly negative energy solutions, but it turns out it is approximately possible by using a pair of wavelets that approximately satisfy a certain phase relation.
Such pairs of wavelets, called \emph{approximate Hilbert pairs} have already been constructed for other purposes~\cite{selesnick2002design}, and as discussed in \cref{sec:prior work} these are closely related to the construction of (approximate) ground states for critical free fermions.
This construction takes as input two integer parameters $K$ and $L$, such that the support of the wavelet is of size~$2(K+L)$, and there is an approximation parameter $\eps$ which measures how accurately the phase relation is satisfied.
The wavelet functions give rise to a `classical' circuit, which implements the decomposition of a function in the wavelet basis at different scales.
This circuit should be thought of as a circuit on the single-particle level, and the fermionic quantum circuit is obtained as its second quantization.

Now let $\{O_i\}$, $i = 1, \ldots, n$ be a set of smeared operators that are either linear in the fields or normal-ordered quadratic operators, and which are compactly supported.
We denote the correlation functions by
\begin{align}\label{eq:cor fun intro}
  G(\{ O_i \}) = \langle O_1 \cdots O_n \rangle.
\end{align}
The procedure sketched above discretizes the operators~$O_i$ and constructs a quantum circuit that computes an approximation~$G^{\MERA}_{\cL, \eps}(\{ O_i \})$ of the correlation function, where $\cL$ is the number of layers of the circuit, and $\eps$ is an error parameter.
The structure of the circuit is illustrated in \cref{fig:circuit_intro} (both for the line and circle).

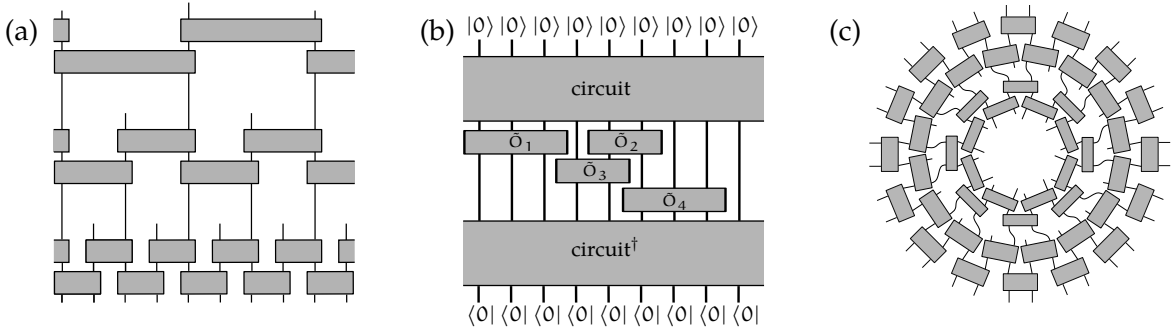
\begin{figure}
\centering
(a)\;\raisebox{-3.5cm}{
\begin{overpic}[width=4.0cm]{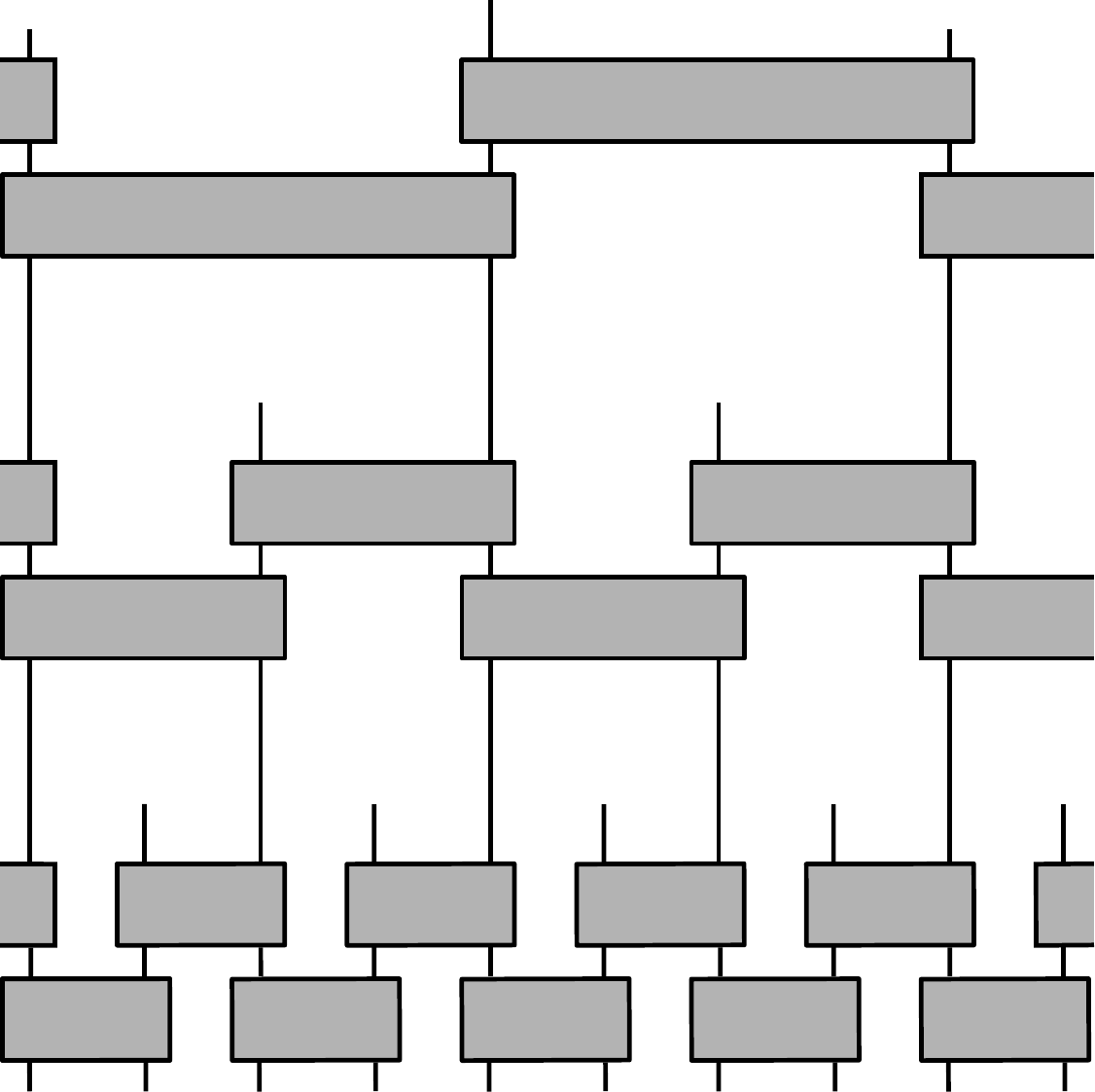}\end{overpic}}
\qquad
(b)\;\raisebox{-3.5cm}{\scriptsize
\begin{overpic}[width=4.0cm,height=3.5cm]{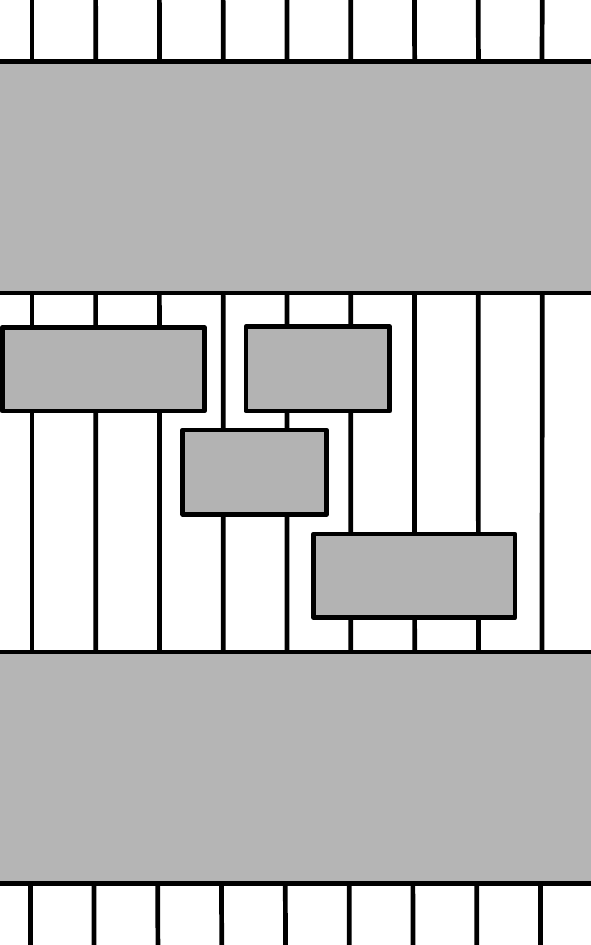}
\put(1,90){$\ket{0}$} \put(12,90){$\ket{0}$} \put(23,90){$\ket{0}$} \put(34,90){$\ket{0}$} \put(45,90){$\ket{0}$} \put(56,90){$\ket{0}$} \put(67,90){$\ket{0}$} \put(78,90){$\ket{0}$} \put(89,90){$\ket{0}$}
\put(36,69){$\text{circuit}$}
\put(15,51.7){\tiny $\tilde{O}_1$}
\put(50,51.7){\tiny $\tilde{O}_2$}
\put(39.5,42){\tiny $\tilde{O}_3$}
\put(66,32){\tiny $\tilde{O}_4$}
\put(36,15.5){$\text{circuit}^{\dagger}$}
\put(1,-6){$\bra{0}$} \put(12,-6){$\bra{0}$} \put(23,-6){$\bra{0}$} \put(34,-6){$\bra{0}$} \put(45,-6){$\bra{0}$} \put(56,-6){$\bra{0}$} \put(67,-6){$\bra{0}$} \put(78,-6){$\bra{0}$} \put(89,-6){$\bra{0}$}
\end{overpic}}
\qquad
(c)\;\raisebox{-3.5cm}{\includegraphics[width=4cm]{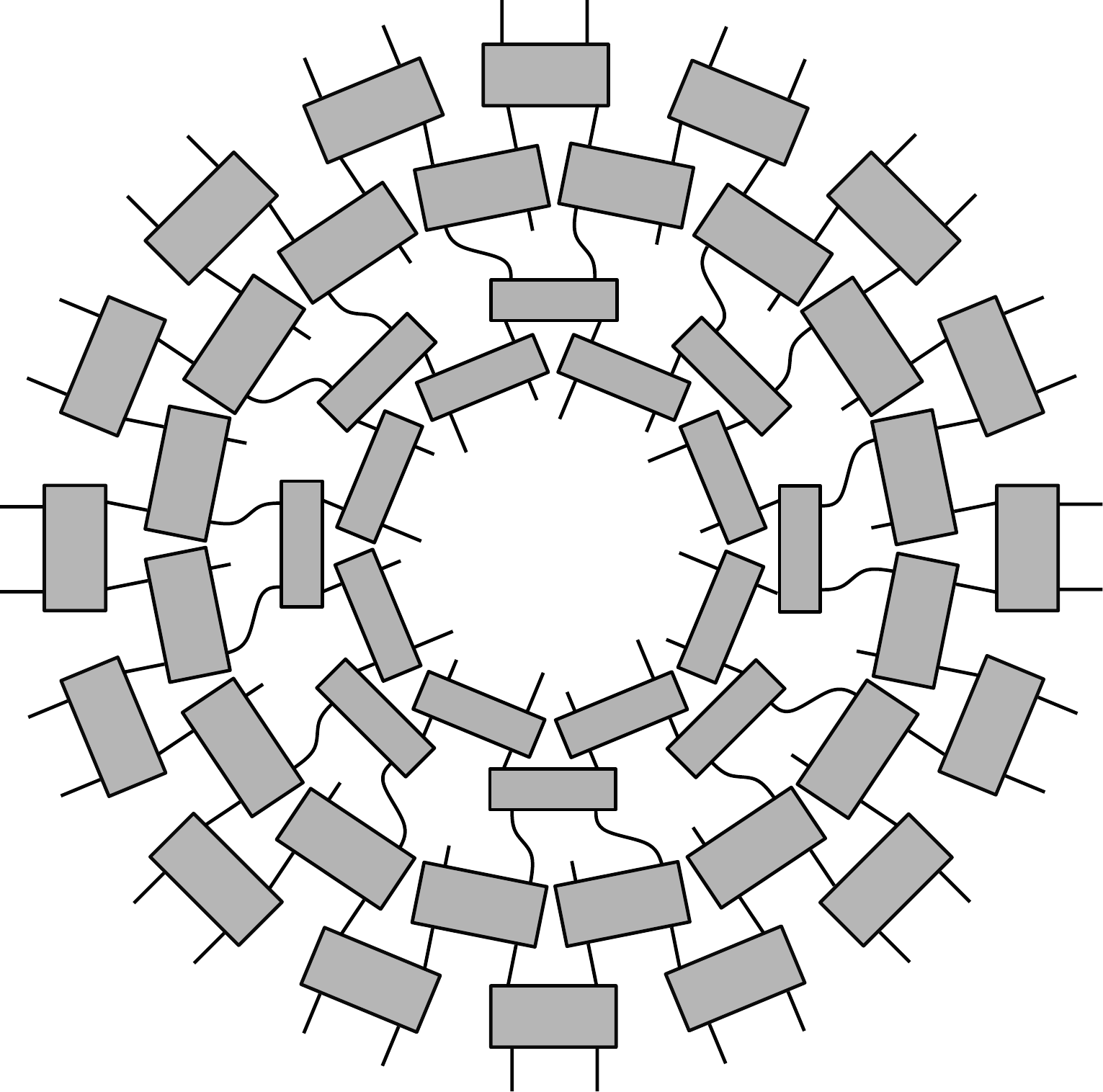}} \\[0.1cm]
\caption{(a)~Structure of the circuit for the Dirac fermion on the line. Each layer is an identical local quantum circuit of depth~$K+L$.
(b)~Correlation functions~\eqref{eq:cor fun intro} are computed as expectation values of discretizations $\tilde{O}_i$ of the operators $O_i$.
(c)~Periodized circuit for the Dirac fermion on the circle.}
\label{fig:circuit_intro}
\end{figure}

The following is a simplified version of our main result.
A precise formulation is given by \cref{thm:approximation}, where we also specify precisely which operators we consider and give explicit bounds for the approximation error.
We assume that we are given a family of wavelet filters with uniformly bounded scaling functions, of support $M$ and approximating the Hilbert pair relation to accuracy $\eps$.
The constructed circuits have depth $D = \lceil \tfrac{M}{2}\rceil$ for a single circuit layer, and the bond dimension of the corresponding MERA tensor network is given by $\chi = 2^{D}$.

\begin{thm}[Informal]\label{thm:introduction}
Let $O_1,\dots,O_n$ be a collection of Dirac field creation or annihilation operators or normal-ordered quadratic operators with compact support and smeared by a differentiable function.
Then the approximation error is bounded by
\begin{align*}
  \lvert G(\{ O_i \}) - G^{\MERA}_{\cL, \eps}(\{ O_i \} \rvert = \bigo(M^2 2^{-\frac{\cL}{3}}) +  \bigo(\eps \log \tfrac{M}{\eps}).
\end{align*}
The constants in the $\bigo$-notation depend on~$n$ and the support and smoothness of the $O_i$.
\end{thm}

Our main theorem provides a justification for the numerical success of MERA for quantum field theories by providing rigorous bounds on the approximation of correlation functions.
To illustrate our result, we show the precise error bounds obtained for a two-point function in \cref{fig:approximation error}.
The error bounds in \cref{thm:introduction} are invariant under rescaling (which is of course a desirable property for a scale invariant theory).
A Dirac fermion can be decomposed into two Majorana fermions.
Our construction is compatible with this decomposition, so we also obtain quantum circuits for Majorana fermions.

\change{There is an extensive body of literature on approximate Hilbert pairs, see for instance \cite{selesnick2005dual, yu2005hilbert, chaudhury2009construction, chaudhury2010hilbert}.
It remains an open problem to prove analytic bounds on the decay of $\eps$ with $M$, see \cite{achard2017new} for recent results in this direction.
For the construction of \cite{selesnick2002design} where $M = 2(K+L)$, numerically the parameter~$\eps$ is seen to decrease exponentially with $\min\{K,L\}$ (see \cref{tab:constants}).}

\begin{figure}
\begin{center}
\includegraphics[width=0.5\linewidth]{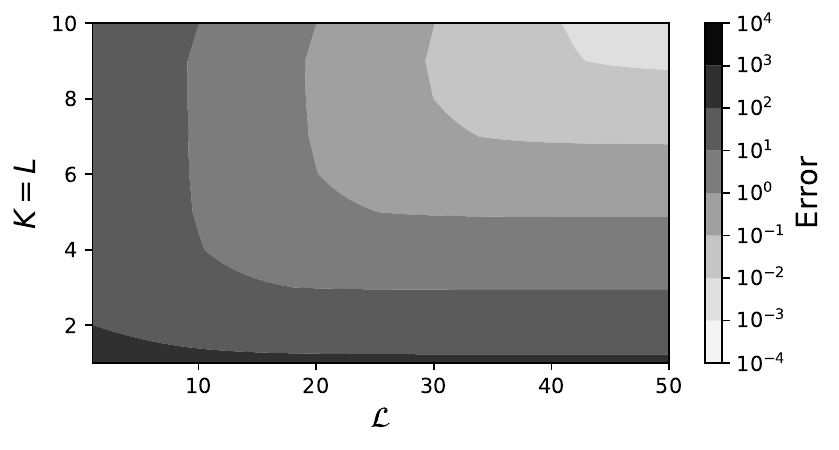}
\caption{The error bound from \cref{thm:introduction} (\cref{thm:approximation}) illustrated for a two-point function.
It is obtained by evaluating \cref{eq:sharp constant} using \cref{tab:constants} for an approximate Hilbert pair with parameters~$K=L$.
The smearing functions are taken to be translates of a function $f$ with $\lVert f \rVert = 1$ and optimal trade-off between smoothness and support (that is, $D = \sqrt{2}$ in the formulation of \cref{thm:approximation}).}
\label{fig:approximation error}
\end{center}
\end{figure}

Note that our construction gives rise to a circuit rather than a MERA tensor network in a canonical way, and that the bond dimension is exponential in the circuit depth, so this provides a potential starting point for investigating quantum algorithms for quantum field theory correlation functions.
One of the interesting features of 1+1 dimensional conformal field theories is that they have many symmetries.
Discretizing the theory necessarily breaks these symmetries.
However, we find that spatial translation, time translation and rescaling by a factor two all have natural implementations on the MERA (where rescaling by a factor 2 is precisely implemented by a single circuit layer).
% Apart from these global symmetries we can also perform piecewise-constant scaling transformations.

% \todo{A wavelet/MERA dictionary, e.g., constant-depth quantum circuit <-> constant-depth linear classical circuit = band matrix, many-body ground state = symbol, ... (Freek: this would be nice, but I don't really know what to put in there that makes sense and provides some insight... )}

\subsubsection*{Numerical examples}
Since the circuits we obtain are the second quantization of a single-particle circuit we can simulate them for high circuit depth (bond dimension).
In \cref{fig:numerics}, (a) and (b) we show approximations to the smeared two-point functions for the fermionic fields and for the stress-energy tensor for $K = L = 1$ and $K=L =3$ in the wavelet construction, corresponding to MERA tensor networks with bond dimensions $\chi = 4$ and $\chi = 64$ respectively.
Another statistic is the entanglement entropy of an interval.
In order to define this one needs a cut-off, for which we use the wavelet discretization.
The Cardy formula~\cite{cardy1986operator} predicts that for a conformal field theory the entanglement entropy of an interval scales as $S_E = \frac{c}{3}\log(L) + c'$ where $c$ is the central charge, $L$ is the size of the interval, and $c'$ a non-universal constant depending on the cut-off.
In \cref{fig:numerics},~(c) we have plotted the entanglement entropies obtained from our construction (for the same wavelets).
For $K = L = 3$ the agreement with the Cardy formula for $c =1$ is already very accurate.
As another numerical illustration of the accuracy of approximation, one may compute eigenvalues of the entanglement renormalization superoperator and extract scaling dimensions of the conformal field theory from its eigenvalues~\cite{pfeifer2009entanglement}.
One way to do so is by applying a Jordan-Wigner transformation to the circuit for the Majorana fermion to obtain a (matchgate) circuit for the Ising model, following the procedure in \cite{evenbly2016entanglement}.
The results are illustrated in \cref{tab:scaling dimensions}.

\begin{figure}
\centering
(a)\raisebox{-5.8cm}{\includegraphics[width=0.28\linewidth]{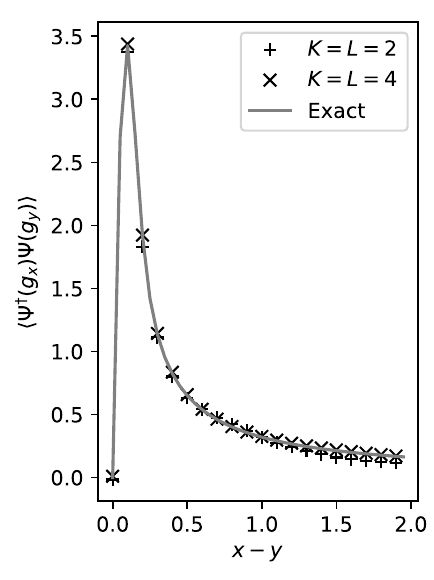}}
\quad
(b)\raisebox{-5.8cm}{\includegraphics[width=0.28\linewidth]{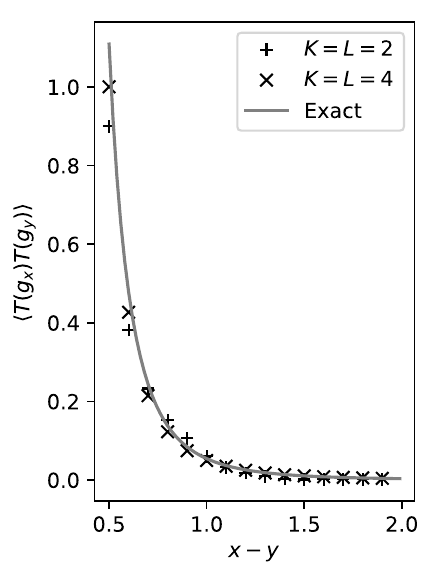}}
\quad
(c)\raisebox{-5.8cm}{\includegraphics[width=0.28\linewidth]{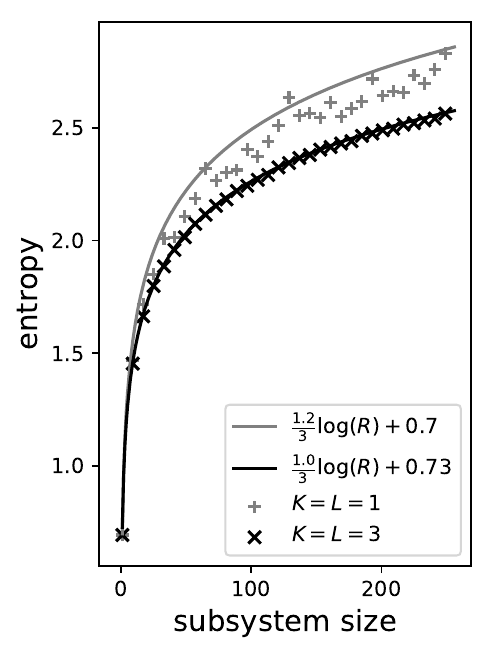}}
\caption{(a) Correlation function $\langle \Psi^{\dagger}(g_x) \Psi(g_y)\rangle$ evaluated using our approximate quantum circuits.
The smearing functions $g_x$, $g_y$ are Gaussians with standard deviation~$\sigma = 0.05$ peaked at $x$ and $y$, respectively.
(b) Correlation functions $\braket{T(g_x)T(g_y)}$ evaluated using our approximation quantum circuits.
The stress-energy tensor is smeared in both space and time; see \cref{subsec:approx} for details.
% The field two-point function scales as $\langle \Psi^{\dagger}(g_x) \Psi(g_y)\rangle \sim \frac{1}{x-y}$, while for the stress-energy two-point function the result scales as $\langle T(g_x)T(g_y) \rangle \sim \frac{c/2}{(x-y)^4}$.
(c)~Subsystem entropies for the corresponding quantum states. % defined by our wavelet approximations.
The logarithmic fits show that we obtain excellent agreement with the Cardy formula for central charge~$c=1$ already for $K+L=6$.}
\label{fig:numerics}
\end{figure}

\begin{table}
\centering
\begin{tabular}{rrrrrr}
\toprule
& $\chi$ & $E$ & $\Delta E / E$ & $\Delta_{\sigma}$ & $\Delta_{\mu}$ \\
\midrule
Exact && $- \frac{4}{\pi}$ && 0.125 & 0.125 \\
\midrule
$K = 1$,	$L = 1$ &	$2$ &	-1.2560	& 0.0135 &	0.0968	& 0.1696 \\
$K = 1$,	$L = 2$ &	$4$ &	-1.2705	& 0.0021 & 0.1360	& 0.1173 \\
$K = 2$,	$L = 1$ & $4$ &	-1.2630	& 0.0081 & 0.1031	& 0.1563 \\
$K = 1$,	$L = 3$ & $8$ &	-1.2727	& 0.0005 & 0.1226	& 0.1283 \\
$K = 2$,	$L = 2$ & $8$ &	-1.2722	& 0.0008 & 0.1310	& 0.1204 \\
$K = 3$,	$L = 1$ & $8$ &	-1.2655	& 0.0061 & 0.1052	& 0.1522 \\
$K = 1$,	$L = 4$ & $16$ &	-1.2731	& 0.0001 & 0.1261	& 0.1242 \\
$K = 2$,	$L = 3$ & $16$ &	-1.2731	& 0.0001 & 0.1238	& 0.1264 \\
\bottomrule
\end{tabular}
\caption{Values of the Majorana fermion energy density $E$, the relative error in energy density $(E + \frac{\pi}{2}) / E$ and the scaling dimensions $\Delta_{\sigma}$ and $\Delta_{\mu}$ for the Majorana CFT (or, equivalently, of the Ising CFT). Some other scaling dimensions, in particular those of the fermion fields themselves, are exactly reproduced because of the structure of the wavelet transform, as discussed in \cref{subsec:symmetries}.
These values were computed by first decoupling the circuit obtained from an approximate Hilbert pair of wavelets for the Dirac fermion into two circuits for Majorana fermions, and then taking a Jordan-Wigner transform.
This yields an entanglement renormalization circuit for the Ising model, in which the spin and disorder fields $\sigma$ and $\mu$ are local.
For details on this procedure, see \cite{evenbly2016entanglement}.}\label{tab:scaling dimensions}
\end{table}

%-----------------------------------------------------------------------------
\subsection{Outlook}
%-----------------------------------------------------------------------------
As mentioned, our quantum circuits implement a `holographic mapping' for the Dirac conformal field theory.
This opens up the possibility to study many interesting questions on how quantum information is organized by such mappings, e.g., in terms of quantum error correcting properties~\cite{kim2017entanglement}.
% Since the free fermion is not an example of a theory with an AdS gravity dual (as it does not have a large central charge), this raises the question to which extent generalized holographic dualities exist for any~$1+1$~dimensional conformal field theory.

% It would also be interesting to extend our work to obtain quantum circuit descriptions of purified thermal states (the so-called thermofield double states).
% For holographic theories, such states are believed to be dual to wormholes connecting asymptotically AdS regions~\cite{maldacena2003eternal}.

In future work we hope to construct entanglement renormalization circuits for more general classes of conformal field theories.
\change{Recently, inspired by the present work, an analogous construction for bosonic free models has been worked out in~\cite{witteveen2020wavelet}.}
A challenging open problem is to extend the relation between wavelet analysis and quantum circuits for conformal field theories to interacting models.
It is not at all clear that this is possible, but a natural starting point could be Wess-Zumino-Witten theories, as many of these can be constructed algebraically as symmetries on a finite number of free massless fermions~\cite{fuchs1995affine}.
\change{The algebraic construction of these theories is closely related to the representation theory of loop groups \cite{wassermann1998operator}, and one starting point could be to revisit this analysis of the loop group in terms of wavelet theory.}

Another direction would be to investigate entanglement renormalization from the perspective of vertex algebras. A recent attempt to discretize vertex algebras to a spin chain model, with a view towards quantum computer simulation of conformal field theories can be found in~\cite{zini2018conformal}. \change{For MPS tensor networks it has been shown that they are sufficiently expressive to compute correlation functions for a very general class of conformal field theories using vertex algebra techniques \cite{konig2016matrix,konig2017matrix}.}

From a computational point of view it would be interesting to investigate whether a wavelet circuit can serve as a starting point for perturbation theory, and get faster convergence of MERA optimization algorithms.
\change{It would be interesting to develop the idea of \cite{kim2017robust} to use variational entanglement renormalization circuits to the continuous setting, and to compare with other approaches proposed to compute correlation functions using quantum computers (see, e.g.,~\cite{preskill2018quantum}).

Finally, there is the open problem of determining the existence of a family of approximate Hilbert pair wavelets where the approximation error decreases sufficiently fast with the support of the wavelets, as discussed above~\cite{selesnick2002design,achard2017new}.}

%-----------------------------------------------------------------------------
\subsection{Plan of the paper}
%-----------------------------------------------------------------------------

The remainder of this work is structured as follows: in \cref{sec:prelim} we recall the algebraic approach to fermionic systems and quasi-free states, in \cref{sec:wavelets} we review wavelet theory and collect some useful estimates.
\cref{sec:quantumcirc} contains the main results, first we derive a wavelet approximation to the free fermion, then use it to construct a quantum circuit and we prove a bound on the approximation error.
We also discuss how to implement certain conformal symmetries with the circuit and remark on the possibility of reproducing conformal data.

%-----------------------------------------------------------------------------
\subsection{Notation and conventions}
%-----------------------------------------------------------------------------
Given a Hilbert space~$\HH$, we write $\braket{\cdot,\cdot}$ for the inner product and $\lVert\cdot\rVert$ for the norm of vectors.
We denote by~$B(\HH)$ the space of bounded operators on~$\HH$ and the operator norm of an operator~$A$ by~$\lVert A\rVert$.
We denote Hermitian adjoints by $A^\dagger$, and we write $A \leq A'$ if the difference~$A'-A$ is positive semidefinite.
We denote identity operators by~$\id_\HH$ and leave out the subscript if the Hilbert space is clear from the context.
If $A$ is Hilbert-Schmidt then we write $\lVert A\rVert_2 = \sqrt{\tr[A^\dagger A]}$ for the Hilbert-Schmidt norm.
For the finite dimensional Hilbert space~$\CC^n$, we use bra-ket notation and write $\ket{0}, \dots , \ket{n-1}$ for the standard basis.
We define the circle~$\SS=\RR/\ZZ$ as the interval~$[0,1]$ with endpoints identified.
We write $L^2(\RR)$, $L^2(\SS)$, etc.~for Hilbert spaces of square-integrable functions, equipped with the Lebesgue measure that assigns unit measure to unit intervals, and we denote by~$\ell^2(\ZZ)$ the Hilbert space of square-integrable sequences.
The Fourier transform of a function~$\phi\in L^2(\RR)$ is denoted by~$\hat\phi\in L^2(\RR)$ and is given by~$\hat\phi(\omega)=\int_{-\infty}^\infty f(x) e^{-ix\omega} dx$ if~$\phi$ is absolutely integrable.
Similarly, the Fourier transform of a function~$\phi\in L^2(\SS)$ is denoted by~$\hat\phi\in\ell^2(\ZZ)$ and can be computed as~$\hat\phi[n]=\int_0^1 f(x) e^{-ix2\pi n} dx$.
We define the Fourier transform of a sequence~$f\in\ell^2(\ZZ)$ to be the~$2\pi$-periodic function $\hat f\in L^2(\RR/2\pi\ZZ)$ given by~$\hat f(\theta) = \sum_{n\in\ZZ} f[n] e^{-i\theta n}$.
\change{Lastly we define the Fourier transform of a sequence~$f\in\ell^2(\ZZ/m\ZZ)$ to be the function on $\ZZ/m\ZZ$ given by $\hat f[n] = \sum_{k \in \ZZ/m\ZZ} f(k) e^{-2\pi i k n / m}$.}
For $\HH=L^2(\RR)$, $L^2(\SS)$, $\ell^2(\ZZ)$ or $\ell^2(\ZZ/m\ZZ)$, and $\lambda\in\HH$, we will denote by $m(\hat{\lambda})$ the \emph{Fourier multiplier} with symbol~$\hat{\lambda}$, defined by multiplication with~$\hat\lambda$ in the Fourier domain (equivalently, convolution with~$\lambda$ in the original domain).
On $\ell^2(\ZZ)$, we define the \emph{downsampling} operator~$\downarrow$  by $(\downarrow f) [n] = f [2n]$; its adjoint is the \emph{upsampling} operator~$\uparrow$ given by~$(\uparrow f) [2n] = f [n]$ and~$(\uparrow f) [2n+1] = 0$ for $f\in\ell^2(\ZZ)$.
We will also use the Sobolev spaces~$H^K(\RR)$ and $H^K(\SS)$, which consist of functions that have a square-integrable weak~$K$-th derivative, denoted~$f^{(K)}$.
All $p$-norms for~$p\neq2$ will be denoted by $\lVert f\rVert_p$.
We write~$\one$ for the constant function equal to one, and~$\one_X$ for the indicator function of a set~$X$.
If~$f\in L^2(\RR)$ and has compact support, we write~$D(f)$ for the size of the smallest interval containing the support of $f$.

%=============================================================================
\section{Preliminaries}\label{sec:prelim}
%=============================================================================
In this section, we briefly review the second quantization formalism for fermions and quasi-free fermionic many-body states (see, e.g.,~\cite{bratteli2003operator} or~\cite{carey1987fermion} for further details), and we describe the vacuum state of massless free fermions in $1+1$ dimensions in terms of this formalism.

%-----------------------------------------------------------------------------
\subsection{The CAR algebra and quasi-free states}\label{subsec:car}
%-----------------------------------------------------------------------------
If $\HH$ is a complex Hilbert space, then let $\CAR(\HH)$ be the algebra of canonical anti-commutation relations or \emph{CAR algebra} on~$\HH$.
It is the free unital $C^\dagger$-algebra generated by elements~$a(f)$ for $f\in\HH$ such that~$f \mapsto a(f)$ is anti-linear and subject to the relations
\begin{align*}
\{a(f), a(g)\} &= 0, \\
\{a(f), a^\dagger(g)\} &= \braket{f,g}
% a(f) a^\dagger(g) + a^\dagger(g) a(f) &= \langle f, g \rangle \id \\
% a(f) a(g) + a(g) a(f) &= 0,
\end{align*}
where $\{x,y\} = xy + yx$ denotes the anti-commutator.
% A state on this algebra defines a fermionic theory.

An important class of states on this algebra are the \emph{gauge-invariant quasi-free} (or \emph{Gaussian}) states.
These states have the property that they are invariant under a global phase and that all correlation functions are determined by the two-point functions.
More precisely, for every operator~$Q$ on~$\HH$ such that $0 \leq Q \leq \id$ there exists a unique gauge-invariant quasi-free state on $\CAR(\HH)$, denoted~$\omega_Q$, such that we have the following version of Wick's rule:
\begin{align*}
  \omega_Q(a^\dagger(f_1) \ldots a^\dagger(f_n) a(g_1) \ldots a(g_m) ) = \delta_{n,m} \det [ \braket{g_i, Qf_j} ]
\end{align*}
Thus, the state is fully specified by its two-point functions~$\omega_Q(a^\dagger(f) a(g)) = \braket{g, Qf}$.
The operator~$Q$ is called the \emph{symbol} of~$\omega_Q$.
It is well-known that~$\omega_Q$ is a pure state if and only if~$Q$ is a projection.
In this case, $Q$ can be interpreted as a projection onto a \emph{Fermi sea} of negative energy modes.
Since throughout this article we will only be interested in this case, we henceforth assume that~$Q$ is a projection.

To obtain a Hilbert space realization, we consider the fermionic Fock space
\begin{align*}
  \fock(\HH) = \bigoplus_{n = 0}^{\infty}  \HH^{\wedge n}
\end{align*}
with the standard representation of $\CAR(\HH)$, defined by~$a(f) \mapsto a_0(f)$ where $a_0^\dagger(f)v = f \wedge v$.
Let~$\ket{\Omega}$ denote the Fock vacuum vector~$1 \in \HH^{\wedge 0}$.
Then $\ket{\Omega}$ is the pure state corresponding to symbol~$Q = 0$.
Now let $Q$ be an arbitrary orthogonal projection and choose a complex conjugation~$\overline{(\cdot)}$ (that is, an antiunitary involution) that commutes with~$Q$.
Then the map $a(f) \mapsto a_Q(f)$, where
\begin{equation}\label{eq:CAR_rep}
  a_Q(f) = a_0\bigl((\id-Q)f\bigr) + a_0^\dagger\bigl(\overline{Qf}\bigr),
\end{equation}
defines a representation of the CAR algebra such that~$\omega_Q$ corresponds to the Fock vacuum vector~$\ket{\Omega}$.
% Now let $Q$ be an arbitrary operator such that $0\leq Q\leq\id$ and choose a complex conjugation~$\overline{\,\cdot\,}$ (that is, an antiunitary involution) that commutes with~$Q$.
% Let $R = \sqrt Q$ denote the positive semidefinite square root of~$Q$.
% Then the map $a(f) \mapsto a_Q(f)$, where
% \begin{align*}
%   a_Q(f) = a_0\bigl((\id-R)f\bigr) + a_0^\dagger\bigl(\overline{Rf}\bigr),
% \end{align*}
% defines a representation of the CAR algebra such that~$\omega_Q$ corresponds to the Fock vacuum vector~$\ket{\Omega}$.
% \begin{align*}
%   \omega_Q(a^\dagger(f_1) \ldots a^\dagger(f_n) a(g_1) \ldots a(g_m) )
% = \braket{\Omega, a_Q^\dagger(f_1) \ldots a_Q^\dagger(f_n) a_Q(g_1) \ldots a_Q(g_m) \Omega}.
% \end{align*}

%-----------------------------------------------------------------------------
\subsection{Second-quantized operators}\label{subsec:second quantization}
%-----------------------------------------------------------------------------
Next we recall the second quantization of operators on~$\HH$.
If $U$ is a unitary on $\HH$ then $U$ defines an automorphism of~$\CAR(\HH)$, known as a \emph{Bogoliubov transformation}, through $a(f) \mapsto a(Uf)$.
Provided that $[U, Q]$ is Hilbert-Schmidt, this automorphism can be implemented by a unitary operator~$\Gamma_Q(U)$ on Fock space, which is unique up to an overall phase.
This means that, for every~$f\in\HH$,
\begin{equation*}
  \Gamma_Q(U) a_Q(f) \Gamma_Q(U)^\dagger = a_Q(Uf).
\end{equation*}
Now consider a unitary one-parameter subgroup~$\{e^{itA}\}$ generated by a bounded Hermitian operator~$A$ on~$\HH$.
This would like to know when $e^{itA}$ can be unitarily implemented in the form
\begin{equation}\label{eq:1psg}
  e^{it\d\Gamma_Q(A)} a_Q(f) e^{-it\d\Gamma_Q(A)} = a_Q(e^{itA} f)
\end{equation}
for~$t\in\RR$ and~$f\in\HH$.
For this, decompose $A$ into blocks with respect to $\HH_\pm$, which we define as the range of the projections~$Q_+ = \id - Q$ and $Q_- = Q$ (corresponding to positive and negative energy modes), respectively:
\begin{align*}
A = \begin{pmatrix}
A_{++} & A_{+-} \\
A_{-+} & A_{--}
\end{pmatrix}
\end{align*}
In~\cite{carey1987fermion,lundberg1976quasi} it is shown that, if~$A$ is bounded and the off-diagonal parts $A_{+-}$, $A_{-+}$ are Hilbert-Schmidt, then there exists a self-adjoint generator~$\d\Gamma_Q(A)$ on~$\fock(\HH)$ such that~\eqref{eq:1psg} holds.
We can moreover fix the undetermined additive constant by demanding that
\begin{equation*}
\braket{\Omega, \d\Gamma_Q(A) \Omega} = 0,
\end{equation*}
which corresponds to \emph{normal ordering} with respect to the state~$\omega_Q$.

If $A$ is trace class then $\d\Gamma_Q(A)$ is bounded and in fact can be defined as an element of~$\CAR(\HH)$.
In general, $\d\Gamma_Q(A)$ is unbounded, but we still have the bound~\cite[(2.53)]{carey1987fermion}
\begin{align}\label{eq:operator_est}
  \lVert \d\Gamma_Q(A)\Pi_n \rVert \leq 4(n+2)\max \{ \lVert A_{++}\rVert, \lVert A_{--}\rVert, \lVert A_{+-}\rVert_2, \lVert A_{-+}\rVert_2 \},
\end{align}
where $\Pi_n$ denotes the orthogonal projection on the subspace of $\fock(\HH)$ spanned by states of no more than~$n$ particles.
Combining~\cite[(2.14), (2.24), (2.25), (2.49)]{carey1987fermion}, one can similarly show that
\begin{align}
\label{eq:second quant vs symbol}
  \lVert \left( \d\Gamma_Q(A) - \d\Gamma_{Q'}(A) \right) \Pi_n \rVert
% \leq (n+2) \left( \lVert Q_+ A Q_+ - Q'_+ A Q'_+ \rVert + \lVert Q_- A Q_- - Q'_- A Q'_- \rVert  + \lVert Q_- A Q_+ - Q'_- A Q'_+ \rVert_2 + \lVert Q_+ A Q_- - Q'_+ A Q'_- \rVert_2 \right)
\leq 4 (n+2) \max_{\delta=\pm} \{ \lVert Q_\delta A Q_\delta - Q'_\delta A Q'_\delta \rVert, \lVert Q_\delta A Q_{-\delta} - Q'_\delta A Q'_{-\delta} \rVert_2 \}
\end{align}
for any two projections~$Q$ and~$Q'$.
This estimate will be useful in our error analysis in \cref{subsec:approx}.
\subsection{Massless free fermions in 1+1 dimensions}\label{sec:freefermions}
%-----------------------------------------------------------------------------
We now describe the vacuum state of the free Dirac fermion quantum field theory in~$1+1$ dimensions in terms of the second quantization formalism.
It will be convenient to consider the Dirac equation in the form
\begin{align*}
i\gamma^{\mu}\partial_{\mu} \psi = 0,
\end{align*}
with the Dirac matrices~$\gamma^0 = i\sigma_z = \begin{psmallmatrix}i&0\\0&-i\end{psmallmatrix}$ and~$\gamma^1 = -\sigma_x = \begin{psmallmatrix}0&-1\\-1&0\end{psmallmatrix}$.
The equation is easily seen to be solved by
\begin{align*}
  \psi_1(x,t) &= \chi_+(x + t) + \chi_-(x-t) \\
  \psi_2(x,t) &= i\left( \chi_+(x+t) - \chi_-(x - t) \right)
\end{align*}
for arbitrary functions $\chi_+$ and $\chi_-$, which we take to be in~$L^2(\RR)$ in order for the solutions to be normalizable. The energy of such a solution is given by
\begin{align*}
  E = \int_{-\infty}^\infty \left(-\omega|\hat{\chi}_+(\omega)|^2 + \omega|\hat{\chi}_-(\omega)|^2\right) \, \d\omega.
\end{align*}
Thus, the space of negative energy solutions is spanned by solutions for which~$\chi_+$ has a Fourier transform with support on the positive half-line (is analytic) and~$\chi_-$ has a Fourier transform with support on the negative half-line (is anti-analytic).

We obtain a single-particle Hilbert space~$\HH = L^2(\RR) \ot \CC^2$ corresponding to $\psi(x,t\!=\!0)$.
The symbol of the vacuum state is given by the projection onto the `Dirac sea' of negative energy solutions.
It can be expressed as
\begin{align}\label{eq:free_fermion_symbol}
Q = \frac{1}{2}\begin{pmatrix}
\id & \Htr \\
-\Htr & \id
\end{pmatrix}
\end{align}
in terms of the \emph{Hilbert transform}, which is the unitary operator on~$L^2(\RR)$ defined by
\begin{align*}
  \widehat{ \Htr f} (\omega) = -i\sgn(\omega) \hat{f}(\omega).
\end{align*}
Indeed, it follows from~$\Htr^\dagger = -\Htr$ that~$Q$ is an orthogonal projection, and $Q \psi = \psi$ if $\psi$ is the restriction to~$t=0$ of a negative-energy solution.
 % and only if $\psi_1(x) = \chi_+(x) + \chi_-(x)$, $\psi_2(x) = i(\chi_+(x) - \chi_-(x))$ for $\chi_\pm$ that have Fourier transform in the positive and negative half-line, respectively.
We further note that the symbol~$Q$ commutes with the component-wise complex conjugation on~$\HH$.
We thus obtain a Fock space realization as described above in \cref{subsec:car}.
The smeared Dirac field can be defined as~$\Psi(f) := a_Q(f)$ for~$f\in\HH$.

We will also be interested in free Dirac fermions on the circle~$\SS$.
In this case, we take~$\HH=L^2(\SS)\ot\CC$.
For periodic boundary conditions, the symbol~$Q^{\per}$ has the same form as in~\eqref{eq:free_fermion_symbol}, where we now let
\begin{align*}
\widehat{ \Htr^{\per} f} [n] = -i\sgn(n) \hat{f}[n]
\end{align*}
where there is some ambiguity in the sign function for~$n = 0$ (reflecting a ground state degeneracy).
For definiteness, we choose~$\sgn(0) = 1$.

For anti-periodic boundary conditions, corresponding to the Dirac equation on the nontrivial spinor bundle over~$\SS$, we define a unitary operator~$T$ on~$\HH$ by~$Tf (x) = e^{-i\pi x}f(x)$ for~$x \in (0,1)$.
Then the symbol is given by $T^{\dagger} Q^{\per} T$.

%-----------------------------------------------------------------------------
\subsection{Self-dual CAR algebra and Majorana fermions}\label{subsec:selfdual}
%-----------------------------------------------------------------------------
Suppose that $\HH_+ \cong \HH_-$, as in the preceding section.
Given an anti-unitary involution~$C$ on~$\HH$ such that $C Q_\delta = Q_{-\delta} C$ for $\delta=\pm$, we can also define the following operators on $\fock(\HH_+) \subset \fock(\HH)$,
\begin{align}\label{eq:c_Q(f)}
  c_Q(f)
= a_0(Q_+ f) + a_0^\dagger(CQ_- f)
% = a_0(Q_+ f) + a_0^\dagger(Q_+C f).
\end{align}
% (i.e., $\bar{Q_- f}$ is replaced by $CQ_- f$).
These satisfy the relations of the \emph{self-dual CAR algebra}, $\CARsd(\HH)$~\cite{araki1971quasifree}, which is generated by elements~$c(f)$ for $f\in\HH$ such that $f \mapsto c(f)$ is antilinear and
\begin{align*}
  \{c(f), c^\dagger(g)\} &= \braket{f, g}, \\
  c^\dagger(f) &= c(Cf).
\end{align*}
for~$f, g \in \HH$.
The second equation implies that a unitary~$U$ on $\HH$ only defines an automorphism~$\Gamma^c(U)$ of~$\CARsd(\HH)$ by~$c(f) \mapsto c(Uf)$ if~$[U,C]=0$ commutes with $C$.
We can also second quantize generators as in \cref{eq:1psg}.
That is, if $A$ is a bounded operator with Hilbert-Schmidt $A_{+-}$, $A_{-+}$, and if~$A^\dagger = -CAC$, we can define its second quantization $\d \Gamma_Q^c(A)$, such that
\begin{equation}
  e^{it\d\Gamma^c_Q(A)} c_Q(f) e^{-it\d\Gamma^c_Q(A)} = c_Q(e^{itA} f).
\end{equation}

We can apply this construction in the situation \cref{sec:freefermions} to obtain a description of massless free Majorana fermions.
Define the anti-unitary involution~$C$ as the following charge conjugation operator which exchanges positive and negative energy modes:
\begin{align}\label{eq:majorana C}
Cf = \begin{pmatrix}
1 & 0 \\
0 & -1
\end{pmatrix} \bar{f}
\end{align}
Then it is clear from~\eqref{eq:free_fermion_symbol} that $C Q = (I-Q) C$, so the above construction applies.
We denote by~$\Phi(f) := c_Q(f)$ the smeared Majorana field.

%=============================================================================
\section{Hilbert pair wavelets}\label{sec:wavelets}
%=============================================================================
Our circuits for free-fermion correlation functions will be obtained by second quantizing a wavelet transformation.
In this section, we first review the basic theory of wavelets on the line and circle.
In \cref{subsec:wavelets} we explain the definition of a wavelet basis, and how a choice of wavelet basis stratifies a function space into different scales.
Next, in \cref{subsec:wavelet_dec} we explain how these different scales are related through filters, and in \cref{subsec:periodic_wavelets} we explain the periodic version.
An important question is how accurately a function $f$ is approximated if all but a finite number of scales are truncated.
This is discussed in \cref{subsec:wavelet approximations}, where we prove some results that are completely standard in the wavelet literature, but which we work out for convenience of the reader, and in order to be able to carefully keep track of all the constants involved.
Using an argument from Fourier analysis in \cref{lem:technical} we show in \cref{lem:UV} an approximation result for a `UV cut-off' for a sufficiently smooth $f$, where we discard all detail at fine scales, or alternatively in \cref{lem:sampling error}, if we sample $f$.
Next we show in \cref{lem:IR} that for compactly supported functions we can also discard large scale wavelet components up to a small error, which should be thought of as an `IR cut-off`.
Finally, in \cref{subsec:approx hilbert pair} we introduce a way to implement the Hilbert transform using wavelets. Since we want to use compactly supported wavelets, this can only be done approximately, and in \cref{lem:filtererror}, \cref{lem:filtererror periodic} and \cref{lem:scaling error} we the bound approximation errors this gives rise to.

For a more extensive introduction to wavelets we refer the reader to, e.g., Chapter~7 in~\cite{mallat2008wavelet}.
We then define the central notion of an \emph{approximate Hilbert pair} of wavelet filters (\cref{def:hilbert_pair}) and derive some estimates that will later be used to derive our first-quantized approximation results.

%-----------------------------------------------------------------------------
\subsection{Wavelet bases}\label{subsec:wavelets}
%-----------------------------------------------------------------------------
A wavelet basis is an orthonormal basis for~$L^2(\RR)$ consisting of scaled and translated versions of a single localized function~$\psi\in L^2(\RR)$, called the \emph{wavelet function}.
If we define
\begin{align*}
  W_j = \sup \, \{ \psi_{j,k} : k \in \ZZ \},
  \quad\text{where}\quad
  \psi_{j,k}(x) &= 2^{\frac j2} \psi(2^j x - k),
\end{align*}
then $L^2(\RR) = \bigoplus_j W_j$.
We can therefore interpret of~$W_j$ as the \emph{space of functions at scale~$j$}, also called the detail space at scale~$j$, where large~$j$ corresponds to fine scales and small~$j$ to coarse scales.

In signal processing, wavelet bases are often constructed from an auxiliary function~$\phi\in L^2(\RR)$, known as the \emph{scaling function}.
\change{The idea is that this scaling function will such that
\begin{align*}
  V_j = \Span \, \{ \phi_{j,k} : k \in \ZZ \},
  \quad\text{where}\quad
  \phi_{j,k}(x) &= 2^{\frac j2} \phi(2^j x - k)
\end{align*}
is the space spanned by all $\psi_{l,k}$ for $l < j$, and the $\phi_{j,k}$ are an orthonormal basis for $V_j$.
Then we can decompose any function $f \in L^2(\RR)$ as
\begin{align*}
  f = \sum_{k \in \ZZ} s_k \phi_{j,k} + \sum_{l = j}^{\infty} \sum_{k \in \ZZ} w_{l,k} \psi_{l,k}
\end{align*}
where
\begin{align*}
  s_k = \langle \phi_{j,k}, f \rangle \qquad  w_{l,k} = \langle \psi_{l,k}, f \rangle.
\end{align*}
Thus we can interpret~$V_j$ as the space of \emph{functions up to (but excluding) details at scale~$j$}.}
To be precise, we demand that the $V_j$ form a complete filtration of $L^2(\RR)$, i.e.,
\begin{align*}
  \{0\} \subseteq \ldots \subseteq V_j \subseteq V_{j+1} \subseteq \ldots \subseteq L^2(\RR),
  \quad
  \overline{\bigcup_j V_j} = L^2(\RR),
\end{align*}
and that the wavelets at scale~$j$ span exactly the orthogonal complement of~$V_j$ in~$V_{j+1}$:
\begin{align}\label{eq:V=V+W}
  V_{j+1} = W_j \op V_j
\end{align}
for all $j\in\ZZ$.
A sequence of subspaces~$\{V_j\}_{j\in\ZZ}$ as above is said to form a \emph{multiresolution analysis}, since \cref{eq:V=V+W} allows to recursively decompose a signal in some~$V_j$ scale by scale.
The orthogonality between scaling and wavelet function is well-illustrated by the \emph{Haar wavelet} (see \cref{fig:scaling and wavelet},~(a)), which was used in Qi's exact holographic mapping~\cite{qi2013exact}.
We will use pairs of wavelets that are tailored to target the vacuum of the Dirac theory (see \cref{subsec:approx hilbert pair} below).

\begin{figure}
\centering
(a)
\!\!\raisebox{-4.5cm}{\includegraphics[width=0.45\linewidth]{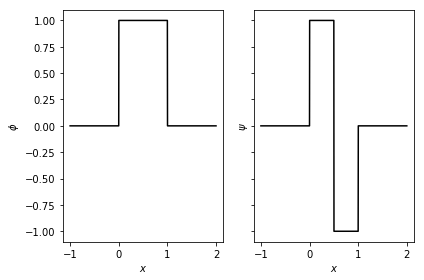}}
\quad
(b)
\!\!\raisebox{-4.5cm}{\includegraphics[width=0.45\linewidth]{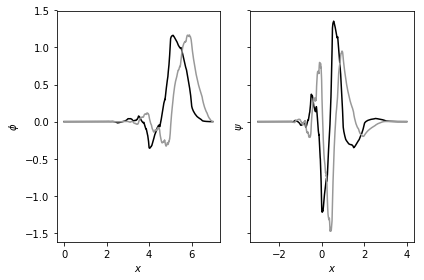}}
\caption{(a)~
Scaling and wavelet function for the Haar wavelet ($\phi = \one_{[0,1)}$, $\psi = \one_{[0,\frac{1}{2})} - \one_{[\frac{1}{2},1)}$).
% Scaling function~$\phi = \one_{[0,1)}$ and wavelet function~$\psi = \one_{[0,\frac{1}{2})} - \one_{[\frac{1}{2},1)}$ for the Haar wavelet.
(b)~Scaling and wavelet functions for the approximate Hilbert pair with parameters~$K=L=2$ due to Selesnick ($\phi^h$, $\psi^h$ in black; $\phi^g$, $\psi^g$ in gray).
See \cref{subsec:approx hilbert pair,tab:constants} for further detail.}
\label{fig:scaling and wavelet}
\end{figure}

Wavelet bases as above can be obtained by deriving them from filters.
A sequence~$g_s \in \ell^2(\ZZ)$ is called a \emph{scaling filter} (or \emph{low-pass filter}) if its Fourier transform satisfies, for all $\theta\in\RR/2\pi\ZZ$,
\begin{align}\label{eq:cond scaling filter}
|\hat{g}_s(\theta)|^2 + |\hat{g}_s(\theta + \pi)|^2 = 2 \quad\text{and}\quad \hat{g}_s(0) = \sqrt{2}.
\end{align}
Under mild technical conditions on~$g_s$ (see, e.g.,~\cite[Thm~7.2]{mallat2008wavelet}), which we always assume to be satisfied, we can define scaling and wavelet functions $\phi$, $\psi\in L^2(\RR)$ such that
\begin{align*}
  \psi(x) &= \sqrt2 \sum_{n\in\ZZ} g_w[n] \phi(2x-n), \\
  \phi(x) &= \sqrt2 \sum_{n\in\ZZ} g_s[n] \phi(2x-n).
\end{align*}
The sequence $g_w\in\ell^2(\ZZ)$ is known as the \emph{wavelet filter} (or \emph{high-pass filter}) and it can be computed via
\begin{align}\label{eq:wavelet from scaling filter}
  \hat g_w(\theta) = e^{-i\theta} \overline{\hat{g}_s(\theta + \pi)},
  \quad\text{i.e.}\quad
  g_w[n] = (-1)^{1-n} \bar g_s[1-n].
\end{align}
Thus, the expansion coefficients of the wavelet and scaling function at scale~$j=0$ in terms of scaling functions at scale~$j=1$ are precisely given by the wavelet and scaling filters, respectively (cf.~\cref{eq:V=V+W}).
This generalizes immediately to arbitrary scales:
For all~$j,k\in\ZZ$,
\begin{align}
\label{eq:relation_filter_wavelet}
  \psi_{j,k} &= \sum_{n\in\ZZ} g_w[n] \phi_{j+1,2k+n}, \\
\label{eq:relation_filter_scaling}
  \phi_{j,k} &= \sum_{n\in\ZZ} g_s[n] \phi_{j+1,2k+n}.
\end{align}
In Fourier space, these relations read
\begin{align}
\label{eq:relation_filter_wavelet_fourier}
  \hat{\psi}(\omega) &= \frac{1}{\sqrt{2}}\hat{g}_w\Bigl(\frac{\omega}{2}\Bigr)\hat{\phi}\Bigl(\frac{\omega}{2}\Bigr),\\
\label{eq:relation_filter_scaling_fourier}
  \hat{\phi}(\omega) &= \frac{1}{\sqrt{2}}\hat{g}_s\Bigl(\frac{\omega}{2}\Bigr)\hat{\phi}\Bigl(\frac{\omega}{2}\Bigr)
\end{align}
for all~$\omega\in\RR$.
The Fourier transform of the scaling function can be expressed as an infinite product of evaluations of the scaling filter:
\begin{equation}\label{eq:inf_product}
\hat{\phi}(\omega) = \prod_{k = 1}^{\infty} \frac{1}{\sqrt{2}}\hat{g}_s(2^{-k} \omega)
\end{equation}
In particular, it is bounded by one, i.e., $\lVert\hat\phi\rVert_\infty=1$.
It is also useful to note that the wavelet function averages to zero, i.e., $\int_{-\infty}^\infty \psi(x) dx = 0$.
% This follows from \hat g_s(0) = sqrt(2) => \hat g_w(0) = 0 => \hat psi(0) = 0.

Throughout this article, we will always work with filters of \emph{finite length} (the length of a sequence~$f\in\ell^2(\ZZ)$ is defined as the minimal number~$M$ such that~$f$ is supported on~$M$ consecutive sites).
Specifically, we will assume that the support of the scaling filter is~$\{0,\dots,M-1\}$.
In the signal processing literature, such filters are called finite impulse response~(FIR) filters with $M$~taps.
It is clear from \cref{eq:wavelet from scaling filter} that in this case the wavelet filter is supported in~$\{2-M,\dots,1\}$, hence has finite length~$M$ as well.
If the filters have finite length then the wavelet and scaling functions are compactly supported on intervals of width~$M$~\cite[Prop~7.2]{mallat2008wavelet}.

% We will discuss further properties in \cref{subsec:approx hilbert pair} below.
% Discuss regularity

% Pairs of filters~$(g_s,g_w)$ are called conjugate mirror filters in the signal processing literature.

%-----------------------------------------------------------------------------
\subsection{Wavelet decompositions}\label{subsec:wavelet_dec}
%-----------------------------------------------------------------------------
Suppose that we would like to express a given function~$f\in L^2(\RR)$ in a wavelet basis.
As a first step, we replace $f$ by $P_j f \in V_j$, where $P_j\colon L^2(\RR)\to V_j$ denotes the orthogonal projection onto the space of functions below scale~$j$.
This is corresponds to removing high frequency components (in signal processing) or to a UV cut-off (in physics).
We explain in \cref{lem:UV} below how to bound the error $\lVert f - P_j f\rVert$ in terms of a Sobolev norm.
To express $P_j f$ in terms of the orthonormal basis~$\{\phi_{j,k}\}_{k\in\ZZ}$ of~$V_j$, define the partial isometries
\begin{align}\label{eq:def alpha}
  \alpha_j \colon L^2(\RR) \to \ell^2(\ZZ), \qquad
  (\alpha_j f)[k] = \braket{\phi_{j,k}, f},
\end{align}
where we note that $P_j = \alpha_j^\dagger \alpha_j$.
We show below that, if~$f$ is sufficiently smooth, the coefficients $\alpha_j f$ can be well-approximated by sampling~$f$ on a uniform grid with spacing~$2^{-j}$ (\cref{lem:sampling error}).

Next, we iteratively obtain the wavelet coefficients of~$P_j f$ at all scales~$n<j$.
For this purpose, let
\begin{align*}
  \beta_j \colon L^2(\RR) \to \ell^2(\ZZ), \quad (\beta_j f)[k] = \braket{\psi_{j,k}, f},
\end{align*}
and define the unitary operator
\begin{align}\label{eq:W}
  W \colon \ell^2(\ZZ) \rightarrow \ell^2(\ZZ) \ot \CC^2, \qquad
  W f = (\downarrow m(\overline{\hat g_w}) f) \op (\downarrow m(\overline{\hat g_s}) f),
\end{align}
where we recall that the downsampling operator is given by~$(\downarrow f)[n] = f[2n]$.
Then, \cref{eq:relation_filter_wavelet,eq:relation_filter_scaling} imply that
\begin{align*}
  W \alpha_{j+1} f = \beta_j f \op \alpha_j f
\end{align*}
for all $f\in L^2(\RR)$ and $j\in\ZZ$.
That is, applying~$W$ to the scaling coefficients at some scale~$j$ yields in the first component the wavelet coefficients and in the second component the scaling coefficients at one scale coarser.
Note that, due to the scale invariance of the wavelet basis, the operator~$W$ does \emph{not} depend explicitly on~$j$.
We can iterate this procedure to obtain a map
\begin{align}\label{eq:W(L)}
  W^{(\cL)} \colon \ell^2(\ZZ) \rightarrow \ell^2(\ZZ) \ot \CC^{\cL + 1}, \qquad
  W^{(\cL)} = (\id_{\ell^2(\ZZ) \ot \CC^{\cL-1}} \op W) \cdots (\id_{\ell^2(\ZZ)} \op W) W,
\end{align}
which decomposes through successive filtering the scaling coefficients at scale~$j + 1$ into the wavelet coefficients at scales~$j$ to $j - \cL + 1$ and the scaling coefficients at scale~$j - \cL + 1$.
That is:
\begin{align*}
  W^{(\cL)} \alpha_{j+1}
= \beta_j \op W^{(\cL-1)} \alpha_j
= \dots
= \beta_j \op \beta_{j-1} \op \dots \op \beta_{j-\cL+1} \op \alpha_{j-\cL+1},
\end{align*}
or
\begin{align*}
  W^{(\cL)} \alpha_{j+1} f
= \sum_{l=0}^{\cL-1} \beta_{j-l+1} f \ot \ket l
+ \alpha_{j-\cL+1} f \ot \ket{\cL}
\end{align*}
for all~$f\in L^2(\RR)$.
The unitaries~$W$ and~$W^{(\cL)}$ are known as ($\cL$ layers of) the \emph{discrete wavelet transform}.
Note that $W^{(\cL)}$ can be readily implemented by a scale-invariant linear circuit consisting of convolutions and downsampling circuit elements (see \cref{subsec:periodic circuit} and \cref{fig:wavelet_decomposition} for a visualization).

\begin{figure}
\begin{center}
\includegraphics[width=0.9\linewidth]{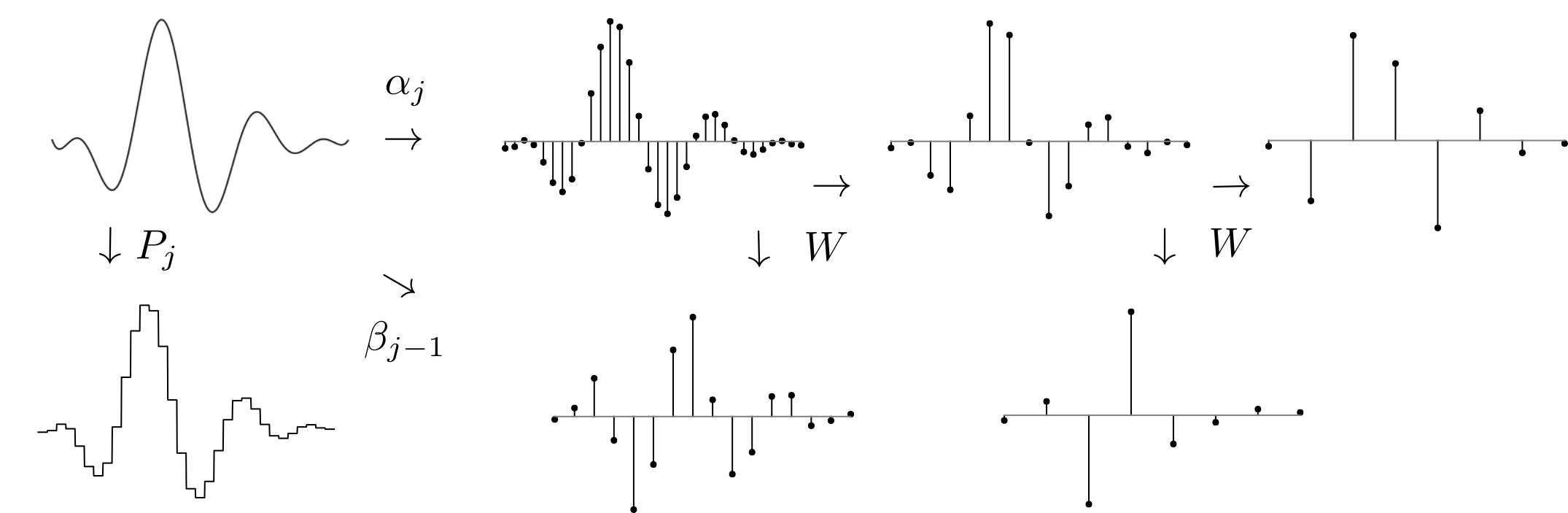}
\caption{Illustration of the various maps defined in \cref{subsec:wavelet_dec} in the case of the Haar wavelet.}
\label{fig:wavelet_decomposition}
\end{center}
\end{figure}

%-----------------------------------------------------------------------------
\subsection{Periodic wavelets}\label{subsec:periodic_wavelets}
%-----------------------------------------------------------------------------
Given a wavelet~$\psi$ on $\RR$ with scaling function~$\phi$ and filters~$g_s$ and~$g_w$, one can construct a corresponding family of \emph{periodic} wavelet and scaling functions on the circle~$\SS$.
Following~\cite[Section 7.5]{mallat2008wavelet}, we define for~$j \geq 0$ and $k=1,\dots,2^j$ the functions
\begin{align*}
\psi^{\per}_{j,k}(x) &= \sum_{m \in \ZZ} \psi_{j,k}(x + m), \\
\phi^{\per}_{j,k}(x) &= \sum_{m \in \ZZ} \phi_{j,k}(x + m)
\end{align*}
in~$L^2(\SS)$.
If we set $V^{\per}_j = \Span \, \{ \phi^{\per}_{j,k} : k = 1,\dots,2^j \}$ and $W^{\per}_j = \Span \, \{ \psi^{\per}_{j,k} : k = 1,\dots, 2^j \}$ then we have
\begin{align*}
  \CC \one = V_0 \subseteq V_1 \subseteq \dots \subseteq L^2(\SS), \quad
  \overline{\bigcup_{j\geq0} V_j} = L^2(\SS), \quad\text{and}\quad
  V_{j+1} = W_j \op V_j.
\end{align*}
The space~$V_0$ is one-dimensional and consists of the constant functions. % (i.e., $\phi^{\per}_{0,1}(x) = 1$ for $x\in\SS$).
Thus, $\{\psi^{\per}_{j,k}\}_{j\geq0, k=1,\dots,2^j}$ together with $\phi^{\per}_{0,1} = \one$ form an orthonormal basis of~$L^2(\SS)$.
Similarly to before, we denote by~$\alpha_j^{\per}, \beta_j^{\per} \colon L^2(\SS) \to \CC^{2^j}$ denote the partial isometries that send a function to its expansion coefficients with respect to the periodized scaling and wavelet basis functions (for fixed $j$), and we denote by~$P_j^{\per} = (\alpha^{\per}_j)^\dagger \alpha^{\per}_j$ to be the orthogonal projection onto $V_j \subseteq L^2(\SS)$.

Since the radius of the circle sets a coarsest length scale, the corresponding filters are now scale-dependent and given by
\begin{align*}
  g^{\per}_{s,j}[n] &= \sum_{m \in \ZZ} g_s[n + 2^jm]\\
  g^{\per}_{w,j}[n] &= \sum_{m \in \ZZ} g_w[n + 2^jm]
\end{align*}
for $j \geq 0$ and $n=1,\dots,2^j$.
As before, they give rise to unitary maps
\begin{align}\label{eq:W(Lper)}
  W^{(\cL),\per} \colon \CC^{2^\cL} \rightarrow \bigoplus_{j = 0}^{\cL-1} \CC^{2^j} \op \CC, \quad
  W^{(\cL),\per} \alpha^{\per}_{\cL} = \beta^{\per}_{\cL-1} \op \ldots \op \beta_0^{\per} \op \alpha_0^{\per}
  % \alpha^{\per}_{\cL} f &\mapsto \bigl(\beta^{\per}_{\cL-1} f, \dots, \beta^{\per}_0 f, \braket{\phi_{0,1}^{\per}, f} \bigr)
\end{align}
that expand a signal at a certain scale into (all) its wavelet coefficients and the remaining scaling coefficient (which is the average of~$f$).

We note that $g^{\per}_{s,j} = g_s$ and $g^{\per}_{w,j} = g_w$ for sufficiently large~$j$ (namely when $2^j$ is at least as large as the cardinality of the filters' supports).
This is intuitive since at sufficiently fine scales the periodicity of the circle is no longer visible.
See \cref{subsec:periodic circuit} for more detail.

%-----------------------------------------------------------------------------
\subsection{Wavelet approximations}\label{subsec:wavelet approximations}
%-----------------------------------------------------------------------------
We need to know how well we can approximate functions if we are only allowed to use a finite number of scales.
In this section we will state three results (the last of which is adapted from~\cite{wojtaszczyk1997mathematical}) that give quantitative bounds assuming that the wavelets are compactly supported and bounded.
These conditions can easily be relaxed, but we will not need this.
The proofs are given in \cref{sec:appendix}.

\change{We know that the shifted and rescaled wavelet functions form an orthonormal basis for $L^2(\RR)$, so we can write
\begin{align*}
  f = \sum_{j \in \ZZ} \sum_{k \in \ZZ} \beta_{j}(f)[k] \psi_{j,k}.
\end{align*}
We will be interested in approximating the function $f$ over a finite number of scales, that is, we will try to approximate $f$ by
\begin{align*}
  f \approx \sum_{j= j_0}^{j_1} \sum_{k \in \ZZ} \beta_{j}(f)[k] \psi_{j,k}.
\end{align*}
Recall that $j \gg 0$ corresponds to small scale structures in $f$, and $j \ll 0$ to large scale structure.
In this section we will give quantative bounds which show that one can approximate $f$ to good precision using a finite number of scales if $f$ is sufficiently smooth and compactly supported.
To be precise, in \cref{lem:UV} we show that large $j$ can be discarded with small error if $f$ is sufficiently smooth, and in \cref{lem:IR} we show that we can truncate in the other direction if $f$ is compactly supported.}

Our first result bounds the error incurred by leaving out detail, corresponding to a UV cut-off.
Recall that the Sobolev spaces~$H^K(\RR)$ and~$H^K(\SS)$ consist of functions~$f$ with square-integrable weak~$K$-th derivative, denoted~$f^{(K)}$.

\begin{lem}[UV cut-off]\label{lem:UV}
Assume that the Fourier transform of the scaling filter~$\hat{g}_s(\theta)$ has a zero of order~$K$ at $\theta = \pi$.
Then there exists a constant~$C_{\UV}$ such that for every $f \in H^K(\RR)$ and $j\in\ZZ$, we have that
\begin{align*}
  \lVert P_j f - f \rVert &\leq 2^{-Kj} C_{\UV} \lVert f^{(K)} \rVert.
\end{align*}
Similarly, for every $f\in H^1(\SS)$ and $j\geq0$, we have that
\begin{align*}
  \lVert P^{\per}_j f - f \rVert &\leq 2^{-Kj} C_{\UV} \lVert f^{(K)} \rVert.
\end{align*}
If the scaling filter is supported in $\{0,\dots,M-1\}$, then these estimates hold with~$K=1$ and~$C_{\UV} \leq 2M^2$.
\end{lem}

In fact, under mild technical conditions the `UV cut-off' from \cref{lem:UV} can be well-approximated by sampling the function on a dyadic grid, as shown in the following lemma.

\begin{lem}[Sampling error]\label{lem:sampling error}
There exists a constant~$C_\phi$ such that the following holds:
For every~$f \in H^1(\RR)$ and~$f_j$ the sequence defined by~$(f_j)_k := 2^{-j/2} f(2^{-j} k)$ for $k\in\ZZ$ (we identify~$f$ with its unique representative as a continuous function), we have
\begin{align*}
  \lVert \alpha_j f - f_j \rVert \leq 2^{-j} C_\phi \lVert f' \rVert.
\end{align*}
Likewise, for every~$f\in H^1(\SS)$ and~$f_j \in \CC^{2^j}$ the vector with components~$(f_j)_k := 2^{-j/2} f(2^{-j} k)$,
\begin{align*}
  \lVert \alpha_j^{\per} f - f_j \rVert \leq 2^{-j} C_\phi \lVert f' \rVert.
\end{align*}
If the scaling filter is supported in~$\{0,\dots,M-1\}$, then these estimates hold with~$C_\phi\leq 2M^2$.
\end{lem}

The final lemma of this section bounds the error incurred by leaving out coarse scale components from compactly supported functions, corresponding to an IR cut-off.

\begin{lem}[IR cut-off]\label{lem:IR}
Assume that the scaling function~$\phi$ satisfies
\begin{align*}
\sqrt{\sum_{k\in\ZZ} \lvert\phi(y - k)\rvert^2} \leq C_{\IR}
\end{align*}
for all $y \in \RR$.
Then for every $f\in L^2(\RR)$ with compact support,
\begin{align*}
  \lVert P_j f \rVert \leq 2^{j/2} \sqrt{D(f)} C_{\IR} \lVert f\rVert.
\end{align*}
In particular, if $\phi$ is bounded and supported in an interval of width~$M$, this is true with $C_{\IR} \leq \sqrt{M} \lVert \phi\rVert_\infty$.
\end{lem}

\noindent
We recall from \cref{subsec:wavelets} that both the scaling and the wavelet function are supported on intervals of the same width, which explains why we use the symbol~$M$ in both cases.
For the periodized wavelet transform, it is possible to prove a similar result when restricting to functions~$f\in L^2(\SS)$ with average zero (since the identity function is orthogonal to all wavelet basis functions).

%-----------------------------------------------------------------------------
\subsection{Approximate Hilbert pair wavelets}\label{subsec:approx hilbert pair}
%-----------------------------------------------------------------------------
Our construction of a quantum circuit that approximates fermionic correlation functions is based on approximating the Hilbert transform, which we saw appearing in the symbol in \cref{sec:freefermions}, by using wavelets.
Thus, we are looking for a pair of wavelet and scaling filters~$g_w$,~$g_s$ and~$h_w$,~$h_s$ such that the associated wavelet functions~$\psi^g$ and $\psi^h$ satisfy
\begin{align*}
  \psi^h = \Htr \psi^g,
\end{align*}
which we recall means that $\hat\psi^h(\omega) = -i\sgn(\omega) \hat\psi^g(\omega)$ for all $\omega\in\RR$.
Such a pair of wavelets is called a \emph{Hilbert pair}.
Two equivalent conditions on the scaling and wavelet filters, respectively, to generate a Hilbert pair are~\cite{selesnick2001hilbert}
\begin{align}
\label{eq:scaling_filters_hpair}
  \hat{h}_s(\theta) &= \mu_s(\theta)\hat{g}_s(\theta), \\
\nonumber
  \hat{h}_w(\theta) &= \mu_w(\theta) \hat{g}_w(\theta),
\end{align}
where $\mu_s$ and $\mu_w$ are periodic functions in~$L^\infty(\RR/2\pi\ZZ)$ defined by
\begin{align}
\label{eq:def mu_s}
  \mu_s(\theta) &= e^{-i\frac{\theta}{2}}, \\
\nonumber
  \mu_w(\theta) &= -i\sgn(\theta)e^{i\frac{\theta}{2}}
\end{align}
for~$\lvert\theta\rvert<\pi$.
In this situation, \cref{eq:relation_filter_wavelet_fourier,eq:relation_filter_scaling_fourier} implies that the scaling functions~$\phi_g$ and~$\phi_h$ will be related by
\begin{align}\label{eq:scaling functions hpair}
  \hat{\phi}^h(\omega) = \lambda_s(\omega)\hat{\phi}^g(\omega),
\end{align}
where $\lambda_s\in L^\infty(\RR)$ is defined by
\begin{equation}\label{eq:scaling_function_relation}
  \lambda_s(\omega) = -i\sgn(\omega)\mu_w(-\omega).
\end{equation}
We refer to~\cite{selesnick2001hilbert,selesnick2002design} for further detail.
Since the Hilbert transform does not preserve compact support, we can not hope for exact Hilbert pair wavelets using compactly supported wavelets.
However, an approximate version can be realized.
The following definition describes the notion of approximation that is appropriate in our context.

\begin{dfn}\label{def:hilbert_pair}
An \emph{$\eps$-approximate Hilbert pair} consists of a pair of wavelet and scaling filters, $g_w$, $g_s$, $h_w$, $h_s$, with corresponding wavelet functions $\psi^g$, $\psi^h$ and scaling functions $\phi^g$, $\phi^h$, such that
% \begin{enumerate}
% \item
\begin{align}\label{eq:filter error uniform}
  \lVert \hat{h}_s - \mu_s \hat{g}_s \rVert_\infty \leq \eps.
\end{align}
% and, for all~$\theta\in[-\frac\pi2,\frac\pi2]$,
% \begin{align}\label{eq:filter error linear}
%   \lvert \hat{h}_s(\theta) - \mu_s(\theta) \hat{g}_s(\theta) \rvert \leq \delta \lvert\theta\rvert.
% \end{align}
That is, the error in the phase relation~\eqref{eq:scaling_filters_hpair} for the scaling filters is bounded by~$\eps$.
 % and vanishes at~$\theta=0$.
This condition can readily be checked numerically.
\end{dfn}

One of the first systematic constructions of approximate Hilbert pairs is due to Selesnick~\cite{selesnick2002design,selesnick2001hilbert}.
His construction depends on two parameters, $K$ and $L$, where~$K$ is the number of vanishing moments of the wavelets (relevant for the approximation power of the wavelet decomposition and for the smoothness of the wavelets) and where~$L$ is essentially the number of terms in a Taylor expansion of the relation in \cref{eq:scaling_filters_hpair} at~$\theta=0$.
By construction, the filters are real and have finite length~$M = 2(K + L)$, so the wavelet and scaling functions are compactly supported on intervals of width~$M$.
Numerically, one can see that the parameter~$\eps$ in \cref{eq:filter error uniform} decays exponentially with $\min\{K,L\}$~\cite{haegeman2018rigorous}, while the other relevant parameters from \cref{lem:UV}, \cref{lem:sampling error} and \cref{lem:IR} remain bounded or grow much more slowly than the worst-case bounds that we provided, as can be seen in \cref{tab:constants}.
\change{For a more extensive treatment of the theory of approximate Hilbert pairs, see \cite{selesnick2005dual, yu2005hilbert, chaudhury2009construction, chaudhury2010hilbert} and the more recent work \cite{achard2017new}.
See \cref{fig:scaling and wavelet},~(b) for an illustration of Selesnick's wavelets with parameters~$K=L=2$.}

\begin{table}
\centering
\begin{tabular}{rrrrrrrrr}
\toprule
  $K$ &   $L$ &   $M$ &       $\eps$ &      $C_{\UV}$ &      $C_{\IR}$ &     $C_\chi$ &     $C'_\chi$ &     $C_\phi$ \\
\midrule
  1 &   1 &   4 &  0.264099 &  0.619741 &  2.542073 &  1.166423 &  1.142220 &  1.254999 \\
  2 &   2 &   8 &  0.068221 &  0.622182 &  1.217454 &  1.155488 &  0.295133 &  2.296890 \\
  3 &   3 &  12 &  0.018338 &  0.624782 &  1.190944 &  1.154757 &  0.079283 &  2.116091 \\
  4 &   4 &  16 &  0.005020 &  0.626782 &  1.150151 &  1.154705 &  0.021691 &  1.251461 \\
  5 &   5 &  20 &  0.001389 &  0.628374 &  1.130260 &  1.154701 &  0.005999 &  2.120782 \\
  6 &   6 &  24 &  0.000387 &  0.629686 &  1.120354 &  1.154701 &  0.001671 &  2.106891 \\
  7 &   7 &  28 &  0.000108 &  0.630795 &  1.114293 &  1.154701 &  0.000468 &  1.234832 \\
  8 &   8 &  32 &  0.000030 &  0.631752 &  1.108135 &  1.154701 &  0.000132 &  2.434899 \\
  9 &   9 &  36 &  0.000009 &  0.632674 &  1.106718 &  1.154701 &  0.000037 &  1.923738 \\
 10 &  10 &  40 &  0.000003 &  0.638023 &  1.440101 &  1.154701 &  0.000011 &  5.752427 \\
\bottomrule
\end{tabular}
\caption{Numerical values of various constants for Selesnick's approximate Hilbert pairs with parameters~$K = L$. It appears that $\eps$ decays exponentially with increasing $K=L$, while the other constants from \cref{lem:UV}, \cref{lem:sampling error}, \cref{lem:IR} and \cref{lem:scaling error} are well-behaved.}\label{tab:constants}
\end{table}

If we periodize an (approximate) Hilbert pair as described in \cref{subsec:periodic_wavelets}, we get periodic wavelets that are (approximately) related by the Hilbert transform on the circle.
The following lemma is an improved version of~\cite[(A7)]{haegeman2018rigorous}.
It controls the error incurred by using approximate instead of exact Hilbert pairs both on the line and on the circle.

\begin{lem}\label{lem:filtererror}
Consider an $\eps$-approximate Hilbert pair.
Let~$W_g^{(\cL)}$ and~$W_h^{(\cL)}$ denote the corresponding wavelet transforms for~$\cL$ layers, defined as in \cref{eq:W,eq:W(L)} using the filters~$g$ and~$h$, respectively.
Then,
\begin{align}
\label{eq:filtererror_wavelet}
\lVert P_w \bigl( W_g^{(\cL)}- W_h^{(\cL)}m(\mu_w) \bigr) \rVert &\leq \eps \cL, \\
\label{eq:filtererror_scaling}
\lVert P_s \bigl( W_g^{(\cL)} - W_h^{(\cL)}m(\mu_w) \bigr) f \rVert &\leq \eps \cL \lVert f\rVert + 2 \lVert P_s W_g^{(\cL)} f\rVert \qquad (\forall f\in\ell^2(\ZZ)),
\end{align}
where $P_w = \id_{\ell^2(\ZZ)} \ot \sum_{k=0}^{\cL-1} \ket k\bra k$ denotes the projection onto the wavelet coefficients and~$P_s = \id - P_w$ the projection onto the scaling coefficients.
\end{lem}
\begin{proof}
As in \cref{eq:W}, denote by $W_g, W_h \colon \ell^2(\ZZ) \rightarrow \ell^2(\ZZ) \ot \CC^2$ the unitaries corresponding to a single layer of the wavelet transform:
\begin{align*}
  W_g = (\downarrow m(\overline{\hat g_w})) \op (\downarrow m(\overline{\hat g_s})) \quad\text{and}\quad
  W_h = (\downarrow m(\overline{\hat h_w})) \op (\downarrow m(\overline{\hat h_s}))
\end{align*}
One may easily verify the relation
\begin{align*}
  m(\mu_w) \downarrow m(\mu_s) = \downarrow m(\mu_w).
\end{align*}
This allows us to rewrite
\begin{equation}\label{eq:move_mu_through_W}
\begin{aligned}
W_h m(\mu_w)
&= (\downarrow m(\overline{\hat h_w})) \op (\downarrow m(\overline{\hat h_s})) m(\mu_w) \\
% &= (\downarrow m(\mu_w \overline{\hat h_w})) \op (\downarrow m(\mu_w \overline{\hat h_s})) \\
&= (\downarrow m(\mu_w \overline{\hat h_w})) \op (\downarrow m(\mu_w) m(\overline{\hat h_s})) \\
&= (\downarrow m(\mu_w \overline{\hat h_w})) \op (m(\mu_w) \downarrow m(\mu_s \overline{\hat h_s})) \\
&= (\id_{\ell^2(\ZZ)} \op m(\mu_w)) \tilde W_h,
\end{aligned}
\end{equation}
where we introduced
\begin{align*}
  \tilde{W}_h := (\downarrow m(\mu_w \overline{\hat h_w})) \op (\downarrow m(\mu_s \overline{\hat h_s})).
\end{align*}
Now consider $\cL$~layers of the transform.
For~$l=1,\dots,\cL$, define~$W_g^l := \id_{\ell^2(\ZZ) \ot \CC^{l-1}} \op W_g$ and similarly~$W_h^l$ and~$\tilde W_h^l$, so that~$W_g^{(\cL)} = W_g^\cL \cdots W_g^1$ etc.
By using \cref{eq:move_mu_through_W}, we find that
\begin{align*}
  W_h^{(\cL)} m(\mu_w)
= W_h^\cL \cdots W_h^1 m(\mu_w)
= \left( \id_{\ell^2(\ZZ) \ot \CC^{\cL}} \op m(\mu_w) \right) \tilde W_h^\cL \cdots \tilde W_h^1
\end{align*}
Our assumption~\eqref{eq:filter error uniform} on the scaling filter error in an approximate Hilbert pair implies that, for all~$l$,
% $\lVert \overline{\hat{g}_s} - \mu_s \overline{\hat{h}_s} \rVert_\infty \leq \eps$ and
% $\lVert \overline{\hat{g}_w} - \mu_w \overline{\hat{h}_w} \rVert_\infty \leq \eps$
\begin{align}\label{eq:one_layer_error}
  \lVert W_g^l - \tilde W_h^l \rVert
= \lVert W_g - \tilde W_h \rVert
\leq \max \, \{ \lVert \overline{\hat{g}_s} - \mu_s \overline{\hat{h}_s} \rVert_\infty, \lVert \overline{\hat{g}_w} - \mu_w \overline{\hat{h}_w} \rVert_\infty \}
\leq \eps.
\end{align}
Next we write a telescoping sum
\begin{align*}
&\quad W_g^{(\cL)} - W_h^{(\cL)} m(\mu_w)
= W_g^\cL \cdots W_g^1 - \left( \id_{\ell^2(\ZZ) \ot \CC^{\cL}} \op m(\mu_w) \right) \tilde W_h^\cL \cdots \tilde W_h^1 \\
&= \bigl( \id_{\ell^2(\ZZ) \ot \CC^{\cL}} \op m(\mu_w) \bigr) \left( W_g^\cL \cdots W_g^1 - \tilde W_h^\cL \cdots \tilde W_h^1 \right)
+ \bigl( 0_{\ell^2(\ZZ) \ot \CC^{\cL}} \op (\id - m(\mu_w)) \bigr) W_g^\cL \cdots W_g^1 \\
&= \bigl( \id_{\ell^2(\ZZ) \ot \CC^{\cL}} \op m(\mu_w) \bigr) \sum_{l=1}^\cL W_g^\cL \cdots W_g^{l+1} \left( W_g^l - \tilde W_h^l \right) \tilde W_h^{l-1} \cdots \tilde W_h^1
+ \bigl( 0 \op (\id - m(\mu_w)) \bigr) W_g^\cL \cdots W_g^1
\end{align*}
Using \cref{eq:one_layer_error} and the fact that~$\lVert W_g^l\rVert = \lVert \tilde W_h^l\rVert = 1$ for all~$l$, we can therefore bound
\begin{align*}
  \lVert P_w \bigl( W_g^{(\cL)} - W_h^{(\cL)} m(\mu_w) \bigr) \rVert
% = \lVert P_w \sum_{l=1}^\cL W_g^\cL \cdots W_g^{l+1} \left( W_g^l - \tilde W_h^l \right) \tilde W_h^{l-1} \cdots \tilde W_h^1 \rVert
\leq \sum_{l=1}^\cL \lVert W_g^l - \tilde W_h^l \rVert \leq \eps \cL
\end{align*}
and, since furthermore $\lVert m(\mu_w) \rVert = 1$,
\begin{align*}
  \lVert P_s \bigl( W_g^{(\cL)} - W_h^{(\cL)} m(\mu_w) \bigr) f \rVert
&\leq \sum_{l=1}^\cL \lVert W_g^l - \tilde W_h^l \rVert \lVert f \rVert + 2 \lVert P_s W_g^\cL \cdots W_g^1 f \rVert
\leq \eps \cL \lVert f\rVert + 2 \lVert P_s W_g^{(\cL)} f\rVert.
\end{align*}
Thus we have established the desired bounds.
\end{proof}

A completely similar argument establishes a version for the periodized wavelets:

\begin{lem}\label{lem:filtererror periodic}
Consider an~$\eps$-approximate Hilbert pair.
Let~$W_g^{(\cL),\per}$ and~$W_h^{(\cL),\per}$ denote the periodized wavelet transforms for~$\cL$ layers, defined as in \cref{eq:W(Lper)} using the periodizations of the filters~$g$ and~$h$, respectively.
Then,
\begin{align}
\label{eq:filtererror periodic wavelet}
\lVert P_w^{\per} \bigl( W_g^{(\cL),\per}- W_h^{(\cL)}m(\mu_{w,\cL}^{\per}) \bigr) \rVert &\leq \eps \cL, \\
\label{eq:filtererror periodic scaling}
\lVert P_s^{\per} \bigl( W_g^{(\cL),\per} - W_h^{(\cL)}m(\mu_{w,\cL}^{\per}) \bigr) f \rVert &\leq \eps \cL \lVert f\rVert + 2 \lVert P_s^{\per} W_g^{(\cL),\per} f\rVert \qquad (\forall f\in\CC^{2^\cL}),
\end{align}
\change{where $\mu_{w,\cL}^{\per}[n] := \mu(2^{-\cL+1}\pi n)$} and where $P_w^{\per}$ denotes the projection onto the $2^\cL-1$ many wavelet coefficients and~$P_s^{\per} = \id - P_w^{\per}$ the projection onto the remaining scaling coefficient.
\end{lem}

Next, we will show that expanding a function~$f$ in the scaling basis for an approximate Hilbert pair results in approximately the same coefficients as if one were to expand the function in the scaling basis for an exact Hilbert pair (cf.~\cref{eq:scaling functions hpair}).

\begin{lem}\label{lem:scaling error}
Consider an $\eps$-approximate Hilbert pair.
Then there exists a constant~$C_\chi>0$, depending only on the scaling filters, such that the following holds:
For every $f\in H^1(\RR)$,
\begin{align*}
  \lVert \alpha_j^h f - \alpha_j^g m(\lambda_{s,j})^\dagger f \rVert \leq 2^{-j} C_\chi \lVert f'\rVert.
\end{align*}
where~$\lambda_{s,j}(\omega) := \lambda_s(2^{-j}\omega)$.
Similarly, for $f\in H^1(\SS)$ we have that
\begin{align*}
  \lVert \alpha_j^{h,\per} f - \alpha_j^{g,\per} m(\lambda^{\per}_{s,j})^\dagger f \rVert \leq 2^{-j} C_\chi \lVert f'\rVert.
\end{align*}
where~\change{$\lambda^{\per}_{s,j}[n] := \lambda_{s,j}(2\pi n)$}.
If the scaling filters are supported in~$\{0,\dots,M-1\}$ then these bounds hold with~$C_\chi \leq 2M^2$.
\end{lem}
\begin{proof}
By \cref{eq:cond scaling filter,eq:def mu_s}, $\hat h_s - \mu_s \hat g_s$ vanishes at $\theta=0$, so there exists a constant~$C>0$ such that
\begin{align}\label{eq:linear filter bound for chi}
  \frac1{\sqrt2} \lvert \hat h_s(\theta) - \mu_s(\theta) \hat g_s(\theta) \rvert \leq C \lvert \theta\rvert
\end{align}
for all~$\theta\in[-\pi,\pi]$.
As a consequence, we can derive the following bound on the Fourier transform of~$\chi := \phi^h - m(\lambda_s) \phi^g$:
For all $\omega\in[-\pi,\pi]$,
\begin{equation}\label{eq:scaling error linear}
\begin{aligned}
  \lvert \hat\chi(\omega)\rvert
&= \lvert \hat\phi^h(\omega) - \lambda_s(\omega) \hat\phi^g(\omega) \rvert
= \lvert \prod_{k=1}^\infty \frac1{\sqrt2} \hat h_s(2^{-k}\omega) - \prod_{k=1}^\infty \frac1{\sqrt2} \mu_s(2^{-k}\omega) \hat g_s(2^{-k}\omega) \rvert \\
&\leq \sum_{k=1}^\infty \frac1{\sqrt2} \lvert \hat h_s(2^{-k}\omega) - \mu_s(2^{-k}\omega) \hat g_s(2^{-k}\omega) \rvert
\leq \sum_{k=1}^\infty C \lvert 2^{-k}\omega \rvert
\leq C \lvert\omega\rvert
\end{aligned}
\end{equation}
using a telescoping series and the fact that~$\lVert\hat h_s\rVert_\infty=\lVert\mu_s \hat g_s\rVert_\infty=\sqrt2$.
Moreover, $\lVert\hat\chi\rVert_\infty\leq2$.
Thus, \cref{lem:technical} shows that, for all~$f\in H^1(\RR)$,
\begin{align*}
  \lVert \alpha_j^h f - \alpha_j^g m(\lambda_{s,j})^\dagger f \rVert^2
&= \sum_{k\in\ZZ} \lvert \braket{\phi^h_{j,k}, f} - \braket{\phi^g_{j,k}, m(\lambda_{s,j})^\dagger f} \rvert^2
% &= \sum_{k\in\ZZ} \lvert \braket{\phi^h_{j,k}, f} - \braket{m(\lambda_{s,j} \phi^g_{j,k}, f} \rvert^2
&= \sum_{k\in\ZZ} \lvert \braket{\chi_{j,k}, f} \rvert^2
&\leq 2^{-2j} C_\chi^2 \lVert f'\rVert^2,
\end{align*}
where $C_\chi^2 = C^2 + 4/3$.
The case when~$f\in H^1(\SS)$ works analogously.

Finally, assume that the scaling filters are supported in~$\{0,\dots,M-1\}$.
In this case, we know from \cref{lem:filter fourier bounds} that, for all $\theta\in[-\pi,\pi]$,
\begin{align*}
  \frac1{\sqrt2} \lvert \hat g_s(\theta) - 1 \rvert \leq \frac{M^2}2 \, \lvert\theta\rvert
  \quad\text{and}\quad
  \frac1{\sqrt2} \lvert \hat h_s(\theta) - 1 \rvert \leq \frac{M^2}2 \, \lvert\theta\rvert
\end{align*}
and hence
\begin{align*}
  \frac1{\sqrt2} \lvert \hat h_s(\theta) - \mu_s(\theta) \hat g_s(\theta) \rvert
&\leq \frac1{\sqrt2} \lvert \hat h_s(\theta) - 1 \rvert + \frac1{\sqrt2} \lvert \hat g_s(\theta) - 1 \rvert + \frac1{\sqrt2} \lvert 1 - \mu_s(-\theta) \rvert
% \\ &\leq M^2 \lvert\theta\rvert + \frac1{\sqrt2} \lvert 1 - e^{i\theta/2} \rvert
\leq \left(M^2 + \frac12\right) \lvert\theta\rvert
\end{align*}
for $\theta\in[-\pi,\pi]$.
Thus \cref{eq:linear filter bound for chi,eq:scaling error linear} hold with~$C=M^2+1/2$, hence we have~$C_\chi \leq 2M^2$.
\end{proof}

The bounds in \cref{lem:scaling error} hold for any pair of wavelets, not only for approximate Hilbert pairs.
For the latter, not only is the constant~$C$ small in practice, but one can also use the relation between the filters \cref{eq:filter error uniform} and a slightly adapted version of \cref{lem:technical} to show that in fact \cref{lem:scaling error} holds with
\begin{align*}
  C'_\chi = 3(C + \eps).
\end{align*}
For the Selesnick approximate Hilbert pairs this leads to significantly smaller constants, see \cref{tab:constants}, but since this does not substantially impact our the scaling of our final bounds on correlation functions we do not pursue this direction further.

%=============================================================================
\section{Approximation of correlation functions}\label{sec:quantumcirc}
%=============================================================================
In this section we first explain how to approximate the symbols by using an approximate Hilbert pair of wavelets.
We then prove our main technical result on the approximation of correlation functions. \change{To achieve this we adapt the approach pioneered in~\cite{haegeman2018rigorous} to the continuum setting.
In \cref{subsec:symbol approx} we discuss the approximation of the symbol by a wavelet construction.
The key difference to~\cite{haegeman2018rigorous} is that we consider the symbol in the continuum setting, and we argue that the discretization maps $\alpha_j^g$ and $\alpha_j^h$ allow us to relate the continous and discrete symbols.
We determine how the wavelet approximation results in \cref{subsec:wavelet approximations} can be used to bound the corresponding errors.
Moreover, we extend the results to the periodic case, and we provide a slightly improved scaling of the error bounds.
In \cref{subsec:approx} we discuss how the approximation of the symbol leads to approximation of correlation functions.
One main difference to~\cite{haegeman2018rigorous} is that we need to consider the dependence on the smoothness of operators we insert, which follows naturally from our approximation of the symbol.
The result we prove is also more general than the discrete result in~\cite{haegeman2018rigorous} since we also allow normal ordered quadratic operators (such as smearings of the stress-energy tensor) in the correlators.}

%-----------------------------------------------------------------------------
\subsection{Symbol approximations from Hilbert pairs}\label{subsec:symbol approx}
%-----------------------------------------------------------------------------
Recall from \cref{eq:free_fermion_symbol} that the symbol of the vacuum state of the free Dirac fermion on the real line is given by the following operator on the single-particle Hilbert space~$L^2(\RR) \ot \CC^2$:
\begin{align}\label{eq:dirac symbol restated}
Q = \frac{1}{2}\begin{pmatrix}
\id & \Htr \\
-\Htr & \id
\end{pmatrix}
= \begin{pmatrix}
\id & 0 \\
0 & \Htr^\dagger
\end{pmatrix}
\left( \id_{L^2(\RR)} \ot \ket +\bra + \right)
\begin{pmatrix}
\id & 0 \\
0 & \Htr
\end{pmatrix},
\end{align}
where~$\ket+ = \frac1{\sqrt2}(\ket0+\ket1)$.

To obtain a suitable approximation, consider an approximate Hilbert pair as in \cref{def:hilbert_pair}.
As before, we denote by~$g_w$, $h_w$, $g_s$,~$h_s$ the wavelet and scaling filters, by~$\alpha^g_j$ and~$\alpha^h_j$  discretization maps (defined as in \cref{eq:def alpha}) and by~$W_g^{(\cL)}$ and~$W_h^{(\cL)}$ the $\cL$-layer discrete wavelet transformats (defined in \cref{eq:W(L)}).
We now approximate \cref{eq:dirac symbol restated} by first truncating to a finite number of scales, using one of the two wavelet transforms, and then by replacing the Hilbert transform of the one wavelet basis by the other wavelet basis.
Schematically,
\begin{align*}
  \id_{L^2(\RR)} \;\leadsto\; \alpha_j^{h,\dagger} W^{(\cL),\dagger}_h P_w W^{(\cL)}_h \alpha_j^h, \qquad
  P_w W^{(\cL)}_h \alpha_j^h \Htr \;\leadsto\; P_w W^{(\cL)}_g \alpha_j^g,
\end{align*}
where $P_w = \id_{\ell^2(\ZZ)} \ot \sum_{k=0}^{\cL-1} \ket k\bra k$ denotes the orthogonal projection onto the wavelet coefficients.

\begin{dfn}[Approximate symbol]
For any approximate Hilbert pair, $j\in\ZZ$, and $\cL\in\NN$, define the \emph{approximate symbol} as the following projection on~$L^2(\RR) \ot \CC^2$:
\begin{align}\label{eq:approx symbol}
  \tilde Q_{j,\cL}
% = \begin{pmatrix}
% \alpha_j^{h,\dagger} W^{(\cL),\dagger}_h & 0 \\
% 0 & \alpha_j^{g,\dagger} W^{(\cL),\dagger}_g
% \end{pmatrix}
% \left( P_w \ot \ket +\bra + \right)
% \begin{pmatrix}
% W^{(\cL)}_h \alpha_j^h & 0 \\
% 0 & W^{(\cL)}_g \alpha_j^g
% \end{pmatrix},
:= \alpha_j^\dagger W^{(\cL),\dagger} \bigl( P_w \ot \ket +\bra + \bigr) W^{(\cL)} \alpha_j,
\end{align}
% where $\alpha_j \colon L^2(\RR) \ot \CC^2 \to \ell^2(\ZZ) \ot \CC^2$, $\alpha_j := \alpha_j^h \op \alpha_j^g$ and $W^{(\cL)} \colon \ell^2(\ZZ) \ot \CC^2 \to \ell^2(\ZZ) \ot \CC^{\cL + 1} \ot \CC^2$, $W^{(\cL)} := W^{(\cL)}_h \op W^{(\cL)}_g$.
where~$\alpha_j := \alpha_j^h \op \alpha_j^g$ and~$W^{(\cL)} := W^{(\cL)}_h \op W^{(\cL)}_g$.
\end{dfn}

\noindent
The symbol~$\tilde Q_{j,\cL}$ should be seen as an approximation of the true symbol at scales ranging from~$2^{-j+1}$ to~$2^{-j+\cL}$.

On the circle~$\SS$ we proceed similarly, except that there is now a natural largest scale.
For periodic boundary conditions, we use the following symbol, which intuitively approximates the true symbol at scales above~$2^{-\cL}$:

\begin{dfn}[Approximate symbol, periodic case]
For any approximate Hilbert pair and~$\cL\in\NN$, define the \emph{approximate periodic symbol} as the following projection on~$L^2(\SS) \ot \CC^2$:
\begin{align}\label{eq:approx symbol per}
  \tilde Q^{\per}_{\cL} := \alpha_{\cL}^{\per,\dagger} W^{(\cL),\per,\dagger} \bigl( P_w \ot \ket+\bra+ + P_s \ot \ket{L}\bra{L} \bigr) W^{(\cL),\per} \alpha_{\cL}^{\per},
\end{align}
where~$\alpha_j^{\per} := \alpha_j^{h,\per} \op \alpha_j^{g,\per}$ and~$W^{(\cL),\per} := W^{(\cL),\per}_h \op W^{(\cL),\per}_g$ refer to the periodic versions as defined in \cref{subsec:periodic_wavelets}; $P_s$ projects onto the single scaling coefficient and~$\ket{L} := \frac{1}{\sqrt{2}} (\ket{0}-i\ket{1})$ to ensure compatibility with our choice for the Hilbert transform on constant functions.
\change{Had we made a different choice for the value of the Hilbert transform on constant functions this would only change the top level state (it is a well-known fact that the Dirac fermion with periodic boundary conditions has a two-fold ground state degeneracy).
We also observe that $\ket{L}\bra{L} = \frac12(\id + \sigma_2) = \frac12(\id + \gamma_{\chir})$ is the chiral projector, where $\gamma_{\chir} = \gamma_0\gamma_1 = \sigma_2$.
Had we chosen the convention $\sgn(0) = -1$ we would take the state $\ket{R} := \frac{1}{\sqrt{2}} (\ket{0}+i\ket{1})$, and $\ket{R}\bra{R} = \frac12(\id - \gamma_{\chir})$.}
\end{dfn}

\noindent
In \cref{subsec:periodic circuit} we explain how to deal with anti-periodic boundary conditions.

\begin{lem}\label{lem:move through discretization}
The following relation holds: $\alpha^g_j m(\lambda_{s,j})^{\dagger} \Htr = m(\mu_w) \alpha^g_j$. Similarly, in the periodic case it holds for all~$f \in L^{2}(\SS)$ with zero mean that $\alpha^{g,\per}_j m(\lambda^{\per}_{s,j})^{\dagger} \Htr f = m(\mu_{w,j}^{\per}) \alpha^{g,\per}_j f$.
\end{lem}

\begin{proof}
We want to show that $\alpha^g_j(m(\lambda_{s,j})^{\dagger} \Htr f) = m(\mu_w) \alpha^g_j(f)$ for $f \in L^2(\RR)$.
By rescaling $f$ it is easy to see that it suffices to show the result for $j = 0$.
We know that by \cref{eq:scaling_function_relation}, $\Htr = m(\lambda_s)m(\mu_w)$, so $\alpha^g_0(m(\lambda_{s})^{\dagger} \Htr f) = \alpha^g_0(m(\mu_w)f)$.
Next we take a Fourier transform and observe that
\begin{align*}
  \widehat{\alpha^g_0(f)}(\theta) = \frac{1}{2\pi} \sum_{n \in \ZZ} \overline{\hat{\phi}^g(\theta + 2\pi n)} \hat f(\theta + 2\pi n).
\end{align*}
Since $\mu_w$ is $2\pi$-periodic the result follows.
In the periodic case it holds that
\begin{align*}
  \widehat{\alpha^{g,\per}_j(f)}[n] = \sum_{m \in \ZZ}\overline{\hat \phi^{g,\per}[n + 2^j m]} \hat f[n + 2^j m].
\end{align*}
which similarly implies the desired result.
Note that the ambiguity in our choice of $\sgn(0)$ in the definition of $\Htr$ is not relevant if we assume that $f$ has mean zero.
\end{proof}

The following result shows that the symbols in \cref{eq:approx symbol,eq:approx symbol per} indeed yield reasonable approximations when restricted to appropriate functions.

\begin{prop}\label{prop:approx_symbol}
Consider an $\eps$-approximate Hilbert pair with scaling filters supported in~$\{0,\dots,M-1\}$.
\begin{enumerate}
\item Let $f\in H^1(\RR) \ot \CC^2$ with compact support.
Then, for all $j\in\ZZ$, $\cL\in\NN$, and $\cL'=0,\dots,\cL$,
\begin{align*}
  \lVert \left( Q - \tilde Q_{j,\cL} \right) f \rVert
\leq 3 \eps\cL' \lVert f \rVert + 2^{(j-\cL')/2} 7 \sqrt{M D(f)} B \lVert f \rVert + 2^{-j} 5 M^2 \lVert f' \rVert,
\end{align*}
where~$B := \max \{ \lVert\phi_g\rVert_\infty, \lVert\phi_h\rVert_\infty \}$.
\item Let $f\in H^1(\SS) \ot \CC^2$.
Then, for all $\cL\in\NN$ and $\cL'=0,\dots,\cL$,
\begin{align*}
  \lVert \left( Q^{\per} - \tilde Q^{\per}_{\cL} \right) f \rVert
\leq 2 \eps \cL' \lVert f\rVert + 2^{-\cL'} 9 M^2 \lVert f' \rVert.
\end{align*}
\end{enumerate}
\end{prop}
\noindent In \cref{thm:approximation}, we will describe how to choose~$j$ and~$\cL'$ optimally for a given number of layers~$\cL$.
\begin{proof}
(i) Let
\begin{equation}\label{eq:def Q_j}
\begin{aligned}
  Q_j
&:= \begin{pmatrix} \id & 0 \\ 0 & \Htr^{\dagger} m(\lambda_{s,j}) \end{pmatrix} \alpha_j^\dagger \left( \id_{\ell^2(\ZZ)} \ot \ket+\bra+ \right) \alpha_j \begin{pmatrix} \id & 0 \\ 0 & m(\lambda_{s,j})^{\dagger} \Htr \end{pmatrix} \\
&= \alpha_j^\dagger \begin{pmatrix} \id & 0 \\ 0 & m(\mu_w)^\dagger \end{pmatrix} \left( \id_{\ell^2(\ZZ)} \ot \ket+\bra+ \right) \begin{pmatrix} \id & 0 \\ 0 & m(\mu_w) \end{pmatrix} \alpha_j,
\end{aligned}
\end{equation}
where we used \cref{lem:move through discretization}.
Then, using the first formula,
\begin{align*}
  \lVert \left( Q - Q_j \right) f \rVert
% = \lVert \left( \id_{L^2(\RR)} \ot \ket+ \bra+ - \begin{pmatrix} \alpha_j^{h,\dagger} & 0 \\ 0 & m(\lambda_{s,j}) \alpha_j^{g,\dagger} \end{pmatrix} \left( \id_{\ell^2(\ZZ)} \ot \ket+\bra+ \right) \begin{pmatrix} \alpha_j^h & 0 \\ 0 & \alpha_j^g m(\lambda_{s,j})^{\dagger} \end{pmatrix} \right) \begin{pmatrix}\id & 0 \\ 0 & \Htr\end{pmatrix}f \rVert \\
% = \frac12 \lVert \left( \begin{pmatrix}\id_{L^2(\RR)} & \id_{L^2(\RR)} \\ \id_{L^2(\RR)} & \id_{L^2(\RR)}\end{pmatrix} - \begin{pmatrix}\alpha_j^{h,\dagger} \alpha_j^h & \alpha_j^{h,\dagger} \alpha_j^g m(\lambda_{s,j})^{\dagger} \\ m(\lambda_{s,j}) \alpha_j^{g,\dagger} \alpha_j^h & m(\lambda_{s,j}) \alpha_j^{g,\dagger} \alpha_j^g m(\lambda_{s,j})^{\dagger}\end{pmatrix} \right) \begin{pmatrix}\id & 0 \\ 0 & \Htr\end{pmatrix}f \rVert
% &= \frac12 \lVert \begin{pmatrix}\id - \alpha_j^{h,\dagger} \alpha_j^h & \id - \alpha_j^{h,\dagger} \alpha_j^g m(\lambda_{s,j})^{\dagger} \\ \id - m(\lambda_{s,j}) \alpha_j^{g,\dagger} \alpha_j^h & \id - m(\lambda_{s,j}) \alpha_j^{g,\dagger} \alpha_j^g m(\lambda_{s,j})^{\dagger}\end{pmatrix} \begin{pmatrix}\id & 0 \\ 0 & \Htr\end{pmatrix}f f \rVert \\
% &= \frac12 \lVert \begin{pmatrix}(\id - \alpha_j^{h,\dagger} \alpha_j^h) f_1 + (\id - \alpha_j^{h,\dagger} \alpha_j^g m(\lambda_{s,j})^{\dagger}) \Htr f_2 \\ (\id - m(\lambda_{s,j}) \alpha_j^{g,\dagger} \alpha_j^h) f_1 + (\id - m(\lambda_{s,j}) \alpha_j^{g,\dagger} \alpha_j^g m(\lambda_{s,j})^{\dagger}) \Htr f_2 \end{pmatrix} \rVert \\
&\leq \frac12 \Bigl( \lVert (\id - \alpha_j^{h,\dagger} \alpha_j^h) f_1 \rVert + \lVert (\id - m(\lambda_{s,j}) \alpha_j^{g,\dagger} \alpha_j^g m(\lambda_{s,j})^{\dagger}) \Htr f_2 \rVert \\
&\qquad + \lVert (\id - m(\lambda_{s,j}) \alpha_j^{g,\dagger} \alpha_j^h) f_1 \rVert + \lVert (\id - \alpha_j^{h,\dagger} \alpha_j^g m(\lambda_{s,j})^{\dagger}) \Htr f_2 \rVert \Bigr) \\
&= \frac12 \Bigl( \lVert (\id - P_j^h) f_1 \rVert + \lVert (\id - P_j^g) m(\lambda_{s,j})^{\dagger} \Htr f_2 \rVert \\
&\qquad + \lVert (\id - m(\lambda_{s,j}) \alpha_j^{g,\dagger} \alpha_j^h) f_1 \rVert + \lVert (\id - \alpha_j^{h,\dagger} \alpha_j^g m(\lambda_{s,j})^{\dagger}) \Htr f_2 \rVert \Bigr) \\
&\leq \frac12 \Bigl( \lVert (\id - P_j^h) f_1 \rVert + \lVert (\id - P_j^h) \Htr f_2 \rVert + \lVert (\id - P_j^g) m(\lambda_{s,j})^\dagger f_1 \rVert + \lVert (\id - P_j^g) m(\lambda_{s,j})^{\dagger} \Htr f_2 \rVert \\
&\qquad + \lVert (\alpha_j^g m(\lambda_{s,j})^\dagger - \alpha_j^h) f_1 \rVert + \lVert (\alpha_j^g m(\lambda_{s,j})^\dagger - \alpha_j^h) \Htr f_2 \rVert \Bigr).
\end{align*}
The norms in the first line can be upper-bounded by using \cref{lem:UV} (for the second, note that~$\lVert (m(\lambda_{s,j})^{\dagger} f_i)' \rVert = \lVert f_i' \rVert$ for $i=1,2$).
For the norms in the second line we use \cref{lem:scaling error}.
Together, we find that
\begin{equation}\label{eq:q_vs_qj_bound}
\begin{aligned}
  \lVert \left( Q - Q_j \right) f \rVert
&\leq \frac12 \Bigl( 2^{-j} C_{\UV} \left(2 \lVert f'_1 \rVert + 2 \lVert \Htr f_2' \rVert\right) + 2^{-j} C_\chi \left( \lVert f_1'\rVert + \lVert \Htr f_2'\rVert \right) \Bigr) \\
&\leq \frac12 2^{-j} \left( 2 C_{\UV} + C_\chi \right) \sqrt2 \lVert f' \rVert
\leq 2^{-j} 5 M^2 \lVert f' \rVert
\end{aligned}
\end{equation}
where we used that the Hilbert transform preserves the norm of the derivative ($\lVert \Htr f_2' \rVert = \lVert f_2' \rVert$).

Next, we define
\begin{align*}
  Q_{j,\cL} := \alpha_j^\dagger \begin{pmatrix} \id & 0 \\ 0 & m(\mu_w)^{\dagger} \end{pmatrix} \left( W_h^{(\cL),\dagger} P_w W_h^{(\cL)} \ot \ket{+}\bra{+} \right) \begin{pmatrix} \id & 0 \\ 0 & m(\mu_w) \end{pmatrix} \alpha_j.
\end{align*}
Using the second expression in \cref{eq:def Q_j}, we can then split the remaining error as
\begin{equation}\label{eq:q_j_vs_q_MERA}
  \lVert \left( Q_j - \tilde Q_{j,\cL} \right) f \rVert
\leq \lVert \left( Q_j - Q_{j,\cL'} \right) f \rVert
+ \lVert \left( Q_{j,\cL'} - \tilde Q_{j,\cL'} \right) f \rVert
+ \lVert \left( \tilde Q_{j,\cL'} - \tilde Q_{j,\cL} \right) f \rVert
\end{equation}
The third term in \cref{eq:q_j_vs_q_MERA} can be estimated using \cref{lem:IR}:
\begin{align*}
  \lVert \left( \tilde Q_{j,\cL'} - \tilde Q_{j,\cL} \right) f \rVert
&\leq \lVert \alpha_j^\dagger W^{(\cL'),\dagger} \bigl( P_s \ot \ket +\bra + \bigr) W^{(\cL')} \alpha_j f \rVert
\leq \lVert P_{j-\cL'} f \rVert \\
&\leq 2^{(j-\cL')/2} \sqrt{M D(f)} \max \{ \lVert\phi_g\rVert_\infty, \lVert\phi_h\rVert_\infty \} (\lVert f_1\rVert + \lVert f_2\rVert) \\
&\leq 2^{(j-\cL')/2} \sqrt 2 \sqrt{M D(f)} B \lVert f\rVert.
\end{align*}
For the second term in \cref{eq:q_j_vs_q_MERA}, we use \cref{eq:filtererror_wavelet} in \cref{lem:filtererror}:
\begin{align*}
  \lVert Q_{j,\cL'} - \tilde Q_{j,\cL'} \rVert
% &\leq \lVert \left( \alpha_j^\dagger \begin{pmatrix} W_h^{(\cL),\dagger} & 0 \\ 0 & m(\mu_w)^{\dagger} W_h^{(\cL),\dagger} \end{pmatrix} \left( P_w  \ot \ket{+}\bra{+} \right) \begin{pmatrix} W_h^{(\cL)} & 0 \\ 0 & W_h^{(\cL)} m(\mu_w) \end{pmatrix} \alpha_j - \alpha_j^\dagger W^{(\cL),\dagger} \bigl( P_w \ot \ket +\bra + \bigr) W^{(\cL)} \alpha_j \right) f \rVert \\
% &\leq \lVert P_w (W_h^{(\cL)} m(\mu_w) - W_g^{(\cL)}) \alpha_j^g f_2\rVert + \lVert \alpha_j^\dagger \begin{pmatrix} 0 & 0 \\ 0 & m(\mu_w)^{\dagger} W_h^{(\cL),\dagger} - W_g^{(\cL),\dagger} \end{pmatrix} \left( P_w  \ot \ket{+}\bra{+} \right) \begin{pmatrix} W_h^{(\cL)} & 0 \\ 0 & W_g^{(\cL)} \end{pmatrix} \alpha_j f \rVert \\
\leq \lVert P_w (W_h^{(\cL)} m(\mu_w) - W_g^{(\cL)}) \rVert + \lVert m(\mu_w)^{\dagger} W_h^{(\cL),\dagger} - W_g^{(\cL),\dagger} \rVert
\leq 2\eps\cL'
\end{align*}
Finally, for the first term in \cref{eq:q_j_vs_q_MERA}, we would like to apply \cref{lem:IR}, but we need to be careful because~$m(\mu_w)$ does not preserve compact support.
So we first use \cref{eq:filtererror_scaling} in \cref{lem:filtererror} to get rid of $m(\mu_w)$, and then apply \cref{lem:IR}:
\begin{align*}
  \lVert \left( Q_j - Q_{j,\cL'} \right)f \rVert
&= \lVert \left( P_s \ot \ket+\bra+ \right) W_h^{(\cL')} \begin{pmatrix} \id & 0 \\ 0 & m(\mu_w) \end{pmatrix} \alpha_j f \rVert \\
&\leq \lVert P_s ( W_h^{(\cL')} m(\mu_w) - W_g^{(\cL')} ) \alpha^g_j f \rVert + \lVert \left( P_s \ot \id \right) W^{(\cL')} \alpha_j f \rVert \\
&\leq \eps \cL' \lVert \alpha^g_j f_2 \rVert + 2 \lVert P_s W_g^{(\cL')} \alpha^g_j f \rVert + \lVert \left( P_s \ot \id \right) W^{(\cL')} \alpha_j f \rVert \\
&\leq \eps \cL' \lVert f \rVert + 3 \left( \lVert P^g_{j-\cL'} f_2 \rVert + \lVert P^h_{j-\cL'} f_1 \rVert \right) \\
&\leq \eps \cL' \lVert f \rVert + 2^{(j-\cL')/2} 5 \sqrt{M D(f)} B \lVert f\rVert.
\end{align*}
Thus, we can upper bound \cref{eq:q_j_vs_q_MERA} by
\begin{align}\label{eq:q_j_vs_q_MERA_bound}
  \lVert \left( Q_j - \tilde Q_{j,\cL} \right) f \rVert
\leq 3 \eps\cL' \lVert f\rVert + 2^{(j-\cL')/2} 7 \sqrt{M D(f)} B \lVert f \rVert.
% \leq 7 \left( 2^{(j-\cL')/2} \sqrt{M D(f)} B + \eps\cL' \right) \lVert f \rVert.
\end{align}
Combining \cref{eq:q_vs_qj_bound,eq:q_j_vs_q_MERA_bound} we obtain the desired bound.

(ii) Using $\phi^{g,\per}_{0,1} = \phi^{h,\per}_{0,1} = \one$, it is easy to see that our choice of input to the scaling layer ensures that
\begin{align*}
  Q^{\per} \one = \tilde Q^{\per}_{\cL} \one,
\end{align*}
so we can assume without loss of generality that~$f$ has zero mean or, equivalently, that~$P_0 f = 0$ and we may apply \cref{lem:move through discretization}.
Similarly as before (but without having to worry about an IR cut-off), we introduce
\begin{align*}
  Q^{\per}_{\cL} := \alpha_{\cL}^{\per,\dagger} \begin{pmatrix} \id & 0 \\ 0 & m(\mu_w)^{\dagger} \end{pmatrix} \left( W_h^{(\cL),\per,\dagger} P_w W_h^{(\cL),\per} \ot \ket{+}\bra{+} \right) \begin{pmatrix} \id & 0 \\ 0 & m(\mu_w) \end{pmatrix} \alpha_{\cL}^{\per}
\end{align*}
and use a triangle inequality
\begin{align*}
  \lVert \left( Q^{\per} - \tilde Q^{\per}_{\cL} \right) f \rVert
\leq \lVert \left( Q^{\per} - Q^{\per}_{\cL'}  \right) f \rVert + \lVert \left( Q^{\per}_{\cL'} - \tilde Q^{\per}_{\cL'} \right) f \rVert + \lVert \left( \tilde Q^{\per}_{\cL'}- \tilde Q^{\per}_{\cL} \right) f \rVert.
\end{align*}
For the first term, we use \cref{lem:UV,lem:scaling error} and obtain
\begin{align*}
  \lVert \left( Q^{\per} - Q^{\per}_{\cL'} \right) f \rVert \leq 2^{-\cL'} 5 M^2 \lVert f' \rVert,
\end{align*}
in complete analogy to \cref{eq:q_vs_qj_bound}.
For the second term, note that we can ignore the scaling part in \cref{eq:approx symbol per} since we assumed that~$P_0 f = 0$.
Thus, we can use \cref{eq:filtererror periodic wavelet} in \cref{lem:filtererror periodic} and find
\begin{align*}
  \lVert Q^{\per}_{\cL'} - \tilde Q^{\per}_{\cL'} \rVert \leq 2 \eps \cL'.
\end{align*}
Finally, the third term can be upper bounded by using \cref{lem:UV},
\begin{align*}
  \lVert \left( \tilde Q^{\per}_{\cL'}- \tilde Q^{\per}_{\cL} \right) f \rVert
\leq \lVert (\id - P^{\per}_{\cL'}) f \rVert
\leq 2^{-\cL'} \sqrt 2 C_{\UV} \lVert f' \rVert
\leq 2^{-\cL'} 4 M^2 \lVert f' \rVert
\end{align*}
(note that here we are comparing different UV cut-offs, in contrast to before).
By combining these bounds we obtain the desired result.
\end{proof}

\noindent
If we keep track of all the wavelet constants in the proof of \cref{prop:approx_symbol} rather than bounding them in terms of~$M$ then the proof shows in fact the bound
\begin{align}\label{eq:tracking constants}
  \lVert \left( Q - \tilde Q_{j,\cL} \right) f \rVert
\leq 3 \eps\cL' \lVert f \rVert + 2^{(j-\cL')/2} 7 \sqrt{D(f)} C_{\IR} \lVert f \rVert + 2^{-j} \frac{1}{\sqrt2}(2C_{\UV} + C_{\chi}) \lVert f' \rVert,
\end{align}
which will be useful if we want to investigate numerically how fast our error bounds converge.

%-----------------------------------------------------------------------------
\subsection{Approximation bounds for correlation functions}\label{subsec:approx}
%-----------------------------------------------------------------------------
The bounds on the approximate symbol from \cref{prop:approx_symbol} can be used to estimate the approximation error for correlation functions.
We start with the Dirac fermion on the line, whose vacuum state is the quasi-free state~$\omega_Q$ with symbol~$Q$ defined in \cref{eq:free_fermion_symbol}.
We are interested in correlation functions of the form involving the smeared Dirac field~$\Psi(f)$ and normal-ordered quadratic operators.
In the Fock representation, the two-component Dirac field is implemented by the operators~$\Psi(f) := a_Q(f)$, defined as in \cref{eq:CAR_rep}, and the normal-ordered quadratic operators the~$\d\Gamma_Q(A)$ defined in \cref{subsec:second quantization}.
Thus, we wish to approximate correlation functions of the form
\begin{align}\label{eq:exact correlation}
  G(\{O_i\}) := \braket{\Omega|O_1 \cdots{} O_n|\Omega},
\end{align}
where each~$O_i$ is either a component of~$\Psi(f)$ or its adjoint~$\Psi^\dagger(f)$, or a normal-ordered operator~$\d\Gamma_Q(A)$.

We would like to approximate such correlation functions by using the symbol~$\tilde Q_{j,\cL}$ defined in \cref{eq:approx symbol}.
Thus we fix an approximate Hilbert pair, $j\in\ZZ$, and $\cL>0$, and consider
\begin{align}\label{eq:approx correlation}
  \tilde G_{j,\cL}(\{O_i\}) := \braket{\Omega|\tilde O_1 \cdots{} \tilde O_n|\Omega},
\end{align}
where the~$\tilde O_i$ are obtained from the~$O_i$ by replacing~$\Psi(f)$ by~$\tilde\Psi(f) := a_{\tilde Q_{j,\cL}}(P_j f)$ and~$\d\Gamma_Q(A)$ by~$\d\Gamma_{\tilde Q_{j,\cL}}(P_j A P_j)$, respectively.

On the circle, we denote the corresponding correlation functions for periodic boundary conditions by~$G^{\per}(\{O_i\})$ and $\tilde G^{\per}_{\cL}(\{O_i\})$, respectively.
They are defined in terms of the symbol~$Q^{\per}$ and its approximation~$\tilde Q^{\per}_{\cL}$ defined in \cref{eq:approx symbol per}.
We discuss anti-periodic boundary conditions in \cref{subsec:periodic circuit} below.

The following theorem is our main technical result (already stated informally in \cref{thm:introduction}).
It states that~$G(\{O_i\}) \approx \tilde G_{j,\cL}(\{O_i\})$ under appropriate conditions (and similarly in the periodic case).

\begin{thm}\label{thm:approximation}
Consider an $\eps$-approximate Hilbert pair with scaling filters supported in~$\{0,\dots,M-1\}$, scaling functions bounded by~$B$, and $\eps\in(0,1)$.
\begin{enumerate}
\item
Let $f_1,\dots,f_n$ be compactly supported functions in~$H^1(\RR) \ot \CC^2$ and let~$A_1,\dots,A_m$ be Hilbert-Schmidt integral operators with compactly supported kernels in~$H^1(\RR^2) \ot M_2(\CC)$, all with $L^2$-norm at most~1.
Let $O_i=\Psi(f_i)$ or~$\Psi^\dagger(f_i)$ for~$i=1,\dots,n$ and~$O_{n+i}=\d\Gamma_Q(A_i)$ for~$i=1,\dots,m$.
Then we can find, for every~$\cL>0$, a scale~$j\in\ZZ$ such that
\begin{align*}
  \left\lvert G(\{O_i\}) - \tilde G_{j,\cL}(\{O_i\}) \right\rvert
\leq 8^m m! (n+m) \left( \change{6} \eps \log_2\frac{3 C^3 D}\eps + C \, D^{1/3} 2^{-\frac{\cL}3} \right).
\end{align*}
The constant~$C := \change{14} (\sqrt{2M}B + M^2)$ depends only on the Hilbert pair, and the constant~$D := \max\{1, d(f,A) D(f,A)\}$ depends only on the smoothness and support of the smearing functions, where~$d(f, A) := \max \{ \lVert f_i' \rVert, \lVert \nabla A_i \rVert \}$ and~$D(f,A) := \max \{ D(f_i), D(A_i) \}$; $\nabla A_i$ denotes the gradient of the kernel of~$A_i$ and~$D(A_i)$ denotes the side length of the smallest square supporting the kernel.
\item
Let $f_1,\dots,f_n$ be functions in $H^1(\SS) \ot \CC^2$ and let~$A_1,\dots,A_m$ be Hilbert-Schmidt integral operators with kernels in~$H^1(\SS)\ot M_2(\CC)$, all with $L^2$-norm at most~1.
Then we have, for every~$\cL>0$, that
\begin{align*}
  \left\lvert G^{\per}(\{O_i\}) - \tilde G^{\per}_{\cL}(\{O_i\}) \right\rvert
\leq 8^m m! (n+m) \left( 6 \eps \log_2\frac{59 M^2 D}{\eps} + 26 M^2 D 2^{-\cL} \right).
\end{align*}
The constant~$D$ is defined as~$D := \max \{ 1, \lVert f_i' \rVert, \lVert \nabla A_i \rVert \}$, with~$\nabla A_i$ the gradient of the kernel of~$A_i$.
\end{enumerate}
\end{thm}

\noindent
Before giving the proof, we comment on some aspects of the theorem.
The main idea behind the theorem and its proof is that the approximation of the correlation functions is accurate as long as the approximation to the symbol is accurate on the scales at which the system is probed.
Quite intuitively, large support requires us to accurately approximate large scales, and strong fluctuations (large derivatives) require accuracy at small scales.
The constant $D=\max\{1,d(f,A)D(f,A)\}$ reflects the number of scales needed for accurate approximation for given smearing functions~$f_i$ and kernels~$A_i$.
Intuitively, $D$ is invariant under dilatations, reflecting the scale invariance of the theory.
% \begin{align*}
%   F(x) = 2^{c/2} f(2^c x)
% \quad\Rightarrow\quad
% \lVert F\rVert_2 = \lVert f\rVert_2, \quad
% D(F) = 2^{-c} D(f), \quad
% F'(x) = 2^c 2^{c/2} f'(2^c x), \quad
% \lVert F'\rVert = 2^c \lVert f'\rVert.
% \end{align*}
On the circle~$\SS$, there is a natural largest scale, allowing for a slightly simpler formulation.
While we state the theorem for the Dirac fermion, \cref{prop:approx_symbol} readily implies a similar result for correlation functions of the Majorana fermion (\cref{subsec:selfdual}).

Our assumptions on the operators~$A_i$ imply that they are in fact trace class.
Thus, the operators~$\d\Gamma(A_i)$ and~$\d\Gamma_Q(A_i)$ can be directly defined in the CAR algebra, so we could work directly with the state~$\omega_Q$ on the algebra rather than in the Fock space representation.
Such an approach could improve the dependence on~$m$ of the bounds, since one can estimate $\lVert \d \Gamma_Q(A_i) \rVert = \lVert \d \Gamma(A_i) - \omega_Q(\d \Gamma(A_i)) \rVert \leq 2\lVert A_i \rVert_1$.

While in \cref{thm:approximation} we order the insertions in~$G(\{O_i\})$ in a particular way, other orderings are also possible.
This follows either from using the commutation relations (leading to terms depending on~$A_k f_l$) or by directly adjusting the proof (leading to a change in the dependence on~$n$ and~$m$, since in the proof we would insert the particle-number projections~$\Pi_{2k}$ in different places).

Finally, we note that in the proofs of both \cref{prop:approx_symbol,thm:approximation} we bound the wavelet parameters~$C_{\UV}$, $C_{\IR}$, and~$C_\chi$ from \cref{lem:UV,lem:IR,lem:scaling error} in terms of the support~$M$ to arrive at simpler expressions.
Sharper numerical bounds can be obtained by using~$C_{\UV}$, $C_{\IR}$, and~$C_\chi$ directly (see \cref{tab:constants}).
If one tracks these constants throughout the proof, using \cref{eq:tracking constants} rather than \cref{prop:approx_symbol}, one sees that $C$ can be taken to be
\begin{align}\label{eq:sharp constant}
C = 2(4C_{\UV} + C_{\chi}) + \change{20}C_{\IR}.
\end{align}
The precise numerical constants are not very important, but we can use this as in \cref{fig:approximation error} to illustrate \cref{thm:approximation} numerically for two-point functions (using \cref{tab:constants} to evaluate \cref{eq:sharp constant}).
We see that, even for relatively small circuit depth, our \cref{thm:approximation} combined with numerical results of \cref{tab:constants} yields a reasonably small upper bound on the approximation error.

\begin{proof}[Proof of \cref{thm:approximation}]
(i)~We first estimate the error in the correlation functions in terms of the corresponding symbols for fixed~$j\in\ZZ$ and~$\cL'\in\{0,\dots,\cL\}$.
We define $Q_- := Q$, $Q_+ := \id - Q$, $\tilde Q_- := \tilde Q_{j,\cL}$, and $\tilde Q_+ := P_j - \tilde Q_{j,\cL}$~(!).
For $i=1,\dots,n$,
\begin{align*}
  \lVert O_i - \tilde O_i \rVert
= \lVert a_Q(f_i) - a_{\tilde Q_{j,\cL}}(P_j f_i) \rVert
% = \lVert a_0\bigl((\id-Q)f_i\bigr)- a_0\bigl((\id-\tilde Q_{j,\cL}) P_j f_i\bigr) + a_0^\dagger\bigl(\overline{Qf_i}\bigr) - a_0^\dagger\bigl(\overline{\tilde Q_{j,\cL}P_jf_i}\bigr) \rVert \\
\leq \lVert (Q_+ - \tilde Q_+) f_i\rVert + \lVert (Q_- - \tilde Q_-) f_i\rVert,
\end{align*}
where we used the definition of the operators~$\tilde O_i$ described above, \cref{eq:CAR_rep} and that~$\tilde Q_{j,\cL} P_j = P_j \tilde Q_{j,\cL} = \tilde Q_{j,\cL}$.
By \cref{prop:approx_symbol}, we have the estimate
\begin{align*}
  \lVert (Q_- - \tilde Q_-) f_i\rVert
\leq 3 \eps\cL' \lVert f_i \rVert + 2^{(j-\cL')/2} 7 \sqrt{M D(f_i)} B \lVert f_i \rVert + 2^{-j} 5 M^2 \lVert f_i' \rVert.
\end{align*}
Moreover, using \cref{lem:UV},
\begin{align*}
  \lVert (Q_+ - \tilde Q_+) f_i\rVert
% = \lVert (\id - Q_- - P_j + \tilde Q_-) f_i \rVert
\leq \lVert P_j f_i - f_i \rVert + \lVert (Q_- - \tilde Q_-) f_i \rVert
% \leq 2^{-j} 2M^2 \lVert f'_{i,1} \rVert + 2^{-j} 2M^2 \lVert f'_{i,2} \rVert + \lVert (Q_- - \tilde Q_-) f_i \rVert.
\leq 2^{-j} 4 M^2 \lVert f'_i \rVert + \lVert (Q_- - \tilde Q_-) f_i \rVert.
\end{align*}
Thus we find that
% \begin{align*}
%   \lVert (Q_\delta - \tilde Q_\delta) f_i\rVert
% \leq 3 \eps\cL' \lVert f_i \rVert + 2^{(j-\cL')/2} 7 \sqrt{M D(f_i)} B \lVert f_i \rVert + 2^{-j} 9 M^2 \lVert f_i' \rVert
% \end{align*}
% and hence
\begin{align}\label{eq:estimate_cr_an}
  \lVert O_i - \tilde O_i \rVert
% \leq 6 \eps\cL' \lVert f_i \rVert + 2^{(j-\cL')/2} 14 \sqrt{M D(f_i)} B \lVert f_i \rVert + 2^{-j} 14 M^2 \lVert f_i' \rVert.
\leq 6 \eps\cL' + 2^{(j-\cL')/2} 14 \sqrt{M D(f_i)} B + 2^{-j} 14 M^2 \lVert f_i' \rVert
\end{align}
using~$\lVert f_i\rVert\leq1$.
For $i=n+1,\dots,n+m$, if we let~$\Pi_n$ denote the projection onto the~$n$-particle subspace of the Fock space then by \cref{eq:second quant vs symbol} we have the bound
\begin{align*}
  \lVert \left( O_i - \tilde O_i \right) \Pi_{2k} \rVert
&\leq 4(2k+2) \max_{\delta=\pm} \{ \lVert Q_\delta A_i Q_\delta - \tilde Q_\delta A_i \tilde Q_\delta \rVert, \lVert Q_\delta A_i Q_{-\delta} - \tilde Q_\delta A_i \tilde Q_{-\delta} \rVert_2 \} \\
&\leq 4(2k+2) \max_{\delta=\pm} \{
\lVert (Q_\delta - \tilde Q_\delta) A_i \rVert_2 + \lVert A_i (Q_\delta - \tilde Q_\delta) \rVert_2 \}.
\end{align*}
To estimate $\lVert (Q_\delta - \tilde Q_\delta) A_i \rVert_2$, let $\{e_n\}$ be an orthonormal basis of $L^2(\RR) \ot \CC^2$, so
\begin{align*}
  \lVert (Q_\delta - \tilde Q_\delta) A_i \rVert_2^2 &= \sum_n \lVert (Q_\delta - \tilde Q_\delta) A_i e_n \rVert^2 \\
  &\leq \sum_n \left( \left(3 \eps\cL'  + 2^{(j-\cL')/2} 7 \sqrt{M D(A_i)} B \right)\lVert A_i e_n \rVert + 2^{-j} 9 M^2 \lVert (A_i e_n)' \rVert\right)^2
\end{align*}
using \cref{prop:approx_symbol} and \cref{lem:UV} (for $\delta=+$) and the fact that, by our assumption on the support of the kernel of~$A_i$, the support of~$A_i e_n$ is contained in an interval of size~$D(A_i)$.
Since~$A_i$ has a kernel~$h_i$ in $H^1(\RR) \ot M_2(\CC)$, it holds that $(A_i e_n)' = (\partial_x A_i) e_n$, where~$\partial_x A_i$ denotes the integral operator with kernel~$\partial_x h_i$.
Thus, we conclude \change{using Cauchy-Schwarz}
\begin{align*}
  \lVert (Q_\delta - \tilde Q_\delta) A_i \rVert_2 \leq 3 \eps\cL' \lVert A_i \rVert_2 + 2^{(j-\cL')/2} 7 \sqrt{M D(A_i)} B \lVert A_i \rVert_2 + 2^{-j} 9 M^2 \lVert \partial_x A_i \rVert_2.
\end{align*}
Since the adjoint of an integral operator has the transposed and conjugated kernel, we obtain the same bound on~$\lVert A_i (Q_\delta - \tilde Q_\delta) \rVert_2 = \lVert (Q_\delta - \tilde Q_\delta) A_i^\dagger \rVert_2$ but with~$\lVert\partial_y A_i\rVert$ in place of~$\lVert\partial_x A_i\rVert$,
and hence
\begin{align}\label{eq:estimate_sq_ops_bound}
  \lVert \left( O_i - \tilde O_i \right) \Pi_{2k} \rVert
% \leq 4(2k+2) \left( 6 \eps\cL' \lVert A_i \rVert + 2^{(j-\cL')/2} 14 \sqrt{M D(A_i)} B \lVert A_i \rVert + 2^{-j} 9 M^2 \left( \lVert \partial_x A_i \rVert + \lVert \partial_y A_i \rVert \right) \right)
% 9 * \sqrt{2} < 14
% \leq 4(2k+2) \left( 6 \eps\cL' \lVert h_i \rVert + 2^{(j-\cL')/2} 14 \sqrt{M D(A_i)} B \lVert h_i \rVert + 2^{-j} 14 M^2 \lVert \nabla A_i \rVert \right).
\leq 4(2k+2) \left( 6 \eps\cL' + 2^{(j-\cL')/2} 7 \sqrt{M D(A_i)} B + 2^{-j} 14 M^2 \lVert \nabla A_i \rVert_2 \right)
\end{align}
using~$\lVert A_i \rVert_2 = \lVert h_i \rVert \leq 1$, and where we have written $\nabla A_i$ for the operator which has the gradient of $h_i$ as kernel.
To estimate the error in the correlation functions, we use a telescoping sum
\begin{align}\label{eq:telescope cor}
  \left\lvert G(\{O_i\}) - \tilde G_{j,\cL}(\{O_i\}) \right\rvert \leq \sum_{i=1}^{n+m} \delta_i,
\end{align}
where
\begin{align*}
  \delta_i = \lvert\braket{\Omega|O_1 \cdots O_{i-1} (O_i - \tilde O_i) \tilde O_{i+1} \cdots \tilde O_{n+m}|\Omega}\rvert.
\end{align*}
Now, $\lVert O_i\rVert \leq 1$ for $i=1,\dots,n$ by~$\lVert f_i\rVert\leq1$.
For $i=1,\dots,m$, we can replace~$O_{n+i}$ by $O_{n+i} \Pi_{2(m-i)}$, and similarly for~$\tilde O_{n+i}$.
Since~$\lVert O_{n+i} \Pi_{2(m-i)} \rVert \leq 8(m-i+1)$ by \cref{eq:operator_est} and~$\lVert A_{n+i}\rVert_2\leq1$, we find that, for~$i=1,\dots,n$,
\begin{align*}
  \delta_i
&\leq 8^m m! \left( 6 \eps\cL' + 2^{(j-\cL')/2} 14 \sqrt{M D(f_i)} B + 2^{-j} 14 M^2 \lVert f_i' \rVert \right) \\
\end{align*}
by \cref{eq:estimate_cr_an} and, for~$i=1,\dots,m$,
\begin{align*}
  \delta_{n+i}
% \leq 8^m (m-j)! \lVert \left( O_{n+j} - \tilde O_{n+j} \right) \Pi_{2(m-j)} \rVert
&\leq 8^m m! \left( 6 \eps\cL' + 2^{(j-\cL')/2} 14 \sqrt{M D(A_i)} B + 2^{-j} 14 M^2 \lVert \nabla A_i \rVert \right)
\end{align*}
by \cref{eq:estimate_sq_ops_bound}.
If we plug these bounds into \cref{eq:telescope cor} we obtain
\begin{equation}\label{eq:G error unoptimized}
\begin{aligned}
  &\quad \left\lvert G(\{O_i\}) - \tilde G_{j,\cL}(\{O_i\}) \right\rvert \\
% \leq 8^m m! \left( \sum_{i=1}^n \left( 6 \eps\cL' + 2^{(j-\cL')/2} 14 \sqrt{M D(f_i)} B + 2^{-j} 14 M^2 \lVert f_i' \rVert \right)
% + \sum_{i=1}^m \left( 6 \eps\cL' + 2^{(j-\cL')/2} 14 \sqrt{M D(A_i)} B + 2^{-j} 14 M^2 \lVert \nabla A_i \rVert \right) \right)
% \leq 8^m m! \left( 6(n+m)\eps\cL' + \sum_{i=1}^n \left( 2^{(j-\cL')/2} 14 \sqrt{M D(f_i)} B + 2^{-j} 14 M^2 \lVert f_i' \rVert \right)
% + \sum_{i=1}^m \left( 2^{(j-\cL')/2} 14 \sqrt{M D(A_i)} B + 2^{-j} 14 M^2 \lVert \nabla A_i \rVert \right) \right)
% \leq 8^m m! \left( 10(n+m)\eps\cL' + \sum_{i=1}^n \left( 2^{(j-\cL')/2} 14 \sqrt{M D(f, A)} B + 2^{-j} 14 M^2 d(f, A) \right)
% + \sum_{i=1}^m \left( 2^{(j-\cL')/2} 14 \sqrt{M D(f, A)} B + 2^{-j} 14 M^2 d(f, A) \right) \right)
&\leq 8^m m! (n+m) \left( 6\eps\cL' + 2^{(j-\cL')/2} 14 \sqrt{M D(f, A)} B + 2^{-j} 14 M^2 d(f, A) \right),
\end{aligned}
\end{equation}
where we used the definitions of~$D(f,A)$ and~$d(f,A)$.
We have thus obtained a bound on the approximation error which holds for all~$j\in\ZZ$ and~$\cL'=0,\dots,\cL$.

We now choose~$j$ and~$\cL'$ to obtain that vanishes as the number of layers~$\cL$ increases and~$\eps$ goes to zero.
We first choose $j = \lceil \frac{\cL'}3 + \frac13\log_2\frac{d(f,A)^2}{D(f,A)} \rceil$ and obtain
\begin{align*}
  \left\lvert G(\{O_i\}) - \tilde G_{j,\cL}(\{O_i\}) \right\rvert
% \leq 8^m m! (n+m) \left( 6\eps\cL' + 2^{(\frac{\cL'}3 + \frac13\log_2\frac{d(f,A)^2}{D(f,A)} + 1 - \cL')/2} 14 \sqrt{M D(f, A)} B + 2^{-(\frac{\cL'}3 + \frac13\log_2\frac{d(f,A)^2}{D(f,A)})} 14 M^2 d(f, A) \right),
% \\ \leq 8^m m! (n+m) \left( 6\eps\cL' + 2^{-\frac{\cL'}3} 2^{\frac16\log_2\frac{d(f,A)^2}{D(f,A)}} \sqrt2 \, 14 \sqrt{M} \sqrt{D(f, A)} B + 2^{-\frac{\cL'}3} 2^{\frac13\log_2\frac{D(f,A)}{d(f,A)^2}} 14 M^2 d(f, A) \right),
% \\ \leq 8^m m! (n+m) \left( 6\eps\cL' + 2^{-\frac{\cL'}3} \sqrt2 \, 14 \sqrt{M} B \left( \frac{d(f,A)^2}{D(f,A)} \right)^{\frac16} \sqrt{D(f, A)} + 2^{-\frac{\cL'}3} 14 M^2 \left(\frac{D(f,A)}{d(f,A)^2} \right)^{\frac13} d(f, A) \right),
% \\ \leq 8^m m! (n+m) \left( 6\eps\cL' + 2^{-\frac{\cL'}3} \sqrt2 \, 14 \sqrt{M} B \left( \frac{d(f,A)^{2/3}}{D(f,A)^{1/3}} \right)^{\frac12} \sqrt{D(f, A)} + 2^{-\frac{\cL'}3} 14 M^2 \left(\frac{D(f,A)^{1/3}}{d(f,A)^{2/3}} \right) d(f, A) \right),
% \\ \leq 8^m m! (n+m) \left( 6\eps\cL' + 2^{-\frac{\cL'}3} \sqrt2 \, 14 \sqrt{M} B d(f,A)^{1/3} D(f,A)^{1/3} + 2^{-\frac{\cL'}3} 14 M^2 D(f,A)^{1/3} d(f,A)^{1/3} \right),
&\leq 8^m m! (n+m) \left( 6\eps\cL' + 14 (\sqrt{2M}B + M^2) d(f,A)^{1/3} D(f,A)^{1/3} 2^{-\frac{\cL'}3} \right) \\
&= 8^m m! (n+m) \left( 6\eps\cL' + C \, D^{1/3} 2^{-\frac{\cL'}3} \right),
\end{align*}
using the definitions of~$C$ and~$D$.
We now choose~$\cL' = \min \{ \cL, \lceil\log_2(C^3 D/\eps)\rceil \}$, which is always nonnegative, and obtain
\begin{align*}
  \left\lvert G(\{O_i\}) - \tilde G_{j,\cL}(\{O_i\}) \right\rvert
&\leq 8^m m! (n+m) \left( 6\eps\left( \log_2\frac{C^3 D}\eps + 1\right) + \max \{ C \, D^{1/3} 2^{-\frac{\cL}3}, \eps\} \right) \\
% &\leq 8^m m! (n+m) \left( 10\eps\log_2\frac{C^3 D}\eps + C \, D^{1/3} 2^{-\frac{\cL}3} + 11\eps \right) \\
&\leq 8^m m! (n+m) \left( 6 \eps \log_2\frac{3 C^3 D}\eps + C \, D^{1/3} 2^{-\frac{\cL}3} \right),
\end{align*}
which proves the desired bound.

(ii) The proof for the circle goes along the same lines using the corresponding bound from \cref{prop:approx_symbol} and~$j=\cL$.
Instead of \cref{eq:estimate_cr_an,eq:estimate_sq_ops_bound}, we find that, for all~$\cL'\in\{0,\dots,\cL\}$ and for~$i = 1, \ldots, n$,
\begin{align*}
  \lVert O_i - \tilde O_i \rVert
\leq 4\eps\cL' + 2^{-\cL'} 18 M^2 \lVert f_i' \rVert + 2^{-\cL} 4 M^2 \lVert f_i' \rVert
\leq 4\eps\cL' + 2^{-\cL'} 22 M^2 \lVert f_i' \rVert,
\end{align*}
while for $i = n+1,\ldots n+m$,
% \begin{align*}
%   \lVert (Q_\delta - \tilde Q_\delta) A_i \rVert_2
% % \leq 2\sqrt2 \eps\cL' \lVert h_i \rVert + 2^{-\cL'} 9\sqrt2 M^2 \lVert \partial_x h_i \rVert + 2^{-\cL} 4\sqrt2 M^2 \lVert \partial_x h_i \rVert
% \leq 2\sqrt2 \eps\cL' \lVert h_i \rVert + 2^{-\cL'} 13\sqrt2 M^2 \lVert \partial_x h_i \rVert
% \end{align*}
% and hence, using 4sqrt(2) < 6, 13 sqrt(2)sqrt(2) = 26,
\begin{align*}
  \lVert \left( O_i - \tilde O_i \right) \Pi_{2k} \rVert \leq 8(2k+2) \left( 6 \eps\cL' + 2^{-\cL'} 26 M^2 \lVert \nabla A_i \rVert \right).
\end{align*}
Thus we obtain
% \begin{align*}
%   \delta_i
% \leq 8^m m! \left( 6\eps\cL' + 2^{-\cL'} 22 M^2 \lVert f_i' \rVert \right)
% \leq 8^m m! \left( 6\eps\cL' + 2^{-\cL'} 26 M^2 d(f,A) \right), \\
%   \delta_{n+i}
% \leq 8^m m! \left( 6 \eps\cL' + 2^{-\cL'} 26 M^2 \lVert \nabla A_i \rVert \right)
% \leq 8^m m! \left( 4 \eps\cL' + 2^{-\cL'} 26 M^2 d(f,A) \right).
% \end{align*}
\begin{align*}
 \left\lvert G^{\per}(\{O_i\}) - \tilde G^{\per}_{j,\cL}(\{O_i\}) \right\rvert
\leq 8^m m! (n+m) \left( 6 \eps\cL' + 26 M^2 D 2^{-\cL'} \right)
\end{align*}
in place of \cref{eq:G error unoptimized}.
Finally, we choose~$\cL' = \min \{ \cL, \lceil \log_2\frac{26 M^2 D}{\eps}\rceil \}$, which is always nonnegative, and arrive at
\begin{align*}
 \left\lvert G^{\per}(\{O_i\}) - \tilde G^{\per}_{\cL}(\{O_i\}) \right\rvert
% \\ \leq 8^m m! (n+m) \left( 6 \eps \left( \log_2\frac{26M^2 D}{\eps} + 1 \right) + 26 M^2 D \max \{ 2^{-\cL}, 2^{-\log_2\frac{26 M^2 D}{\eps}} \} \right) \\
% \leq 8^m m! (n+m) \left( 6 \eps \left( \log_2\frac{26M^2 D}{\eps} + 1 \right) + \max \{ 26 M^2 D 2^{-\cL}, \eps \} \right) \\
% \leq 8^m m! (n+m) \left( 6 \eps \left( \log_2\frac{26M^2 D}{\eps} + \frac76 \right) + 26 M^2 D 2^{-\cL} \right) \\
% 26 2^{7/6} <= 59
\leq 8^m m! (n+m) \left( 6 \eps \log_2\frac{59 M^2 D}{\eps} + 26 M^2 D 2^{-\cL} \right).
\end{align*}
This is the desired bound.
\end{proof}

To illustrate \cref{thm:approximation} and to show that the class of operators considered is an interesting class, we now describe how to compute correlation functions involving smeared stress-energy tensors.
The stress-energy tensor is a fundamental object in conformal field theory.
Its mode decomposition form two copies of the Virasoro algebra, encoding the conformal symmetry of the theory \cite{francesco2012conformal}.
It is convenient to choose a different basis and write the Dirac action in the form
%basis change by $\begin{psmallmatrix} 1 & - i \\ 1 & i \end{psmallmatrix}$
\begin{align*}
  S(\Psi) = \frac{1}{2}\int \Psi^{\dagger}\begin{pmatrix} \overline{\partial} & 0 \\ 0 & \partial \end{pmatrix}\Psi \d x \d t
\end{align*}
where $\partial = \partial_x + \partial_t$ and $\overline{\partial} = \partial_x - \partial_t$.
Then, formally, the holomorphic component~$T=T_{zz}$ of the stress-energy tensor, is the normal ordering of $\Psi_1^{\dagger} \partial \Psi_1$.
Solutions of the Dirac equation in this basis are of the form $\chi(x, t) = \chi_+(x + t) \op \chi_-(x - t)$.
The \emph{unsmeared} stress energy tensor $T(x)$ (which is only a formal expression in the algebraic formalism) is given by $T(x) = \d \Gamma_Q(D_x)$ where
\begin{align*}
  D_x \begin{pmatrix} f_1 \\ f_2 \end{pmatrix}
= \begin{pmatrix} \delta_x f_1' \\ 0 \end{pmatrix}
\end{align*}
where $\delta_x$ is a $\delta$-function centered at $x$.
To smear this operator, consider two smearing functions~$h_x$ and~$h_t$.
The~$h_t$ should be thought of as a smearing in the time direction and we use the Dirac equation to interpret this on our Hilbert space corresponding to $t=0$.
Thus, we define
\begin{align*}
  D(h) \begin{pmatrix} f_1 \\ f_2 \end{pmatrix}
= \begin{pmatrix} h_x (h_t \star f_1)' \\ 0 \end{pmatrix}
\end{align*}
where $\star$ denotes convolution.
We then define the \emph{smeared} stress-energy tensor by the normal-ordered second quantization:
$T(h) = \d \Gamma_Q(D(h))$.
If $h_x$ and $h_t$ are compactly supported functions in $H^1(\RR)$, then the operator $T(h)$ satisfies the conditions of \cref{thm:approximation}.
In \cref{fig:numerics},~(b) we show the numerical result of computing two-point functions~$\braket{T(h_1) T(h_2)}$ using our quantum circuits, where the~$h_i$ are taken to be Gaussian smearing functions.
In agreement with our theorem, we find that the two-point functions are approximated accurately for approximate Hilbert pairs of suitably good quality.
(Strictly speaking, the Gaussians need to be approximated by compactly supported functions so that \cref{thm:approximation} applies.)

%=============================================================================
\section{Quantum circuits for correlation functions}\label{sec:circuit}
%=============================================================================
We now explain how the mathematical approximation theorem can be used to construct unitary quantum circuits (in fact, tensor networks of MERA type) that rigorously compute correlation functions for free Dirac and Majorana fermions.
Finally, we discuss how symmetries are approximately implemented by our circuits.

%-----------------------------------------------------------------------------
\subsection{Discrete wavelet transform and single-particle circuits}
%-----------------------------------------------------------------------------
First, we describe how discrete wavelet transforms can be written as \emph{single-particle} (`first quantized' or `classical') linear circuits.
In this context, `single-particle' means that the state space is a direct sum of local state spaces (such as~$\ell^2(\ZZ) = \bigoplus_{n\in\ZZ} \CC$).
Thus let $W\colon\ell^2(\ZZ)\to\ell^2(\CC^2)$ denote a single layer of a discrete wavelet transform, defined as in \cref{eq:W}.
% for wavelet and scaling filters~$g_w$ and~$g_s$, respectively.
By putting the scaling and wavelet outputs on the even and odd sublattice, respectively, we obtain a unitary
\begin{align*}
  W'\colon \ell^2(\ZZ)\to\ell^2(\ZZ), \quad W' := \iota W,
\end{align*}
where
\begin{align*}
  &\iota \colon \ell^2(\ZZ) \ot \CC^2 \to \ell^2(\ZZ) \\
  &\iota (f_w \op f_s) [2n] = f_w[n], \\
  &\iota (f_w \op f_s) [2n+1] = f_s[n].
\end{align*}
It has been shown in~\cite{evenbly2018representation} that if the scaling filters are real and have length~$M$ then~$W'$ can be decomposed into a product~$W' = W_{M/2} \cdots W_1$, where each~$W_k : \ell^2(\ZZ) \to \ell^2(\ZZ)$ is a block-diagonal unitary of the form
\begin{align*}
  W_k = \begin{cases}
    \bigoplus_{r \, \text{odd}} u_{r,r+1}(\theta_k) & \text{if $k$ odd,}\\
    \bigoplus_{r \, \text{even}} u_{r,r+1}(\theta_k) & \text{if $k$ even.}
  \end{cases}
\end{align*}
Here, the~$\theta_k$ are suitable angles and~$u_{r,r+1}(\theta_k)$ denotes the unitary which acts on $\ell^2(\{r,r+1\}) \subseteq \ell^2(\ZZ)$ by the rotation matrix
\begin{align*}
  u(\theta_k) = \begin{pmatrix}
    \cos(\theta_k) & \sin(\theta_k) \\
    -\sin(\theta_k) & \cos(\theta_k)
  \end{pmatrix}.
\end{align*}
See~\cite{evenbly2018representation} for a proof and for an algorithm that finds the $\theta_k$ from the filter coefficients.
Thus, we obtain a decomposition of~$W^g$ into a single-particle linear circuit composed of 2-local unitaries (see \cref{fig:classical_circuit},~(a)).
In the same way we can implement $\cL$~layers of the discrete wavelet transform.
Given a 2-local circuit for a wavelet transform, it is not hard to see that the periodized version of this circuit will give the periodized version of the wavelet transform.
That is, the circuit has the structure shown in \cref{fig:circuit_intro},~(c), with exactly the same angles as for the original circuit on~$\ZZ$ for all scales larger than zero.

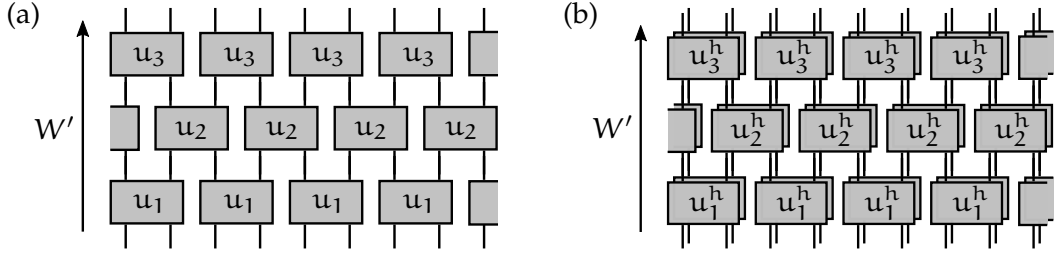
\begin{figure}
\begin{center}
(a)\quad\raisebox{-3cm}{
\begin{overpic}[height=3.2cm]{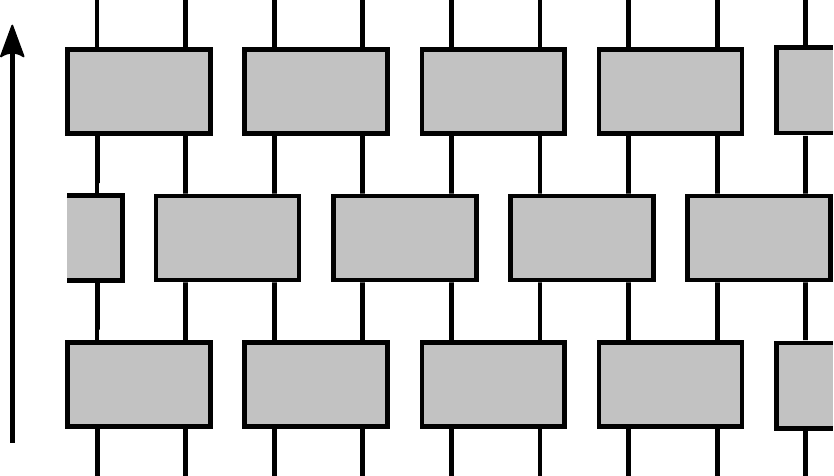}
\put(13.5,9.5){$u_1$} \put(35,9.5){$u_1$} \put(56,9.5){$u_1$} \put(77.5,9.5){$u_1$}
\put(23.5,27){$u_2$} \put(45,27){$u_2$} \put(66,27){$u_2$} \put(87.5,27){$u_2$}
\put(13.5,44.5){$u_3$} \put(35,44.5){$u_3$} \put(56,44.5){$u_3$} \put(77.5,44.5){$u_3$}
\put(-10,26){$W'$}
\end{overpic}}
\qquad
(b)\quad\raisebox{-3cm}{
\begin{overpic}[height=3.2cm]{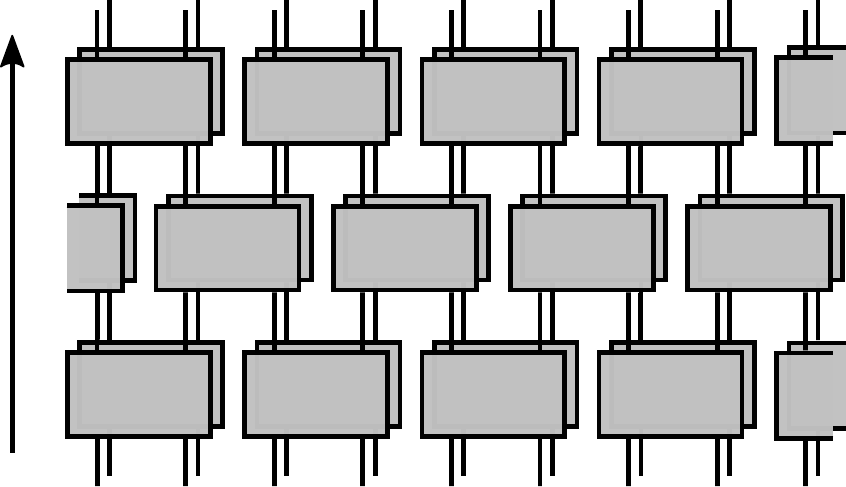}
\put(13,9){$u^h_1$} \put(34,9){$u^h_1$} \put(55,9){$u^h_1$} \put(76,9){$u^h_1$}
\put(23,26.5){$u^h_2$} \put(44,26.5){$u^h_2$} \put(65,26.5){$u^h_2$} \put(86,26.5){$u^h_2$}
\put(13,44){$u^h_3$} \put(34,44){$u^h_3$} \put(55,44){$u^h_3$} \put(76,44){$u^h_3$}
\put(-10,26.5){$W'$}
\end{overpic}}
\caption{(a) Decomposition of the single-layer discrete wavelet transform~$W$ as a 2-local single-particle linear circuit, where we abbreviate~$u_k := u(\theta_k)$.
(b) Circuit for a pair of wavelet transforms, $W = W_g \op W_h$, where we show~$u^h_k := u(\theta^h_k)$ on top of~$u^g_k := u(\theta^g_k)$.}
\label{fig:classical_circuit}
\end{center}
\end{figure}

Given an approximate Hilbert pair (or any pair of wavelets) we can consider~$W := W_h \op W_g$, corresponding to performing both discrete wavelet transforms in parallel.
If we apply the preceding construction to both wavelet transforms~$W_h$ and~$W_g$ we obtain two circuits, one for~$W'_h$ and one for~$W'_g$, parametrized by angles~$\theta_k^h$ and~$\theta_k^g$ for $k=1,\dots,M/2$.
These can be assembled into a single single-particle circuit for
\begin{align*}
  W'\colon\ell^2(\ZZ)\ot\CC^2\to\ell^2(\ZZ)\ot\CC^2, \quad
  W' := W'_g \op W'_h
\end{align*}
As shown in \cref{fig:classical_circuit},~(b), we take each site to carry two degrees of freedom (corresponding to the two components of the Dirac spinor).
Instead we could also arrange the two wavelet transforms on the even and odd sublattices (by conjugating with~$\iota$).
It is straightforward to see that the corresponding circuit can be implemented by 2-local unitaries and swap gates.

%-----------------------------------------------------------------------------
\subsection{Second quantized circuits for correlation functions}
%-----------------------------------------------------------------------------
Since we seek to describe a quantum many-body state of fermions, the circuit that we will construct is naturally a \emph{fermionic} quantum circuit that acts on a fermionic Fock space~$\fock(\ell^2(\ZZ))$, corresponding to a chain of fermions, by local unitaries.
In our case, it will be obtained by second-quantizing the single-particle circuit for the wavelet transforms described above.
Such circuits (sometimes called Gaussian fermionic circuits) can be efficiently simulated classically.
If one would like to implement these circuits on a quantum computer, one would have to convert the circuit to a qubit circuit.
In this case, a very natural way to do so is by applying a Jordan-Wigner transform.
The resulting type of (qubit) circuit is a so-called matchgate circuit~\cite{bravyi2002fermionic,jozsa2008matchgates}.
We refer to~\cite{evenbly2016entanglement} for further discussion of fermionic circuits in the context of wavelet transforms.
For a discussion of fermionic MERA in general, see~\cite{corboz2009fermionic}.

In \cref{subsec:approx}, we proved that its correlation functions~\eqref{eq:exact correlation} are well-approximated by \cref{eq:approx correlation}.
We now explain how the latter can be computed by a fermionic quantum circuit of MERA type.
\change{Loosely speaking, we will do the following, in order to compute some correlation function $G(\{O_i\})$:
\begin{enumerate}
\item We discretize the operators using the scaling functions. This gives a set of operators $\{O_i^{\MERA}\}$ on the lattice.
One can also consider this procedure the other way around: this embeds a discrete theory in the continuous theory by smearing all operators with appropriate scaling functions.
\item We compute the correlation function of the operators $\{O_i^{\MERA}\}$ using the state obtained by applying the second quantization of multiple layers of the wavelet transform circuit $W'$ to a product state.
A single layer of $W'$ is illustrated in \cref{fig:classical_circuit}, we compose the layers as usual in the wavelet transform by feeding the output of one layer into the scaling input of the next layer as in \cref{fig:classical_circuit}
The product state is given by a copy of $\ket+$ on each wavelet mode input of the circuit, or equivalently, by composing with an appropriate single mode unitary $h$, we may take as input copies of $\ket0$, see \cref{fig:classical_circuit} and \cref{fig:mera state}.
\end{enumerate}}
Let us first discuss the case of the free Dirac fermion on the real line in more detail.
We start with the approximate symbol~\eqref{eq:approx symbol}, omitting the isometries~$\alpha_j^\dagger$, and rewrite
\begin{align}\label{eq:symbol via U_MERA}
  W^{(\cL),\dagger} \bigl( P_w \ot \ket +\bra + \bigr) W^{(\cL)}
= U_{\MERA}^{(\cL),\dagger} P U_{\MERA}^{(\cL)},
\end{align}
where~$P := P_w \ot \ket0\bra0$ is a symbol on $\ell^2(\ZZ) \ot \CC^{\cL+1} \ot \CC^2$ and~$U_{\MERA}^{(\cL)} \colon \ell^2(\ZZ) \ot \CC^2 \to \ell^2(\ZZ) \ot \CC^{\cL+1} \ot \CC^2$ is the unitary defined by
\begin{align*}
  U_{\MERA}^{(\cL)} &:= (\id_{\ell^2(\ZZ) \ot \CC^{\cL-1} \ot \CC^2} \op U_{\MERA}) \cdots (\id_{\ell^2(\ZZ) \ot \CC^2} \op U_{\MERA}) U_{\MERA}, \\
  U_{\MERA} &:= (\id_{\ell^2(\ZZ)} \ot h) \op (\id_{\ell^2(\ZZ)} \ot \id_{\CC^2})) W,
\end{align*}
% where
% \begin{align*}
%   U_{\MERA}^{(\cL)} := (\id_{\ell^2(\ZZ) \ot \CC^{\cL}} \ot h \op \id_{\ell^2(\ZZ)}) (W_h^{(\cL)} \op W_g^{(\cL)})
% \end{align*}
% in terms of the Hadamard matrix~$h = \frac1{\sqrt2}\begin{psmallmatrix}1 & 1 \\ -1 & 1\end{psmallmatrix}$ which maps~$h\ket0=\ket+$.
% We can decompose
% \begin{align*}
%   U_{\MERA}^{(\cL)} = (\id_{\ell^2(\ZZ) \ot \CC^{\cL-1} \ot \CC^2} \op U_{\MERA}) \cdots (\id_{\ell^2(\ZZ) \ot \CC^2} \op U_{\MERA}) U_{\MERA},
% \end{align*}
where~$h$ is the Hadamard matrix~$h = \frac1{\sqrt2}\begin{psmallmatrix}1 & 1 \\ 1 & -1\end{psmallmatrix}$ which maps~$h\ket0=\ket+$.
Just like~$W$, $U_{\MERA}$ can be implemented by a circuit of depth~$M/2+1$, where~$M$ is the length of the filters, obtained by composing the circuit for~$W$ with an additional layer of Hadamard unitaries acting on the wavelet outputs (see \cref{fig:classical_layer}).
The unitary~$U_{\MERA}^{(\cL)}$ consists of~$\cL$ such circuit layers.

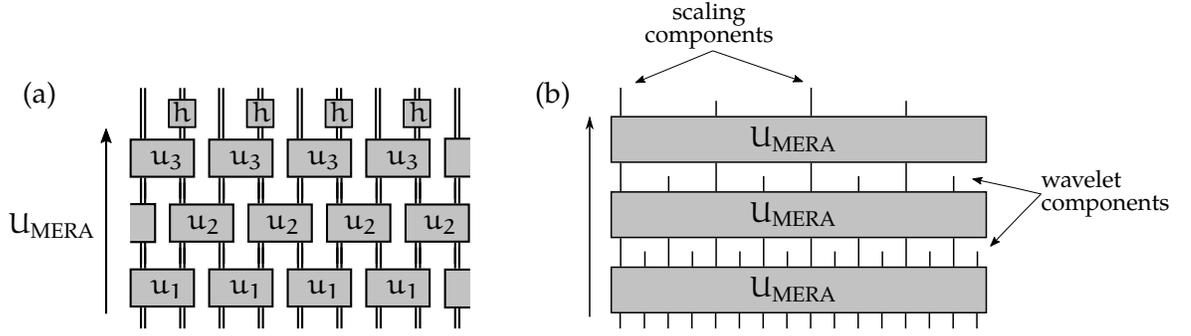
\begin{figure}
\begin{center}
(a)\quad\;\raisebox{-3cm}{
\begin{overpic}[height=3.2cm]{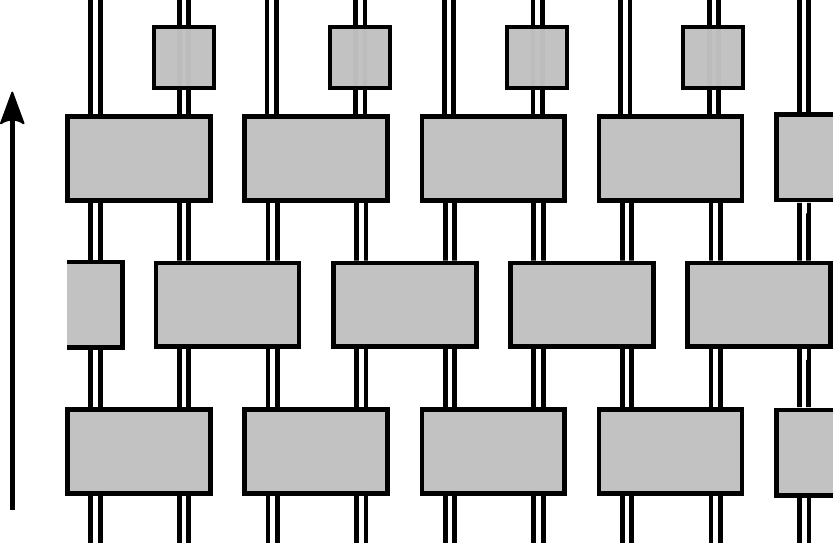}
\put(13.5,9.5){$u_1$} \put(35,9.5){$u_1$} \put(56,9.5){$u_1$} \put(77.5,9.5){$u_1$}
\put(23.5,27){$u_2$} \put(45,27){$u_2$} \put(66,27){$u_2$} \put(87.5,27){$u_2$}
\put(13.5,44.5){$u_3$} \put(35,44.5){$u_3$} \put(56,44.5){$u_3$} \put(77.5,44.5){$u_3$}
\put(19.5,56){$h$} \put(41,56){$h$} \put(62,56){$h$} \put(83.5,56){$h$}
\put(-25,26){$U_{\MERA}$}
\end{overpic}}
\qquad
(b)\;\raisebox{-3cm}{
\begin{overpic}[height=3.7cm]{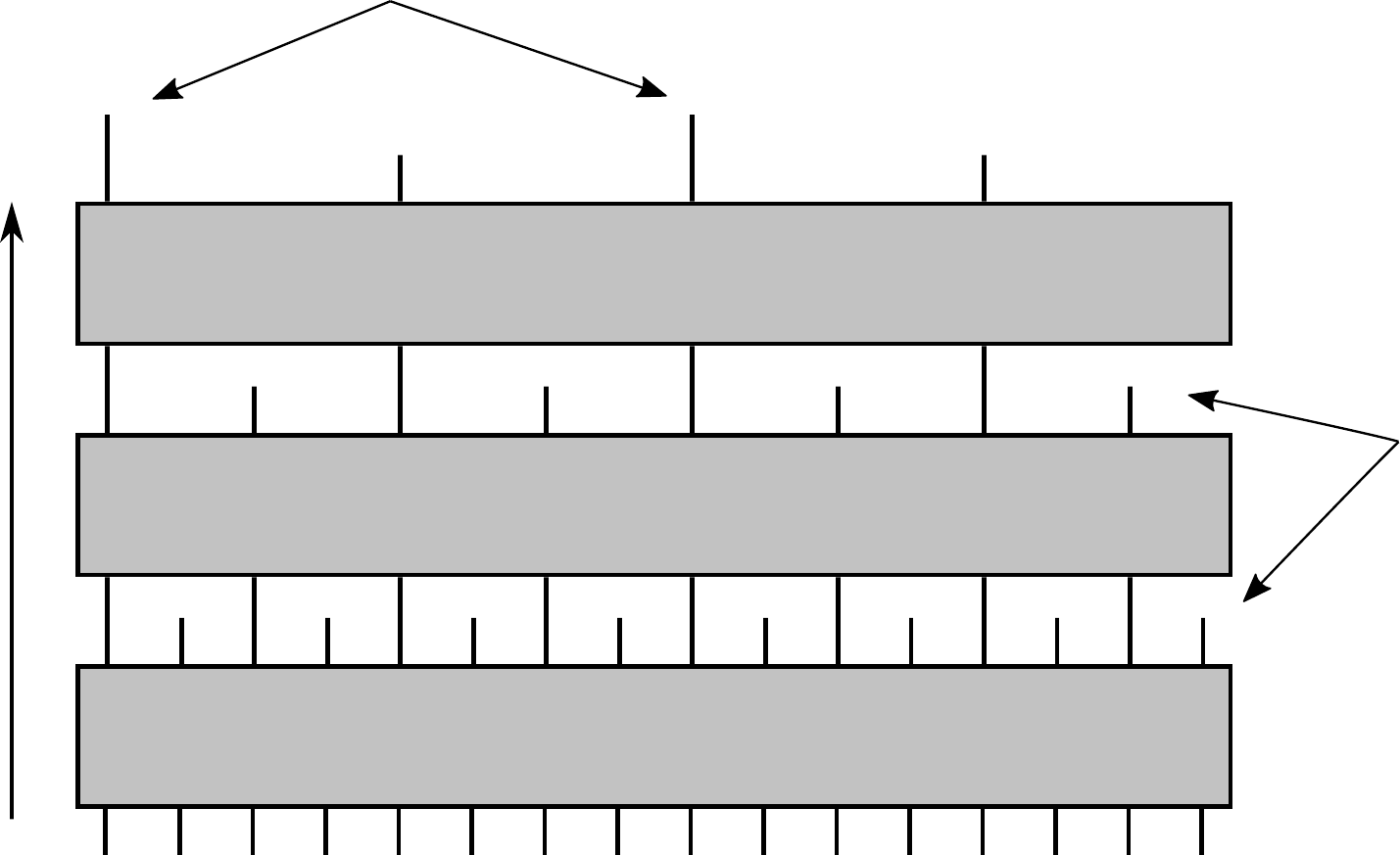}
% \put(-20,22){$U_{\MERA}^{(\cL)}$}
\put(36,7){$U_{\MERA}$}
\put(36,23.5){$U_{\MERA}$}
\put(36,40){$U_{\MERA}$}
\put(19,68){\footnotesize scaling} \put(13,63){\footnotesize components}
\put(100,31){\footnotesize wavelet} \put(100,26){\footnotesize components}
\end{overpic}}
\caption{(a) A single MERA layer~$U_{\MERA}$ before second quantization. The Hadamard unitary~$h$ (dis)entangles the modes of the two wavelet transforms that make up the Hilbert pair. We abbreviate~$u_k := u^h_k \op u^g_k$ (cf.\ \cref{fig:classical_circuit},~(b)).
(b) Illustration of the unitary~$U_{\MERA}^{(\cL)}$ corresponding to $\cL$~MERA layers before second quantization. Each layer is a local circuit of depth~$M/2+1$, as in~(a).}
\label{fig:classical_layer}
\end{center}
\end{figure}

The key point is that in view of \cref{eq:symbol via U_MERA} we can now compute the correlation function~$\tilde G_{j,\cL}(\{O_i\})$ in \cref{eq:approx correlation} as follows.

\begin{dfn}[MERA correlation functions]
Consider an approximate Hilbert pair with filters~$g$,~$h$.
Given a correlation function~\eqref{eq:exact correlation}, $j\in\ZZ$, and $\cL\geq0$, we define the corresponding \emph{MERA correlation function} by
\begin{align}\label{eq:G MERA}
  G^{\MERA}_{j,\cL}(\{O_i\}) := \braket{\Omega | O^{\MERA}_1 \cdots O^{\MERA}_n | \Omega},
\end{align}
where $O^{\MERA}_i$ is obtained from~$O_i$ by replacing~$\Psi(f)$ by~$\Psi_{\MERA}(f) := a_P(U_{\MERA}^{(\cL)} \alpha_j f)$ and~$\d\Gamma_Q(A)$ by~$\d\Gamma_{\MERA}(A) := \d\Gamma_P(U_{\MERA}^{(\cL)} \alpha_j A \alpha_j^\dagger U_{\MERA}^{(\cL),\dagger})$.
Here, $P := P_w \ot \ket0\bra0$.
\end{dfn}

\noindent Importantly,~$P$ is the symbol of a state for which correlation functions can be straightforwardly evaluated.
Indeed, we can intuitively think of~$P = P_w \ot \ket0\bra0$ as the symbol of a `Fermi sea' where half of the wavelet modes are occupied (equivalently, after a Jordan-Wigner transformation this state corresponds to an `infinite product state' where the wavelet qubits are in state~$\ket{101010\dots}$ and the scaling qubits in~$\ket{0000\dots}$).
More precisely, using \cref{eq:CAR_rep}, we find that
\begin{align}\label{eq:Psi MERA}
  \Psi_{\MERA}(f)
= a_0\bigl((\id-P)U_{\MERA}^{(\cL)} \alpha_j f\bigr) + a_0^\dagger\bigl(\overline{P U_{\MERA}^{(\cL)} \alpha_j f}\bigr),
\end{align}
where~$a_0^{(\dagger)}$ are the ordinary creation and annihilation operators on Fock space.
If~$f$ is a smearing function then in order to find~$\Psi_{\MERA}(f)$ we first compute $\alpha_j f$ either by expanding the scaling basis or simply by sampling (\cref{lem:sampling error}), then we apply $\cL$ layers of the local circuit~$U_{\MERA}$ (\cref{fig:classical_layer}), and finally we apply the projections~$P$ and~$\id-P$.
One can proceed similarly for~$\d\Gamma_{\MERA}(A)$.
This shows that the correlation functions~\eqref{eq:approx correlation} and~\eqref{eq:G MERA} can be efficiently calculated in the single-particle picture.

We now explain how to obtain a fermionic \emph{quantum} circuit with rigorous approximation guarantees.
Suppose that, as in \cref{thm:approximation}, we wish to approximate a correlation function involving $\Psi^{(\dagger)}(f_i)$ and $\d\Gamma_Q(A_i)$, where the smearing functions~$f_i$ and the kernel of~$A_i$ are compactly supported.
In this case, it is easy to see that \cref{eq:G MERA} will involve creation and annihilation operators that act only on finitely many sites~$S \subseteq \ZZ$ (which can be computed from the supports as well as the parameters~$j$,~$\cL$, and~$M$).
In this case, we can replace~$\ell^2(\ZZ)$ by~$\ell^2(S)$, $P$ by its restriction~$P_S$ onto~$\HH_S := \ell^2(\ZZ) \ot \CC^{\cL+1} \ot \CC^2$, and the infinitely wide layers~$U_{\MERA}$ by finitely many local unitaries.
Let us denote by~$\ket{P_S}$ the corresponding Slater determinant in the fermionic Fock space~$\fock(\HH_S)$ and by~$\Gamma_0(U_{\MERA}^{(\cL)}) := \bigoplus_{k=0}^\infty (U_{\MERA}^{(\cL)})^{\wedge k}$ the second quantizations of the single-particle unitaries~$U_{\MERA}^{(\cL)}$.
Since second quantization commutes with multiplication, this can be written as a fermionic quantum circuit composed of~$\cL$ many identical layers, each of depth~$M/2+1$ (which structurally looks like \cref{fig:classical_layer},~(b)).
Thus, we recognize that~$\ket{\MERA_{\cL}} := \Gamma_0(U_{\MERA}^{(\cL)})^\dagger \ket{P_S}$ is precisely the quantum state prepared by a fermionic MERA, as illustrated in \cref{fig:mera state}.
Moreover, we can compute the MERA correlation functions by
\begin{align}\label{eq:G via MERA}
  G^{\MERA}_{j,\cL}(\{O_i\}) = \braket{\MERA_{\cL} | O'_1 \cdots O'_n | \MERA_{\cL}},
\end{align}
where~$O'_i$ is obtained from~$O_i$ by replacing~$\Psi(f)$ by~$\Psi'(f) := a_0(\alpha_j f)$ and~$\d\Gamma_Q(A)$ by~$\d\Gamma_0(\alpha_j A \alpha_j^\dagger) - \braket{\MERA_{\cL}|\d\Gamma_0(\alpha_j A \alpha_j^\dagger)|\MERA_{\cL}}$.
Note that $\braket{\MERA_{\cL}|\d\Gamma_0(\alpha_j A \alpha_j^\dagger)|\MERA_{\cL}}$ is actually \emph{finite} because we truncated the range of wavelet scales, so this normal ordering is well-defined (even if the original operator $A$ was not trace class).
Thus, \cref{eq:G via MERA} can be interpreted as an ordinary correlation function in a fermionic MERA.
This at last justifies our notation.

\begin{figure}
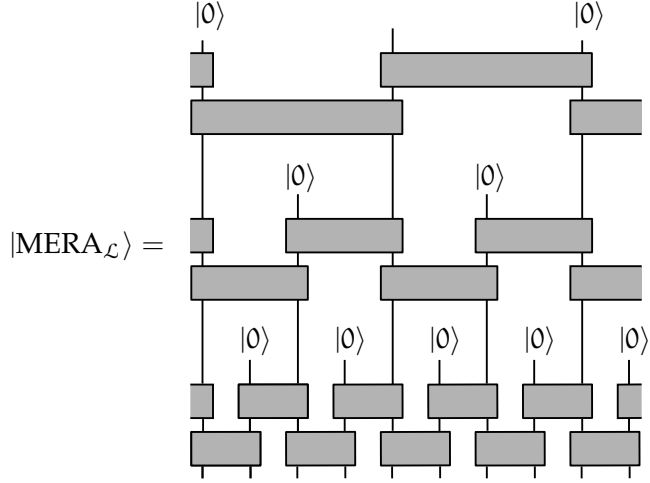

\begin{center}
\begin{overpic}[width=6.0cm]{quantum_circuit_intro}
\put(1,101){$\ket0$} \put(85,101){$\ket0$}
\put(21,65){$\ket0$} \put(63,65){$\ket0$}
\put(-40,50){$\ket{\MERA_{\cL}} = $}
\put(11,29){$\ket0$} \put(32,29){$\ket0$} \put(53,29){$\ket0$}
\put(74,29){$\ket0$} \put(95,29){$\ket0$}
\end{overpic}
\caption{The MERA state is created by applying the unitary circuit to an identical product state at each level. \change{Here we denote by $\ket0$ the Fock vacuum state}.}
\label{fig:mera state}
\end{center}
\end{figure}

% \begin{figure}
% \begin{center}
% \includegraphics[width=0.5\linewidth]{quantum_circuit_layer.png}
% \caption{A single layer of the quantized wavelet circuit, entangling the wavelet modes using $H$. The local single-particle Hilbert spaces are $\CC^2 \op \CC^2$.}
% \label{fig:quantum_circuit_layer}
% \end{center}
% \end{figure}

% \begin{figure}
% \begin{center}
% \includegraphics[width=0.9\linewidth]{quantum_circuit.png}
% \caption{Each layer $\mathcal{W}$ is a 2-local circuit of depth $M/2 +1$, as in \cref{fig:quantum_circuit_layer}. }
% \label{fig:quantum_circuit}
% \end{center}
% \end{figure}

%-----------------------------------------------------------------------------
\subsection{Circle, boundary conditions, Majorana fermions}\label{subsec:periodic circuit}
%-----------------------------------------------------------------------------
For the circle~$\SS$ much the same construction applies.
Given \cref{eq:approx symbol per}, we start with
\begin{align*}
  W^{(\cL),\per,\dagger} \bigl( P_w \ot \ket+\bra+ + P_s \ot \ket{L}\bra{L} \bigr) W^{(\cL),\per}
= U_{\MERA}^{(\cL),\per,\dagger} P_{\per} U_{\MERA}^{(\cL),\per},
\end{align*}
for a suitably defined unitary~$U_{\MERA}^{(\cL),\per}$ and~$P_{\per} = P_w \ot \ket0\bra0 + P_s \ot \ket L\bra L$.
This is already a symbol on a finite-dimensional Hilbert space~$\CC^{2^j} \ot \CC^2$.
As before, $U_{\MERA}^{(\cL),\per}$ is a product of unitaries, one for each layer, but now these unitaries will depend on the scale~$j=0,\dots,\cL-1$ (cf.~\cref{subsec:periodic_wavelets}).
Since taking the periodization of composition of convolutions is the same as periodizing their composition, we can obtain the unitary~$U_{\MERA}^{\per,j}$ for the~$j$-th layer simply by `periodizing' the two-local unitaries~$U_{\MERA}$ and analogously construct the circuit (see \cref{fig:circuit_intro},~(c)).
Just like the filters, the MERA layers become identical for suffiently large~$j$.

This leads to an approximation of the exact correlation functions~$G^{\per}(\{O_i\})$ for periodic boundary conditions
\begin{align*}
  G^{\MERA,\per}_{j,\cL}(\{O_i\}) := \braket{\Omega | O^{\MERA}_{\per,1} \cdots O^{\MERA}_{\per,n} | \Omega},
\end{align*}
where $O^{\MERA}_{\per,i}$ is obtained from~$O_i$ by replacing~$\Psi(f)$ by~$\Psi_{\MERA}^{\per}(f) := a_{P_{\per}}(U_{\MERA}^{(\cL),\per} \alpha^{\per}_j f)$ and~$\d\Gamma_Q(A)$ by~$\d\Gamma_{\MERA}^{\per}(A) := \d\Gamma_{P_{\per}}(U_{\MERA}^{(\cL),\per} \alpha^{\per}_j A \alpha^{\per,\dagger}_j U_{\MERA}^{(\cL),\per,\dagger})$.
As before, this can be interpreted as a correlation function of local operators in a fermionic MERA on a circle.

For anti-periodic boundary conditions on the circle, the symbol was given by~$T^{\dagger} Q^{\per} T$ (see \cref{sec:freefermions}).
This means that we can compute correlation functions for anti-periodic boundary conditions with the same circuit as for the periodic fermion, but replacing~$\alpha_j$ by~$\alpha_j T$.
We note that the smearing functions~$f$ in this case are naturally anti-periodic (they are sections of a nontrivial bundle), so~$Tf$ is periodic and our results apply.

Finally we discuss the case of Majorana fermions.
For simplicity, we only consider the case of the line (cf.~\cref{subsec:selfdual}).
Suppose that we want to approximate a correlation function of the form
\begin{align}\label{eq:majorana correlator}
  G^{\maj}(\{f_i\}) = \braket{\Omega | \Phi(f_1) \dots{} \Phi(f_n) | \Omega},
\end{align}
where the smeared Majorana field is given by $\Phi(f) = a_0((\id-Q) f) + a_0^\dagger(CQ f)$ in terms of the symbol~$Q$ of the free Dirac fermion, and the charge conjugation operator~$C$ defined in \cref{eq:majorana C}.
 % = a_0((\id-Q) f) + a_0^\dagger((\id_{L^2(\RR)} \ot Z) \overline{Qf})$, where $Z=\begin{psmallmatrix}1&0\\0&-1\end{psmallmatrix}$.
Consider the self-dual CAR algebra on the range of $P' =  P_w \ot \id_{\CC^2}$ which is a subspace~$\HH'$ of $\ell^2(\ZZ) \ot \CC^2 \ot \CC^{\cL + 1}$ (that is, the subspace corresponding to the wavelet coefficients) with charge conjugation~$C'$ given by the anti-unitary operator on~$\HH'$ which acts by~$x=\begin{psmallmatrix}0&1\\1&0\end{psmallmatrix}$ in the second tensor factor and componentwise complex conjugation in the standard basis.
Similarly to \cref{eq:Psi MERA}, define
\begin{align*}
  \Phi^{\MERA}_{\maj}(f)
= a_0\bigl((P' - P) U_{\MERA}^{(\cL)} \alpha_j f\bigr) + a_0^\dagger\bigl(C' P U_{\MERA}^{(\cL)} \alpha_j f\bigr).
\end{align*}
We note that the above formula defines a representation of the self-dual CAR algebra~$\CARsd(\HH')$ since, clearly,~$C' P = (P'-P)C'$.
As before, we can approximate the correlation function~\eqref{eq:majorana correlator} by
\begin{align*}
  G^{\MERA,\maj}_{j,\cL}(\{f_i\}) = \braket{\Omega | \Phi^{\MERA}_{\maj}(f_1) \dots{} \Phi^{\MERA}_{\maj}(f_n) | \Omega},
\end{align*}
which for compactly supported~$f_i$ can be computed by an ordinary fermionic MERA.
Note that
\begin{align*}
  C' U_{MERA}^{(\cL)} \alpha_j(f)
= U_{MERA}^{(\cL)} C \alpha_j (f)
= U_{MERA}^{(\cL)} \alpha_j (Cf)
\end{align*}
where, with a slight abuse of notation, also write $C$ for the similarly defined operator on $\ell^2(\ZZ) \ot \CC^2$.
Thus, we can also implement $\Gamma^c(U_{MERA}^{(\cL)})$ as a circuit of Majorana fermions, mapping the state on $\CARsd(\HH')$ corresponding to $P$ to the state on $\CARsd(U_{MERA}^{(\cL)}(\HH'))$ with symbol $U_{MERA}^{(\cL),\dagger} P U_{MERA}^{(\cL)}$.

%-----------------------------------------------------------------------------
\subsection{Symmetries}\label{subsec:symmetries}
%-----------------------------------------------------------------------------
For MERA tensor networks, it has been observed that the (local and global) symmetries of the underlying theory can be approximately implemented in terms of the tensor network itself~\cite{milsted2018tensor}.
In particular, a single layer of the MERA should always correspond to a rescaling by a factor two.
In the wavelet construction, the relation between a single MERA layer and rescaling is very explicit.

In fact, we can show that the operator corresponding to a fermionic field has exact scaling dimension~$\frac{1}{2}$, as was already observed in~\cite{evenbly2016entanglement}.
For this, consider (formally) the Dirac fermion field~$\Psi_i(x)$, where~$\delta_x$ is a delta function centered at~$x$ and $i\in\{1,2\}$.
Its MERA realization at scale~$j\in\ZZ$ is given by
\begin{align}\label{eq:MERA field}
  \Psi_i^{\MERA}(x)
= a^\dagger(\alpha_j (\delta_x \ot \ket i))
= \sum_{k\in\ZZ} \bar{\phi}_{j,k}(x) a^\dagger(\ket k \ot \ket i)
\end{align}
(we identify the CAR algebra with its representation).
Since the scaling functions are compactly supported, the right-hand side expression is well-defined and we take it as the definition of~$\Psi_i^{\MERA}(x)$.
Now note that the scaling superoperator for a single MERA layer consists of a conjugation by the second quantization of~$U_{\MERA}$ and a contraction with the quasi-free state with symbol~$\id_{\ell^2(\ZZ)} \ot \ket+\bra+$ on the wavelet output (cf.~\cref{fig:global_sym}, (a)).
Thus, any creation operator~$a^\dagger(f)$ gets mapped to~$a^\dagger(P_s W f)$, where~$P_s$ denotes the projection onto the scaling modes.
Using \cref{eq:MERA field}, it follows that the scaling superoperator maps
\begin{align*}
\Psi_1^{\MERA}(x) &\mapsto \sum_{k\in\ZZ} \bar{\phi}_{j,k}(x) a^\dagger(P_s W^h \ket k \ot \ket 1)
= \sum_{k\in\ZZ} \bar{\phi}_{j,k}(x) a^\dagger((\downarrow m(\overline{\hat h_s}) \ket k) \ot \ket 1) \\
% &= \sum_{k\in\ZZ} \bar{\phi}_{j,k}(x) \sum_{n\in\ZZ} \braket{n|\downarrow m(\overline{\hat h_s})|k} a^\dagger(\ket n \ot \ket 1) \\
% &= \sum_{k\in\ZZ} \bar{\phi}_{j,k}(x) \sum_{n\in\ZZ} \braket{2n|m(\overline{\hat h_s})|k} a^\dagger(\ket n \ot \ket 1) \\
&= \sum_{k\in\ZZ} \bar{\phi}_{j,k}(x) \sum_{n\in\ZZ} \bar h_s[k-2n] a^\dagger(\ket n \ot \ket 1)
% = \sum_{n\in\ZZ} \sum_{k\in\ZZ} \bar h_s[k] \bar{\phi}_{j,k+2n}(x) a^\dagger(\ket n \ot \ket 1) \\
= \sum_{n\in\ZZ} \bar{\phi}_{j-1,n}(x) a^\dagger(\ket n \ot \ket 1) \\
&= \sum_{n\in\ZZ} 2^{-\frac12} \bar{\phi}_{j,n}(\frac x2) a^\dagger(\ket n \ot \ket 1)
= 2^{-\frac12} \, \Psi_1^{\MERA}(\tfrac x2),
\end{align*}
where we used \cref{eq:W,eq:relation_filter_scaling}.
We can argue similarly for the other component, as well as for the adjoints.
Thus, we conclude that a single MERA layer coarse-grains~$\Psi^{\MERA}(x) \mapsto 2^{-\frac12} \Psi^{\MERA}(\frac x2)$.
The interpretation is that a single layer of the MERA corresponds to a rescaling of the fields by a factor two (as it should) and that it \emph{exactly} reproduces the correct scaling dimension of $\frac{1}{2}$ for the fermionic fields.
In general the other scaling dimension of the theory are only approximately reproduced and it would be interesting to prove quantitative bounds (for example, using our \cref{thm:approximation}).

We can also implement other global symmetries on the circuit level.
Translations by steps of size~$2^{-j}$ are trivially implemented by a circuit.
Since we know the explicit time-dependence of the solutions of the Dirac equation, we can implement time translations by transforming with a basis change given by $T = \frac{1}{\sqrt{2}}\begin{psmallmatrix} 1 & i \\ 1 & -i \end{psmallmatrix}$ such that time translation shifts the first component to the right and the second component to the left.
These global symmetries are shown in \cref{fig:global_sym}, and should be interpreted in the sense that if we want to compute correlation functions with these symmetry operators inside the correlator then we can insert the corresponding circuits.
The approximation theorem and the invariance of the free fermion under these transformations show that these symmetries are indeed accurately implemented.

\begin{figure}
\begin{center}
(a)\;
\raisebox{-1.5cm}{\begin{overpic}[height=1.5cm]{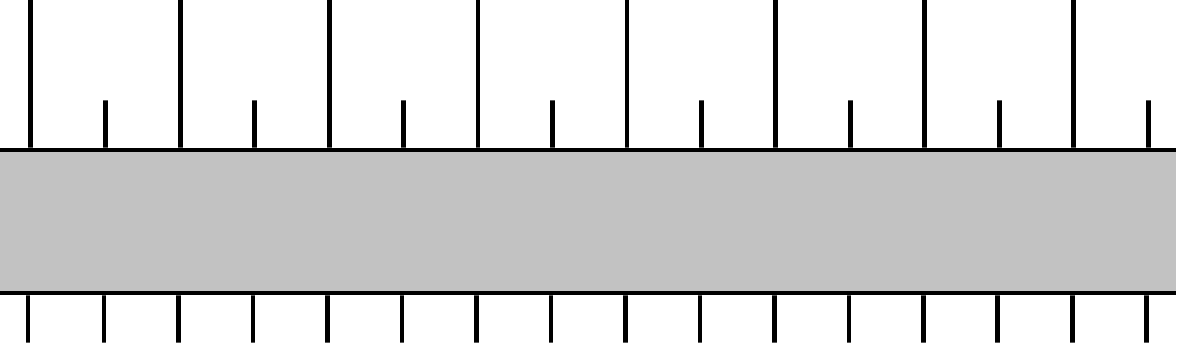}
\put(39,9){$U_{\MERA}$}
\put(6,23){$\ket0$} \put(19,23){$\ket0$} \put(31,23){$\ket0$}
\put(44,23){$\ket0$} \put(57,23){$\ket0$} \put(69.5,23){$\ket0$}
\put(82,23){$\ket0$} \put(94.5,23){$\ket0$}
\end{overpic}}
\qquad
(b)\;
\raisebox{-1.2cm}{\begin{overpic}[height=1cm]{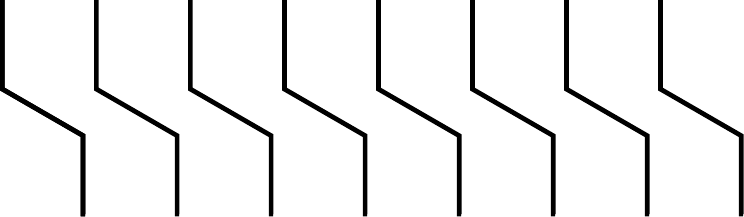}\end{overpic}}
\qquad
(c)\;
\raisebox{-1.5cm}{\begin{overpic}[height=1.8cm]{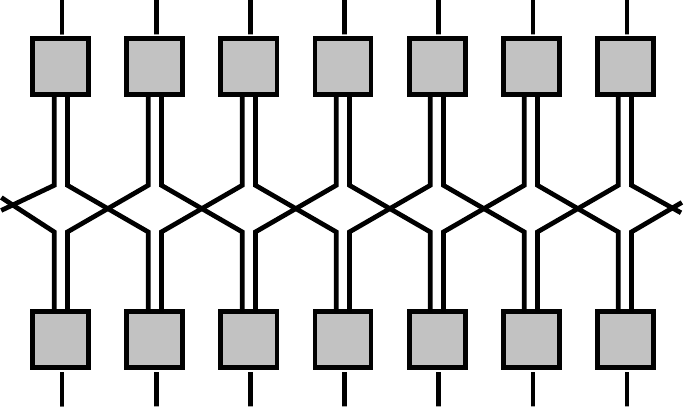}
\put(6.5,47){\scalebox{0.5}{$T$}}
\put(6.5,7){\scalebox{0.5}{$T^{\dagger}$}}
\end{overpic}}
\caption{(a) Rescaling by a factor two is implemented by conjugation with a single MERA layer.
(b) Space translation.
(c) Time translation.}
\label{fig:global_sym}
\end{center}
\end{figure}
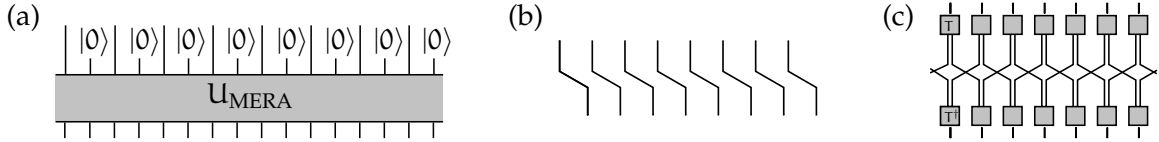

%=============================================================================
\section*{Acknowledgments}
\phantomsection
\addcontentsline{toc}{section}{Acknowledgments}
%=============================================================================
We acknowledge interesting discussions with Sukhbinder Singh.
\change{We would like to thank the anonymous referees for their thoughtful feedback.}
MW acknowledges support by the NWO through Veni grant no.~680-47-459.
VBS expresses his thanks to the University of Amsterdam and the CWI for their hospitality. He acknowledges funding by the ERC consolidator grant QUTE and thanks Frank Verstraete for discussions and his support. BGS is supported by the Simons Foundation as part of the It From Qubit Collaboration.
\begin{appendix}

\section{Proofs of wavelet lemmas}\label{sec:appendix}

In this section we will prove some technical lemmas involving wavelets, amongst which \cref{lem:UV}, \cref{lem:sampling error}, \cref{lem:IR}.
We first state a simple Lipschitz bound for the Fourier transforms of wavelet and scaling filters.

\begin{lem}\label{lem:filter fourier bounds}
Let~$g_s$ be scaling filter supported in~$\{0,\dots,M-1\}$.
Then the corresponding wavelet filter~$g_w$, defined in \cref{eq:wavelet from scaling filter}, is supported in~$\{2-M,\dots,1\}$ and we have that
\begin{align*}
  \lvert \hat g_s(\theta) - \sqrt2 \rvert &\leq \frac{M^2}{\sqrt2} \, \lvert\theta\rvert, \\
  \lvert \hat g_w(\theta) \rvert &\leq \frac{M(M+1)}{\sqrt2} \, \lvert\theta\rvert.
\end{align*}
for all~$\theta\in[-\pi,\pi]$.
\end{lem}
\begin{proof}
By \cref{eq:cond scaling filter},~$\lVert \hat g_s \rVert_\infty=\sqrt2$ and~$\hat g_s(0) = \sqrt2$.
Hence,
\begin{align*}
  \lVert \hat g_s' \rVert_\infty
\leq \left( \sum_{n=0}^{M-1} n \right) \lVert g_s\rVert_\infty
= \frac{M(M-1)}2 \lVert g_s\rVert_\infty
\leq \frac{M(M-1)}2 \lVert \hat g_s\rVert_\infty
\leq \frac{M^2}{\sqrt2},
\end{align*}
where we used that~$\lVert f\rVert_\infty \leq \frac1{2\pi} \lVert \hat f\rVert_1 \leq \lVert \hat f\rVert_\infty$ for any trigonometric polynomial.
Therefore,
\begin{align*}
  \lvert \hat g_s(\theta) - \sqrt2 \rvert
\leq \lvert \hat g_s(\theta) - \hat g_s(0) \rvert
\leq \lVert \hat g_s' \rVert_\infty \, \lvert\theta\rvert
\leq \frac{M^2}{\sqrt2} \, \lvert\theta\rvert.
\end{align*}
Now consider the corresponding wavelet filter~$g_w$ which by \cref{eq:wavelet from scaling filter,eq:cond scaling filter} satisfies~$\lVert \hat g_w \rVert_\infty=\sqrt2$ and~$\hat g_w(0) = 0$ and is supported in $\{2-M,\dots,1\}$.
Then, similarly as above,
\begin{align*}
  \lVert \hat g_w' \rVert_\infty
\leq \left( \sum_{n=2-M}^1 \lvert n\rvert \right) \lVert g_w\rVert_\infty
% = \left(1 + \frac{(M-2)(M-1)}2 \right) \lVert g_w\rVert_\infty
\leq \frac{M(M+1)}2 \lVert \hat g_w\rVert_\infty
\leq \frac{M(M+1)}{\sqrt2},
\end{align*}
so we obtain
\begin{align*}
  \lvert \hat g_w(\theta) \rvert
= \lvert \hat g_w(\theta) - \hat g_w(0) \rvert
\leq \lVert \hat g_w' \rVert_\infty \, \lvert\theta\rvert
\leq \frac{M(M+1)}{\sqrt2} \, \lvert\theta\rvert.
\end{align*}
\end{proof}

In practice, the bounds in \cref{lem:filter fourier bounds} can be pessimistic.
\change{In principle, if the number of vanishing moments of the wavelets increase, one expects better dependence of the bounds on the size of the support, although we are not aware of better bounds than those in \cref{lem:filter fourier bounds} for approximate Hilbert pair wavelets.}

We now proceed to prove the lemmas in \cref{subsec:wavelet approximations}.
Our main tool is the following technical lemma.

\begin{lem}\label{lem:technical}
Let~$\chi \in H^{-K}(\RR)$ such that~$\hat\chi\in L^\infty(\RR)$ and there exists a constant~$C>0$ such that~$\lvert \hat\chi(\omega) \rvert \leq C \lvert\omega\rvert^K$ for all~$\lvert\omega\rvert\leq\pi$.
Define~$C_\chi := (C^2 + \lVert\hat\chi\rVert^2_\infty / 3)^{1/2}$.
Then, for all~$f\in H^K(\RR)$ and~$j\in\ZZ$ we have that
\begin{align*}
  \sum_{k\in\ZZ} \lvert\braket{\chi_{j,k}, f}\rvert^2
\leq 2^{-2Kj} C_\chi^2 \lVert f^{(K)} \rVert^2,
\end{align*}
where~$\chi_{j,k}(x) := 2^{\frac j2} \chi(2^j x - k)$.
Similarly, for all~$f\in H^K(\SS)$ and~$j\geq0$ we have that
\begin{align*}
  \sum_{k=1}^{2^j} \lvert\braket{\chi_{j,k}^{\per}, f}\rvert^2
\leq 2^{-2Kj} C_\chi^2 \lVert f^{(K)} \rVert^2,
\end{align*}
where~$\chi^{\per}_{j,k}(x) = \sum_{m \in \ZZ} \chi_{j,k}(x + m)$.
\end{lem}
\begin{proof}
For $f\in H^K(\RR)$, we start with
\begin{align}
\nonumber
  \sum_{k \in \ZZ} \left\lvert\braket{\chi_{j,k}, f} \right\rvert^2
  &= \sum_{k \in \ZZ} \left\lvert \frac1{2\pi}\braket{ \widehat{\chi_{j,k}}, \widehat{f}} \right\rvert^2 \\
\nonumber
  &= \sum_{k \in \ZZ} \left\lvert \frac1{2\pi}\int_{-\infty}^{\infty} 2^{-j/2}e^{i\omega2^{-j}k} \overline{\hat\chi(2^{-j}\omega)} \hat f(\omega) d\omega \right\rvert^2 \\
\label{eq:series}
  &= \sum_{k\in\ZZ} \left\lvert \frac1{2\pi} \int_{-\infty}^\infty 2^{j/2}  \overline{\hat\chi(\omega)}  \hat f(2^j\omega) e^{i\omega k} d\omega \right\rvert^2.
\end{align}
We can interpret this as the squared norm of the Fourier coefficients of the~$2\pi$-periodic function defined by
\begin{align*}
  F(\theta) := \sum_{m\in\ZZ} 2^{j/2}  \overline{\hat\chi(\theta + 2\pi m)} \hat f(2^j(\theta + 2\pi m)),
\end{align*}
provided the latter is square integrable.
To see this and obtain a quantitative upper bound, we note that, for every $\theta\in[-\pi,\pi]$,
\begin{align}
\nonumber
\lvert F(\theta)\rvert^2
&\leq 2^j \sum_{m\in\ZZ} \left\lvert \frac{\hat\chi(\theta + 2\pi m)}{(\theta + 2\pi m)^K} \right\rvert^2 \sum_{m \in\ZZ} \left\lvert (\theta + 2\pi m)^K \hat f(2^j(\theta + 2\pi m)) \right\rvert^2 \\
&= 2^{-(2K-1)j} \sum_{m\in\ZZ} \left\lvert \frac{\hat\chi(\theta + 2\pi m) }{(\theta + 2\pi m)^K} \right\rvert^2 \sum_{m\in\ZZ} \left\lvert (2^j(\theta + 2\pi m))^K \hat f(2^j(\theta + 2\pi m)) \right\rvert^2
\label{eq:F theta abs squared}
\end{align}
by the Cauchy-Schwarz inequality.
To bound the left-hand side series, we split off the term for~$m=0$ and use the assumptions on~$\hat\chi$ to bound, for $|\theta|\leq\pi$,
\begin{align}
\nonumber
\sum_{m\in\ZZ} \left\lvert \frac{\hat\chi(\theta + 2\pi m)}{(\theta + 2\pi m)^K} \right\rvert^2
&= \left\lvert \frac{\hat\chi(\theta)}{\theta^K} \right\rvert^2 + \sum_{m\neq0} \left\lvert \frac{\hat\chi(\theta + 2\pi m)}{(\theta + 2\pi m)^K} \right\rvert^2
\leq C^2 + \sum_{m\neq0} \frac{\lvert \hat\chi(\theta + 2\pi m)\rvert^2}{\lvert\theta + 2\pi m\rvert^{2K}} \\
\label{eq:def C_chi}
&\leq C^2 + \lVert \hat\chi\rVert_\infty^2 \sum_{m=1}^\infty \frac{2}{(\pi m)^{2K}}
% \leq C^2 + \frac{2 \lVert\hat\chi\rVert^2_\infty}{\pi^{2K}} \zeta(2K)
% \leq C^2 + \frac{2 \lVert\hat\chi\rVert^2_\infty}{\pi^{2K}} \zeta(2)
% \leq C^2 + \frac{2 \lVert\hat\chi\rVert^2_\infty}{\pi^{2K}} \frac{\pi^2}6
% \leq C^2 + \frac{\lVert\hat\chi\rVert^2_\infty}{3 \pi^{2(K-1)}}
\leq C^2 + \frac{\lVert\hat\chi\rVert^2_\infty}3
= C_\chi^2
\end{align}
If we plug this into \cref{eq:F theta abs squared} then we obtain
\begin{align*}
  \lvert F(\theta)\rvert^2
\leq 2^{-(2K-1)j} C_\chi^2 \sum_{m\in\ZZ} \left\lvert (2^j(\theta + 2\pi m))^K \hat f(2^j(\theta + 2\pi m)) \right\rvert^2
\end{align*}
and hence
\begin{align*}
  \frac1{2\pi} \int_{-\pi}^\pi |F(\theta)|^2 d\theta
% = 2^{-(2K-1)j} C_\chi^2 \int_{-\pi}^\pi \sum_{m\in\ZZ} \left\lvert (2^j(\theta + 2\pi m))^K \hat f(2^j(\theta + 2\pi m)) \right\rvert^2 d\theta
&\leq 2^{-(2K - 1)j} \frac{C_\chi^2}{2\pi} \int_{-\infty}^\infty \left\lvert (2^j \omega)^K \hat f(2^j\omega) \right\rvert^2 d\omega \\
&= 2^{-2Kj} \frac{C_\chi^2}{2\pi} \int_{-\infty}^\infty \left\lvert \omega^K \hat f(\omega) \right\rvert^2 d\omega
= 2^{-2Kj} C_\chi^2 \lVert f^{(K)} \rVert^2,
\end{align*}
which is finite since $f\in H^K(\RR)$.
This shows that $F \in L^2(\RR/2\pi\ZZ)$.
By Parseval's theorem we can thus bound~\cref{eq:series} by
\begin{align*}
  \sum_{k \in \ZZ} \left\lvert\braket{\chi_{j,k}, f} \right\rvert^2
% = \frac1{2\pi} \int_{-\pi}^\pi |F(\theta)|^2 d\theta
\leq 2^{-2Kj} C_\chi^2 \lVert f^{(K)} \rVert^2
\leq 2^{-2Kj} C_\chi^2 \lVert f^{(K)} \rVert^2
\end{align*}
as desired.

The proof for $f\in H^K(\SS)$ proceeds similarly.
First note that $\widehat{g^{\per}}(m) = \hat g(2\pi m)$ if we periodize a function~$g\in L^2(\RR)$ by~$g^{\per}(x) := \sum_{n\in\ZZ} g(x+n)$, so
\begin{align}\label{eq:sum}
  \sum_{k=1}^{2^j} \left\lvert\braket{\chi^{\per}_{j,k}, f} \right\rvert^2
  = \sum_{k=1}^{2^j} \left\lvert\braket{ \widehat{\chi^{\per}_{j,k}}, \widehat{f}} \right\rvert^2
  = \sum_{k=1}^{2^j} \left\lvert \sum_{m\in\ZZ} 2^{-j/2}e^{i2\pi m2^{-j}k} \overline{\hat\chi(2^{-j}2\pi m)} \hat f(m) \right\rvert^2
  % = \sum_{k=1}^{2^j} \left\lvert 2^{-j} \sum_{l=1}^{2^j} e^{i2\pi \frac{k l}{2^{-j}}} \sum_{m\in\ZZ} 2^{j/2} \overline{\hat\chi(2^{-j}2\pi (m 2^j + l) )} \hat f(m 2^j + l) \right\rvert^2
  % = \sum_{k=1}^{2^j} \lvert \check v_k\rvert^2
  % = 2^{-j} \sum_{l=1}^{2^j} \lvert v_l\rvert^2
\end{align}
which we recognize as squared norm of the inverse discrete Fourier transform of a vector~$v$ with~$2^j$ components
\begin{align*}
  v_l := 2^{j/2} \sum_{m\in\ZZ} \overline{\hat\chi(2\pi m + 2\pi 2^{-j} l)} \hat f(2^j m + l),
\end{align*}
where it is useful to take~$l\in\{-2^{j-1}+1,\dots,2^{j-1}\}$.
To see that the components of this vector are well-defined and obtain a quantitative bound, we estimate
\begin{align*}
  \lvert v_l\rvert^2
&= 2^j \left\lvert \sum_{m\in\ZZ} \overline{\hat\chi(2\pi m + 2\pi 2^{-j} l)} \hat f(2^j m + l) \right\rvert^2 \\
&\leq 2^j \sum_{m\in\ZZ} \left\lvert \frac {\hat\chi(2\pi m + 2\pi 2^{-j} l)}{(2\pi m + 2\pi 2^{-j} l)^K} \right\rvert^2 \sum_{m\in\ZZ} \left\lvert (2\pi m + 2\pi 2^{-j} l)^K \hat f(2^j m + l) \right\rvert^2 \\
&= 2^{-(2K-1)j} \sum_{m\in\ZZ} \left\lvert \frac {\hat\chi(2\pi m + 2\pi 2^{-j} l)}{(2\pi m + 2\pi 2^{-j} l)^K} \right\rvert^2 \sum_{m\in\ZZ} \left\lvert (2\pi(2^j m + l))^K \hat f(2^j m + l) \right\rvert^2.
\end{align*}
Since~$\lvert2\pi 2^{-j} l\rvert\leq\pi$, we can upper-bound the left-hand side series precisely as in \cref{eq:def C_chi},
% \begin{align*}
%   \sum_{m\in\ZZ} \left\lvert \frac {\hat\chi(2\pi m + 2\pi 2^{-j} l)}{(2\pi m + 2\pi 2^{-j} l)^K} \right\rvert^2
% = \left\lvert \frac {\hat\chi(2\pi 2^{-j} l)}{(2\pi 2^{-j} l)^K} \right\rvert^2 + \sum_{m\neq0} \left\lvert \frac {\hat\chi(2\pi m + 2\pi 2^{-j} l)}{(2\pi m + 2\pi 2^{-j} l)^K} \right\rvert^2
% \leq C^2 + \frac{\lVert\hat\chi\rVert_\infty^2}{1} \sum_{m=1}^\infty \frac {2}{(\pi m)^{2K}} = \ldots
% \end{align*}
\begin{align*}
  \lvert v_l\rvert^2
\leq 2^{-(2K-1)j} C_\chi^2 \sum_{m\in\ZZ} \left\lvert (2\pi(2^j m + l))^K \hat f(2^j m + l) \right\rvert^2,
\end{align*}
and obtain
\begin{align*}
  \lVert v\rVert_2^2
% \leq 2^{-(2K-1)j} C_\chi^2 \sum_{l, m} \left\lvert (2\pi(2^j m + l))^K \hat f(2^j m + l) \right\rvert^2
\leq 2^{-(2K-1)j} C_\chi^2 \sum_{n\in\ZZ} \left\lvert (2\pi n)^K \hat f(n) \right\rvert^2
% = 2^{-(2K-1)j} C_\chi^2 \sum_{n\in\ZZ} \left\lvert \widehat{f'}(n) \right\rvert^2
= 2^{-(2K-1)j} C_\chi^2 \lVert f^{(K)} \rVert^2,
\end{align*}
which is finite since~$f\in H^K(\SS)$.
As before we conclude by using the Plancherel formula in \cref{eq:sum} and plugging in the upper bound.
\begin{align*}
  \sum_{k=1}^{2^j} \left\lvert\braket{\chi^{\per}_{j,k}, f} \right\rvert^2
% = \sum_{k=1}^{2^j} \lvert \check v_k\rvert^2
= 2^{-j} \sum_{k=1}^{2^j} \lvert v_k\rvert^2
\leq 2^{-2Kj} C_\chi^2 \lVert f^{(K)} \rVert^2,
\end{align*}
which concludes the proof.
\end{proof}

We next use \cref{lem:technical} to prove \cref{lem:UV} and \cref{lem:sampling error}, which are wavelet approximation results for sufficiently smooth functions.

\begin{proof}[Proof of \cref{lem:UV}]
For~$f\in H^K(\RR)$ and~$j\in\ZZ$, we have
\begin{align*}
  \lVert P_j f - f \rVert^2 = \sum_{l > j} \sum_{k \in \ZZ} \lvert \braket{\psi_{l,k},f}\rvert^2.
\end{align*}
because the wavelets form an orthonormal basis.
We would like to bound the inner series by using \cref{lem:technical}.
For this, note that since~$\hat{g}_s$ is a trigonometric polynomial with a zero of order~$K$ at~$\theta = \pi$, there exists a constant~$C$ such that
\begin{align}\label{eq:continuity bound g_w}
  \frac1{\sqrt2}\lvert \hat{g}_w(\theta) \rvert = \frac1{\sqrt2}\lvert \hat{g}_s(\theta + \pi) \rvert \leq C \lvert \theta \rvert^K.
\end{align}
Using \cref{eq:relation_filter_wavelet_fourier} and~$\lVert\hat\phi\rVert_\infty=1$, it follows that
\begin{align*}
  \lvert \hat\psi(\omega) \rvert = \lvert \frac{1}{\sqrt{2}}\hat{g}_w(\frac{\omega}{2}) \hat{\phi}(\frac{\omega}{2}) \rvert \leq \frac{C}{2^K} \lvert \omega \rvert^K.
\end{align*}
Since moreover $\lVert\hat\psi\rVert_\infty=1$, we can invoke \cref{lem:technical} with~$\chi=\psi$ and obtain that
\begin{align*}
  \lVert P_j f - f \rVert^2
\leq \sum_{l > j} 2^{-2Kl} C_{\UV}^2 \lVert f^{(K)} \rVert^2
\leq 2^{-2Kj} C_{\UV}^2 \lVert f^{(K)} \rVert^2,
\end{align*}
where~$C_{\UV}^2 = C^2/4^K + 1/3 \leq C^2 + 1/3$.

In the same way we find that, for any~$f\in H^K(\SS)$ and~$j\geq0$,
\begin{align*}
  \lVert P^{\per}_j f - f \rVert^2
= \sum_{l > j} \sum_{k=1}^{2^l} \lvert \braket{\psi^{\per}_{l,k},f}\rvert^2
\leq 2^{-2Kj} C_{\UV}^2 \lVert f^{(K)} \rVert^2,
\end{align*}
again by \cref{lem:technical}.

For the last assertion, we use \cref{lem:filter fourier bounds} to see that, for~$K=1$, \cref{eq:continuity bound g_w} always holds with~$C=M(M+1)/2$, hence we have~$C_{\UV} \leq 2M^2$.
\end{proof}

\begin{proof}[Proof of \cref{lem:sampling error}]
The trigonometric polynomial~$\hat g_s$ satisfies~$\hat g_s(0)=\sqrt2$, so there is a constant~$C>0$ such that
\begin{align}\label{eq:continuity bound g_s}
  \lvert\frac1{\sqrt2} \hat g_s(\theta) - 1\rvert \leq C \lvert \theta \rvert
\end{align}
for $\theta\in[-\pi,\pi]$. % (and hence for all $\theta\in\RR$.
Using the infinite product formula~\eqref{eq:inf_product}, it follows that, for all~$|\omega|\leq\pi$,
\begin{align}\label{eq:phi hat linear bound}
  \lvert \hat\phi(\omega) - 1 \rvert
&= \lvert \prod_{k=1}^\infty \frac1{\sqrt2} \hat g_s(2^{-k}\omega) - 1 \rvert
\leq \sum_{k=1}^\infty \lvert \frac1{\sqrt2} \hat g_s(2^{-k}\omega) - 1 \rvert
\leq \sum_{k=1}^\infty \frac{C}{\sqrt2} 2^{-k} \lvert \omega \rvert
= \frac{C}{\sqrt2} \lvert\omega\rvert
% \leq C |\omega|
\end{align}
using a telescoping sum and the fact that $\lvert\hat g_s\rvert\leq\sqrt2$ (in fact, this holds for all $\omega\in\RR$, but we will not need this).
Now recall from Sobolev embedding theory that $\hat f \in L^1(\RR)$ for any~$f\in H^1(\RR)$.
Thus, the continuous representative of~$f$ can be computed by the inverse Fourier transform, i.e.,
\begin{align*}
  f(x) = \frac1{2\pi} \int_{-\infty}^\infty \hat f(\omega) e^{i\omega x} d\omega
\end{align*}
for all $x\in\RR$.
As a consequence,
\begin{align*}
  \lVert \alpha_j f - f_j \rVert^2
= \sum_{k\in\ZZ} \left\lvert\braket{\phi_{j,k}, f} - 2^{-j/2} f(2^{-j} k)\right\rvert^2
= \sum_{k\in\ZZ} \left\lvert\braket{\chi_{j,k}, f}\right\rvert^2
\end{align*}
where~$\chi := \phi - \delta_0$.
Now, $\hat\chi = \hat\phi - \one$, hence $\lVert\hat\chi\rVert_\infty \leq 2$.
Together with the bound in \cref{eq:phi hat linear bound} we obtain from \cref{lem:technical} that
\begin{align*}
  \lVert \alpha_j f - f_j \rVert^2 \leq 2^{-2j} C_\phi \lVert f' \rVert^2,
\end{align*}
where $C_\phi := C^2 + \frac43$.
The proof for $H^1(\SS)$ proceeds completely analogously.
Finally, \cref{lem:filter fourier bounds} shows that if the scaling filter is supported in $\{0,\dots,M-1\}$ then \cref{eq:continuity bound g_s,eq:phi hat linear bound} always hold with~$C=M^2/2$.
Thus, $C_\phi \leq 2 M^2$.
\end{proof}

Finally, we prove \cref{lem:IR}, which is an approximation result for compactly supported functions.

\begin{proof}[Proof of \cref{lem:IR}]
Let us denote by~$S$ the support of~$f$.
Since the scaling functions for fixed~$j$ form an orthonormal basis of~$V_j$, and using Cauchy-Schwarz, we find that
\begin{align*}
  \lVert P_j f\rVert^2
= \sum_{k\in\ZZ} \lvert\braket{\phi_{j,k}, f}\rvert^2
\leq %\lVert f\rVert^2 \sum_{k\in\ZZ} \lVert \one_S \phi_{j,k}\rVert^2 =
\lVert f\rVert^2 \sum_{k\in\ZZ} \int_S \lvert\phi_{j,k}(x) \rvert^2 \, dx
= \lVert f\rVert^2 \int_{2^j S} \sum_{k\in\ZZ} \lvert\phi(y - k)\rvert^2 \, dy.
\end{align*}
This allows us to conclude that
\begin{align*}
  \lVert P_j f\rVert^2 \leq \lVert f\rVert^2 2^j C_{\IR}^2 D(f),
\end{align*}
which confirms the claim.
If~$\phi$ is bounded and supported on an interval of width~$M$, we can bound $\sum_{k\in\ZZ} \lvert\phi(y - k)\rvert^2 \leq M \lVert\phi\rVert_\infty^2$.
\end{proof}

\end{appendix}
\phantomsection
\addcontentsline{toc}{section}{References}
\bibliographystyle{hunsrtnat}
\bibliography{library}

\end{document}